\documentclass[english]{amsart}
\usepackage {amsmath, amssymb, amsfonts,amsthm,enumerate,color}
\usepackage{times}
\usepackage{babel}
\usepackage{graphicx, color}
 \usepackage[colorlinks=true, pdfstartview=FitV, linkcolor=blue, citecolor=blue, urlcolor=blue]{hyperref}
\usepackage{pdfsync}

\makeatletter

\newtheorem{thm}{Theorem}[section]
\newtheorem{cor}[thm]{Corollary}
\newtheorem{lem}[thm]{Lemma}
\newtheorem{prop}[thm]{Proposition}

\newtheorem{defn}[thm]{Definition}
\newtheorem{prerem}[thm]{Remark}
\newenvironment{remark}%
  {\begin{prerem}\upshape}{\end{prerem}}

\newcommand{\thmref}[1]{Theorem~\ref{#1}}
\newcommand{\defnref}[1]{definition~\ref{#1}}

\newcommand{\secref}[1]{\S\ref{#1}}
\newcommand{\lemref}[1]{Lemma~\ref{#1}}
\newcommand{\eqnref}[1]{~(\ref{#1})}

\newcommand{\germ}{\mathfrak}


\theoremstyle{rem}
\numberwithin{equation}{section}

\newtheorem{preexample}[thm]{Example}
\newenvironment{example}%
  {\begin{preexample}\upshape}{\end{preexample}}

\newtheorem{notation}[thm]{Notation}

\numberwithin{thm}{section}
\newtheorem{lem/defn}[thm]{Lemma/Definition}
\newtheorem{preex/defn}[thm]{Example/Definition}
\newenvironment{ex/defn}%
  {\begin{preex/defn}\upshape}{\end{preex/defn}}

\numberwithin{equation}{section}

\newcommand{\ten}{\otimes}
\newcommand{\del}{\bigtriangleup}

\newcommand{\abs}[1]{\lvert#1\rvert}

\DeclareMathOperator{\End}{End}

\date{}

\begin{document}

\title[$N$-point locality for vertex operators]{$N$-point locality for vertex operators: normal ordered products, operator product expansions, twisted vertex algebras}
\author{Iana I. Anguelova, Ben Cox, Elizabeth Jurisich}


\date{\today}

\maketitle

\tableofcontents

\begin{abstract}
In this paper we study   fields satisfying $N$-point locality and their properties. We obtain residue formulae for $N$-point local fields in terms of derivatives of delta functions and Bell polynomials. We introduce the notion of the space of descendants of $N$-point local fields which includes normal ordered products  and coefficients of operator product expansions.  We show that examples of $N$-point local fields include the vertex operators generating the boson-fermion correspondences of type B, C and D-A. We apply the normal ordered products of  these vertex operators to the setting of the representation theory of the double-infinite rank Lie algebras $b_{\infty}, c_{\infty}, d_{\infty}$.
Finally, we show that  the field theory generated by  $N$-point local fields and their descendants  has a structure of a twisted vertex algebra.
\end{abstract}

\section{Introduction}
\label{sec:intro}
Vertex operators were introduced  in
string theory and now play an important role in many areas such as quantum field theory,
integrable models, statistical physics, representation theory, random matrix theory, and many others.
There are different vertex algebra  theories, each  designed to describe different sets of examples of collections of fields. The best known is the theory of super vertex algebras (see for instance  \cite{BorcVA}, \cite{FLM}, \cite{FHL}, \cite{Kac}, \cite{LiLep}, \cite{FZvi}), which  axiomatizes the properties of some, simplest, systems of
vertex operators. Locality is a property that plays crucial importance for super vertex algebras (\cite{LiLocality}) and  the axioms of super vertex algebras are often given in terms of locality (see \cite{Kac}, \cite{FZvi}). On the other hand, there are field theories which do not satisfy the usual locality property, but rather a generalization. Examples of these are the generalized vertex algebras (\cite{DongLepVA}, \cite{KacGenVA}, $\Gamma$-vertex algebras (\cite{Li2}), deformed chiral algebras (\cite{FR}), quantum vertex algebras (see e.g. \cite{EK}, \cite{BorcQVA}, \cite{Li1}, \cite{LiQuantum3}, \cite{AB}).

We consider $\lambda \in \mathbb{C}^*$ to be  a point of locality for two fields $a(z)$, $b(w)$ if $a(z)b(w)$ is singular at $z=\lambda w$ (for a precise definition  see Section \ref{section:NormalOrdProd}).  The usual vertex algebra locality is then just locality at the single point $\lambda =1$.
In this paper we study the following field theory problems: First, if we start with fields which are local at several, but finite number of points, then what are the properties that these fields and their descendants satisfy? Second, what is the algebraic structure that the system of descendants of such $N$-point local fields, in its entirety,  satisfies?

In Section \ref{sec:locality} we prove the basic property of  $N$-point local distributions: that they can be expressed in terms of delta functions at the points of locality and their derivatives (Theorem \ref{mainresult}). A question in  any field theory is: if we start with some collection of (generating) fields, which fields do we consider to be their descendants? In this paper we take an approach motivated by quantum field theory  and representation theory and (as was the case of usual 1-point locality)  we  consider derivatives of fields, Operator Product Expansion coefficients and normal ordered products of fields to be among the descendant fields (for precise definitions see  Section  \ref{section:NormalOrdProd}, and in particular Definition  \ref{defn:fielddesc}). For $N$-point locality  it is also natural to consider substitutions at the points of locality to be among the descendant fields, i.e., if $a(z)$ is a field and $\lambda \in \mathbb{C}^*$ is  a point of locality, then $a(\lambda w)$ is a descendant field (this of course trivially holds in the case of usual 1-point locality at $\lambda =1$). In Section \ref{section:NormalOrdProd} we study  the descendant fields and their properties. We show that the Operator Product Expansion (OPE) formula holds and we provide  residue formulas for the OPE coefficients. We also prove properties for the normal order products of fields, such as the Taylor expansion property and the residue formulas. For two fields $a(z)$, $b(w)$ we give a general definition of products of fields $a(z)_{(j, n)}b(w)$, where for $n\geq 0$ these new products coincide with   the OPE coefficients and for   $n< 0$ they coincide with  the normal ordered products (see Definition \ref{defn:products-of-fields}). This seemingly unjustified  unification of these two types of descendants  is in fact explained later by Lemma \ref{lem:AnalContFormProd}, which shows that both the OPE coefficients and the normal ordered products are just different residues of the same analytic continuation, see Lemmas \ref{lem:AnalContToLocal} and \ref{lem:localityToanalcont}.) In  Section  \ref{section:NormalOrdProd} we also  prove a   generalization of Dong's Lemma: if we start with fields that are $N$-point local, then the products of fields are also $N$-point local, thus the entire system of descendants will consist of $N$-point local fields. We finish  Section  \ref{section:NormalOrdProd} by proving properties relating the OPE expansions of normal ordered products of fields.

Among earlier studies of fields that satisfy properties closely related to $N$-point locality as defined here, our work is most closely related to that of  \cite{Li2}.
Our definition of products of fields (and thus descendant fields) is quite different, and as a result Li's theory of $\Gamma$-vertex algebras does not generally coincide with our twisted vertex algebras. See the Appendix \secref{section:appendix} for a detailed comparison between the two constructions.
Our definition is motivated from quantum field theory, and we specifically require that the OPE coefficients and the normal ordered products be elements in our twisted vertex algebras. We provide examples that demonstrate certain normal ordered products cannot be elements of  $\Gamma$-vertex algebra (see the Appendix \secref{section:appendix}).
In general, for a finite cyclic group $\Gamma$,  $\Gamma$-vertex algebras are described by  smaller collections of descendant fields, and  are in fact  subsets of the systems of fields  that we consider.

In Section \ref{section: examples} we detail three examples of $N$-point local fields, and calculate examples of their descendants fields. These three examples, although in some sense the simplest possible, are particularly important due to their connection to both representation theory and integrable systems. In Section \ref{section:reptheory} we show that the normal ordered products of the fields from  Section \ref{section: examples} produce representations of the double-infinite rank Lie algebras $b_{\infty}$, $c_{\infty}$ and $d_{\infty}$. Although these representations have been known earlier (see \cite{DJKM-4}, \cite{YouBKP} for $b_{\infty}$, \cite{DJKM6} for $c_{\infty}$,  \cite{WangKac}, \cite{Wang} for  $d_{\infty}$), the proofs we present are written in terms of normal ordered products and generating series, which is new for the cases of $b_{\infty}$ and  $c_{\infty}$. The case of $d_{\infty}$ is interesting because even though the operator product expansion is $1$-point local with point of locality $1$, it is necessary to consider $N$-points of locality in order to obtain the bozon-fermion correspondence in this case (which was done in \cite{AngTVA}).

Finally, in Section \ref{sectio:tva} we state the definition of twisted vertex algebra  \cite{AngTVA}, and show that it can be re-formulated in terms of $N$-point locality. This shows that twisted vertex algebras can be considered a generalization of super vertex algebras. We finish with  results establishing the strong generation theorem for a twisted vertex algebra and the existence of analytic continuations of arbitrary products of fields.

\ \\
\section{Notation and preliminary results}
\label{sec:notation}

In this section we summarize the notation that will be used throughout the paper.
Let $U$ be an associative algebra with unit. We denote by $U[\!]z^{\pm 1}]\!]$ the doubly-infinite series in the formal variable $z$:
\[
U[\!]z^{\pm 1}]\!]=\{s(z) \,|\,  s(z)=\sum_{m\in \mathbb{Z}}s_{m}z^m, \ \ s_{m}\in U\}.
\]
An element of $U[\!]z^{\pm 1}]\!]$ is called a {\bf formal distribution}.
Similarly, $U[\!]z]\!]$ denotes  the series in the formal variable $z$ with only nonnegative powers in $z$:
\[
U[\!]z]\!]=\{s(z) \,|\,  s(z)=\sum_{m\ge 0}s_{m}z^m, \ \ s_{m}\in U\}.
\]
Let $s(z)\in U[\!]z^{\pm 1}]\!]$. The coefficient $s_{-1}$ is "the formal residue $\text{Res}_z s(z)$": $
\text{Res}_z s(z) :=s_{-1}.$
We denote by $U[\!]z^{\pm 1},w^{\pm 1}]\!]$ the doubly-infinite series in variables $z$ and $w$:
\[
U[\!]z^{\pm 1}, w^{\pm 1}]\!]=\{a(z, w) \,|\,  a(z, w)=\sum_{m, n\in \mathbb{Z}}a_{m, n}z^m w^n, \ \ a_{m, n}\in U\}.
\]
Similarly, we will use the notations  $U[\!]z, w]\!]$, $U[\!]z^{\pm 1}, w]\!]$, $U[\!]z, w^{\pm 1}]\!]$ where for instance  $U[\!]z, w]\!]$ denotes  the infinite series in variables $z$ and $w$ with only nonnegative powers in both $z$ and $w$:
\[
U[\!]z, w]\!]=\{a(z, w) \,|\,  a(z, w)=\sum_{m, n\ge 0}a_{m, n}z^m w^n, \ \ a_{m, n}\in U\}.
\]
\begin{remark}
For any $a(z, w)\in U[\!]z^{\pm 1},w^{\pm 1}]\!]$ the formal residue $\text{Res}_z a(z, w)$ is a well-defined series in $U[\!]w^{\pm 1}]\!]$:
\[
\text{Res}_z a(z, w) =\sum_{n\in \mathbb{Z}}a_{-1, n} w^n \in U[\!]w^{\pm 1}]\!]
\]
\end{remark}

A very important example of a doubly-infinite series is given by the formal delta-function at $z=w$ (recall e.g., \cite{LiLocality}, \cite{Kac}, \cite{LiLep}):
\begin{defn} The formal delta-function at $z=w$ is given  by:
$$
\delta(z, w):=\sum_{n\in\mathbb Z}z^nw^{-n-1}.
$$
\end{defn}
The formal delta-function at $z=w$ is an element of $U[\!]z^{\pm 1},w^{\pm 1}]\!]$ for any $U$ --an associative algebra with unit, as each of the coefficients of $\delta(z, w)$ is equal to the unit of $U$. To emphasize its main property (see the first part of Lemma \ref{fact:firstpropdelta}), and in keeping with the a common usage in physics we will denote this formal delta function by $\delta(z- w)$ (even though this is something of an abuse of notation.)

Let $\lambda\in \mathbb{C}$ be a fixed complex number, $\lambda\neq 0$.
In many cases (e.g., when $U=End(V)$ for a complex vector space $V$) we can consider $U$ to contain a copy of $\mathbb{C}$. For the remainder of this paper, we assume $\mathbb{C}\subset U$.
 The formal delta-function at $z=\lambda w$ is given by:
$$
\delta(z-\lambda w):=\delta(z,\lambda w)= \sum_{n\in\mathbb Z}\lambda ^{-n-1} z^nw^{-n-1}.
$$

Again, by abuse of notation we write $\delta(z-\lambda w)$ even though it depends on two formal variables $z$ and $w$, and a parameter $\lambda$.

 As is well known for the formal delta-function at $z=w$ (see e.g., \cite{LiLocality}, \cite{Kac}, \cite{LiLep}) we have the following properties for the formal delta-function at $z=\lambda w$:
\begin{prop} (\cite{LiLocality}, \cite{Kac} Prop. 2.1.)
\label{fact:firstpropdelta}
For $n\in \mathbb{Z}, n\ge 0$, $\lambda\in\mathbb C\backslash\{0\}$,
\begin{enumerate}
\item  For any $f(z)\in U[\!]z^{\pm 1}]\!]$, one has $f(z)\delta(z-\lambda w)=f(\lambda w)\delta(z-\lambda  w)$ and in particular $(z-\lambda w)\delta(z-\lambda  w)=0$.
\item For any $f(z)\in U[\!]z^{\pm 1}]\!]$ we have $\text{Res}_z f(z) \delta(z-\lambda w) =f(\lambda w)$.
\item $\delta(z-\lambda w)=\delta(z, \lambda w) =\lambda^{-1} \delta(w, \lambda^{-1} z)=\lambda^{-1} \delta(w-\lambda^{-1} z)$.
\item $(z-\lambda w)\partial^{(n)} _{\lambda w}\delta(z-\lambda  w) =\partial^{(n-1)} _{\lambda w}\delta(z-\lambda   w)$ for $n\geq 1$.\label{fact:n,n-1}\label{absorption}
\item $(z-\lambda w)^{n+1}\partial^n _z\delta(z-\lambda w)=0$.
\item $\partial _z\delta(z-\lambda w) =-\lambda^{-1} \partial_w  \delta(w-\lambda^{-1} z)$.\label{fact:n+1n}
\end{enumerate}
\end{prop}
%
\begin{notation}{\bf Divided powers}
For any $a\in A$, where $A$ is an  associative  $\mathbb{C}$ algebra, denote $a^{(n)}:=\frac{a^n}{n!}$.
\end{notation}
In this work we encounter formal delta functions at $z=\lambda_i w$, for $\lambda_i\in \mathbb{C}, \ \lambda_i\neq 0$,  $\lambda_i\neq \lambda_j$, $1\leq i,j\leq N$, which satisfy  properties extending the properties above:
\begin{lem}[Factoring properties]\label{factorfacts}
For $j, n\in \mathbb{Z}, j, n\ge 0$,
\begin{enumerate}
\item $(z-\lambda _i w)\delta(z-\lambda _j w)=(\lambda_j-\lambda_i)w \delta(z-\lambda _j w)$\label{factoring2.5}
\label{property1}
\item
 $$(z-\lambda _i w)\partial^{(n)} _{\lambda_jw}\delta(z-\lambda _j w)
=\partial^{(n-1)} _{\lambda_jw}\delta(z-\lambda _j w) +(\lambda_j-\lambda_i)\cdot w \partial^{(n)} _{\lambda_jw}\delta(z-\lambda _jw).$$
\label{property2}
 \item \label{zlambdawfactor}
 \begin{align*}
\partial^{(n)}_{\lambda_jw}\delta(z-\lambda _jw)
&=\frac{(z-\lambda _i w)w^{-n-1}}{(\lambda_j-\lambda_i)^{n+1}}\left(\sum_{k=0}^n(-1)^{n-k}(\lambda_j-\lambda_i)^kw^k\partial_{\lambda_jw}^{(k)}\delta(z-\lambda_jw)\right).
\end{align*}
\end{enumerate}
\end{lem}
\begin{proof}
Part \eqnref{property2} follows immediately from Proposition \ref{fact:firstpropdelta}, \eqnref{fact:n,n-1}:
\begin{align*}
(z-\lambda _i w)\partial^{(n)} _{\lambda_jw}\delta(z-\lambda _j w)&=(z-\lambda _j w)\partial^{(n)} _{\lambda_jw}\delta(z-\lambda _j w) +(\lambda _j -\lambda _i)\cdot w\partial^{(n)} _{\lambda_jw}\delta(z-\lambda _j w)\\
&=\partial^{(n-1)} _{\lambda_jw}\delta(z-\lambda _j w) +(\lambda_j-\lambda_i)\cdot w \partial^{(n)} _{\lambda_jw}\delta(z-\lambda _jw).
\end{align*}
From\eqnref{property2} we get
\begin{align*}
\partial^{(1)} _{\lambda_jw}\delta(z-\lambda _jw)&=\frac{(z-\lambda _i w)}{(\lambda_j-\lambda_i)w}\partial^{(1)} _{\lambda_jw}\delta(z-\lambda _j w)
-\frac{1}{(\lambda_j-\lambda_i)w}\delta(z-\lambda _j w)   \\
&=\frac{(z-\lambda _i w)}{(\lambda_j-\lambda_i)w}\partial^{(1)} _{\lambda_jw}\delta(z-\lambda _j w)
-\frac{(z-\lambda _i w)}{(\lambda_j-\lambda_i)^2w^2}\delta(z-\lambda _j w)   \\
&=\frac{(z-\lambda _i w)}{(\lambda_j-\lambda_i)^2w^2}\left((\lambda_j-\lambda _i)w\partial^{(1)} _{\lambda_jw}\delta(z-\lambda _j w)-\delta(z-\lambda _j w) \right).
\end{align*}
%
So we suppose that
\begin{align*}
\partial^{(n)} _{\lambda_jw}\delta(z-\lambda _jw)
&=\frac{(z-\lambda _i w)}{(\lambda_j-\lambda_i)^{n+1}w^{n+1}}\left(\sum_{k=0}^n(-1)^{n-k}(\lambda_j-\lambda_i)^kw^k\partial_w^{(k)}\delta(z-\lambda_jw)\right).
\end{align*}
Then
\begin{align*}
\partial^{(n+1)} _{\lambda_jw}\delta(z-\lambda _jw)
&=\frac{(z-\lambda_iw)}{(\lambda_j-\lambda_i)\cdot w}\partial^{(n+1)} _{\lambda_jw}\delta(z-\lambda _jw)-\frac{1}{(\lambda_j-\lambda_i)w}\partial^{(n)} _{\lambda_jw}\delta(z-\lambda _j w)  \\
&=\frac{(z-\lambda _i w)}{(\lambda_j-\lambda_i)^{n+2}w^{n+2}}\left(\sum_{k=0}^{n+1}(-1)^{n+1-k}(\lambda_j-\lambda_i)^kw^k\partial_{\lambda_jw}^{(k)}\delta(z-\lambda_jw)\right).
\end{align*}
\end{proof}

For a rational function $f(z,w)$ we denote by $i_{z,w}f(z,w)$
the expansion of $f(z,w)$ in the region $\abs{z}\gg \abs{w}$ (the region in the complex plane outside of all  the points $z=\lambda_i w$, $\lambda_i\in \mathbb{C}, \ 1\leq i\leq n$), and correspondingly for
$i_{w,z}f(z,w)$.

\begin{example}
\begin{align}
i_{z,w}\frac{1}{z-\lambda w}&=\sum _{n\ge 0}\lambda ^{n} z^{-n-1}w^{n}, \\
i_{w, z}\frac{1}{z-\lambda w}&=-\sum _{n\ge 0}\lambda ^{-n-1} z^nw^{-n-1}=-\sum _{n< 0}\lambda ^{n} z^{-n-1}w^{n}.
\end{align}
Moreover
\begin{align}
\delta(z-\lambda  w)&=i_{z, w}\frac{1}{z-\lambda w} -i_{w, z}\frac{1}{z-\lambda w}\label{deltaExpanded1}\\
\partial_{\lambda w}^{(l)}\delta(z-\lambda  w)&=i_{z, w}\frac{1}{(z-\lambda w)^{l+1}} -i_{w, z}\frac{1}{(z-\lambda w)^{l+1}}\label{deltaExpanded2}
\end{align}

\end{example}

For any $a(z)=\sum_{m\ge 0}a_{m}z^m \in U[\!]z^{\pm 1}]\!]$ we use the following notation:
\begin{equation}\label{eqn:plusonfield}
a(z)_-:=\sum_{n< 0}a_nz^{n} \in U[\!]z^{-1}]\!],\quad a(z)_+:=\sum_{n\geq 0}a_nz^{n}\in U[\!]z^{ 1}]\!].
\end{equation}
Similarly for any $a(z, w)=\sum_{m, n\in \mathbb{Z}}a_{m, n}z^m w^n\in U[\!]z^{\pm 1},w^{\pm 1}]\!]$ we use the following notation
\begin{align}
a(z, w)_{z,-}:&=\sum_{m< 0, n\in \mathbb{Z}}a_{m, n}z^{m}w^n \in U[\!]z^{-1}, w^{\pm 1}]\!], \label{plusOnField1}\\
a(z, w)_{z, +}:&=\sum_{m\geq 0, n\in \mathbb{Z}}a_nz^{n}\in U[\!]z, w^{\pm 1}]\!], \label{plusOnField2}\\
a(z, w)_{w,-}:&=\sum_{m\in \mathbb{Z}, n< 0}a_{m, n}z^{m}w^n \in U[\!]z^{\pm 1}, w^{-1}]\!],\label{plusOnField3} \\
a(z, w)_{w, +}:&=\sum_{m\in \mathbb{Z}, n\geq 0}a_nz^{m}w^n\in U[\!]z^{\pm 1}, w]\!].\label{plusOnField4}
\end{align}

Thus from \eqnref{deltaExpanded1} and \eqnref{deltaExpanded2} we see that
\begin{align}
\delta(z-\lambda  w)_{z, -}=i_{z, w}\frac{1}{z-\lambda w}, & \quad \delta(z-\lambda  w)_{z, +}= -i_{w, z}\frac{1}{z-\lambda w},\label{deltaPM1}\\
\big(\partial_{\lambda w}^{(l)}\delta(z-\lambda  w)\big)_{z, -}=i_{z, w}\frac{1}{(z-\lambda w)^{l+1}}, & \quad
\big(\partial_{\lambda w}^{(l)}\delta(z-\lambda  w)\big)_{z, +}= -i_{w, z}\frac{1}{(z-\lambda w)^{l+1}}.\label{deltaPM2}
\end{align}
\begin{lem}{\bf (Cauchy formulas)}
\label{lem:CauchyForm}
For any $a(z)=\sum_{m\ge 0}a_{m}z^m \in U[\!]z^{\pm 1}]\!]$ and $\lambda \in \mathbb{C}$ we have
\begin{align}
\text{Res}_z \left(a(z)i_{z, w}\frac{1}{(z-\lambda w)^{n+1}}\right) & =\partial_{\lambda w}^{(n)} a(\lambda w)_+\\
\text{Res}_z \left(a(z)i_{w, z}\frac{1}{(z-\lambda w)^{n+1}}\right) & =-\partial_{\lambda w}^{(n)} a(\lambda w)_-
\end{align}
\end{lem}
\begin{lem}{\bf (Linear independence property)}
\label{lem:deltaIndep}
Suppose $\lambda_1,\dots,\lambda_N$ are distinct nonzero complex numbers and
\begin{align*}
\sum_{k=1}^N\sum_{l=0}^{n_k-1}c_{kl}(w)\partial_{\lambda_kw}^{(l)}\delta(z-\lambda_kw)=0.
\end{align*}
Then $c_{kl}(w)=0$ for all $1\leq k\leq N$, $0\leq l\leq n_k-1$.
\end{lem}
\begin{proof}  The proof of this is reminiscent of the part of the proof of the Chinese Remainder Theorem.
Fix $1\leq j\leq N$.  We multiply both sides of the equality  by
$$
(z-\lambda_jw)^{n_j-1}\prod_{r=1,r\neq j}\left( z-\lambda_rw\right)^{n_r}
$$
then we obtain by \lemref{fact:firstpropdelta} and \lemref{factorfacts}
\begin{align*}
0&=\sum_{l=0}^{n_j-1}c_{jl}(w)(z-\lambda_jw)^{n_j-1}\prod_{r=1,r\neq j}\left( z-\lambda_rw\right)^{n_r}\partial_{\lambda_jw}^{(l)}\delta(z-\lambda_jw) \\
&=   \prod_{r=1,r\neq j}\left( z-\lambda_rw\right)^{n_r}c_{j,n_j-1}(w)\delta(z-\lambda_jw) \\
&= \prod_{r=1,r\neq j}\left( \lambda_j-\lambda_r\right)^{n_r}w^{n_r}c_{j,n_j-1}(w)\delta(z-\lambda_jw) .
\end{align*}
Taking $\text{Res}_z$ of both sides of the above gives us $c_{j,n_j-1}(w)=0$ since the $\lambda_k$ are distinct.
\end{proof}

\section{$N$-point locality of formal distributions}
\label{sec:locality}
The notion of locality is fundamental to development of conformal field theory and vertex algebra theory, as seen in \cite{LiLocality}, \cite{Kac}, and \cite{FZvi}. Our notion of $N$-point locality is a direct generalization of these.
\begin{defn} \begin{bf}($N$-point local)\end{bf}\label{defn:locality}
Given  $a(z, w)\in U[\!]z^{\pm 1},w^{\pm 1}]\!]$ we say that $a(z,w)$ is {\bf N-point local} at the points $\lambda_i\in \mathbb{C}^*$,  $\lambda_i\neq \lambda_j$, for $1\leq i,j\leq N$
 if there exist  exponents $n_1, n_2, \dots, n_N \in \mathbb{N}^*$ such that
\[
(z-\lambda _1 w)^{n_{1}}(z-\lambda _2 w)^{n_{2}}\cdots (z-\lambda _N w)^{n_{N}} a(z, w) =0.
\]
\end{defn}
We will refer to the distribution as being $N$-point local, and refer to the points $\lambda_i\in \mathbb{C}^*$ as points of locality for $a(z,w)$. We will refer to a formal distribution as satisfying $N$-point locality when we do not want to specify the points $\lambda_i\in \mathbb{C}^*$.

The following theorem is a necessary technical result which will allow us to define the operator product expansion and prove residue formulas for products of fields. This in turn allows us to generalize results from vertex algebra theory to our setting of twisted vertex algebras in the sections that follow. The proof given here has the advantage of being more rudimentary than those appearing  in \cite{Kac}, \cite{FZvi} for example, in the sense that it involves only induction and the factoring properties of delta functions.
\begin{thm}\label{mainresult} A formal distribution
$a(z, w)\in U[\!]z^{\pm 1},w^{\pm 1}]\!]$ satisfies $N$-point locality at a set of points $\lambda_i\in \mathbb{C}^*$ and exponents $n_i$ for  $i = 1\cdots N$ if and only
\begin{equation}
a(z,w)=\sum_{k=1}^N\sum_{l=0}^{n_k-1}c_{kl}(w)\partial_{\lambda_kw}^{(l)}\delta(z-\lambda_kw)
\end{equation}
for some  $c_{kl}(w) \in U[\!]w^{\pm 1}]\!]$. If $N$-point locality holds then the $c_{kl}(w)$ are unique and defined by
\begin{align}
\label{eqn:FormCoeff}
c_{j,n_j-r}(w)&=\text{Res}_z\left( \sum_{i=0}^{r-1}(z-\lambda_jw)^{n_j-r+i}P_{i,j}(w)\prod_{i=1,i\neq j}^N(z-\lambda_iw)^{n_i}a(z,w)\right).
\end{align}
for $1\leq r\leq n_j$;  where
\[
P_{0, j}(w)=\frac{1}{p_{0,j}}, \quad P_{n, j}(w)=\sum_{k=1}^n\frac{(-1)^k k!}{p_{0,j}^{k+1}n!} B_{n,k}(p_{-1,j}, p_{-2,j}, \dots , p_{k-n-1,j}), \ n\geq 1.
\]
and $B_{n,k}(x_1, x_2, \dots, x_{n-k+1})$ are the Bell polynomials in the variables
$$
p_{-i,j}(w)=\partial_z^{(i)}\left(\prod_{r=1,r\neq j}^N(z-\lambda_rw)^{n_r}\right)\Big|_{z=\lambda_jw}.
$$
\end{thm}

In the formula for the coefficients $c_{kl}(w)$ above   $(z-\lambda_k w)^{l}\prod_{i=1,i\neq k}^n(z-\lambda_iw)$  is short  for the binomial expansion \\ \mbox{$i_{z, w}\left((z-\lambda_k w)^{l}\prod_{i=1,i\neq k}^n(z-\lambda_iw)\right)=i_{w, z}\left((z-\lambda_k w)^{l}\prod_{i=1,i\neq k}^n(z-\lambda_iw)\right)$.}

We will present a proof by multi-induction on $n_1, n_2, \dots n_N$.
For the first step, we need to proof the following lemma:
\begin{lem}\label{N_i=1}
Let $a(z, w)\in U[\!]z^{\pm 1},w^{\pm 1}]\!]$ satisfy $N$-point locality for $n_1=n_2=\dots= n_N=1$, i.e.,
\[
(z-\lambda _1 w)(z-\lambda _2 w)\cdots (z-\lambda _N w)a(z, w) =0.
\]
Then
\begin{equation}
a(z,w)=\sum_{k=1}^N c_{k}(w)\delta(z-\lambda_k w), \ \text{where} \ c_{k}(w) \in U[\!]w^{\pm 1}]\!].
\end{equation}
\end{lem}
\begin{proof}
The proof is by induction on $N$. For $N=1$, suppose we have (for brevity denote $\lambda_1$ just by $\lambda$):
\[
(z-\lambda  w)a(z, w)=0,
\]
where $a(z, w)=\sum_{k, l\in \mathbb{Z}}a_{k, l} z^k w^l$, \ $a_{k, l}\in U$.
We have
\begin{equation*}
(z-\lambda  w)a(z, w) =\sum _{k, l}a_{k, l} z^{k+1} w^l -\sum _{k, l}\lambda a_{k, l} z^k w^{l+1}
 =\sum _{k, l}a_{k-1, l} z^{k} w^l -\sum _{k, l}\lambda a_{k, l-1} z^k w^{l}=0.
\end{equation*}
Hence we have the recurrence relation
\[
a_{k-1, l} =\lambda a_{k, l-1}, \ \ \text{for \ any}\ \ k, l\in \mathbb{Z}.
\]
Denote $m: =k+l$, and $A^m_k: =a_{k, m-k}=a_{k, l}$, hence we have
\[
A^{m-1}_{k-1} =A^{m-1}_k, \ \ \text{for \ any}\ \ m, k\in \mathbb{Z}.
\]
Since this is true for any $m \in \mathbb{Z}$, we will suppress the index $m$ and just write
\[
A_{k-1}=\lambda A_k, \ \text{or} \ A_{k}=\lambda^{-1} A_{k-1},
\]
which is a linear recurrence relation for the doubly infinite sequence $A_k, \ k\in \mathbb{Z}$. As is well known, the solution to such a linear recurrence relation is (since the solution to the linear first order characteristic equation is $r=\lambda ^{-1}$):
\[
A_k=A_0\cdot \lambda^{-k}, \  A_0 \ \text{can be chosen arbitrarily}.
\]
Hence, for any $m\in \mathbb{Z}$, we have $A^m_k=A^m_0 \lambda ^{-k}$.
Thus we have
\begin{align*}
a(z, w) & =\sum_{k, l\in \mathbb{Z}}a_{k, l} z^k w^l =\sum _{m, k\in \mathbb{Z}}A^m_k z^k w^{m-k}= \sum _{m, k\in \mathbb{Z}} A^m_0 \lambda ^{-k} z^k w^{m-k} \\
& =\sum _{m \in \mathbb{Z}} A^m_0 \lambda w^{m+1} \sum _{k\in \mathbb{Z}} \lambda ^{-k-1} z^k w^{-k} = \sum _{m \in \mathbb{Z}} A^m_0 \lambda w^{m+1} \delta (z-\lambda w) \\
&=c(w)\cdot  \delta (z-\lambda w).
\end{align*}

Now assume the result holds for $N$, and consider
\[
(z-\lambda _1 w)(z-\lambda _2 w)\cdots (z-\lambda _N w)(z-\lambda _{N+1} w)a(z, w) =0.
\]
Regrouping (associativity applies here)
and applying the induction hypothesis yields
\[
(z-\lambda _{N+1} w)a(z, w)  =\sum_{k=1}^N c_{k}(w)\delta(z-\lambda_k w), \ \text{where} \ c_{k}(w) \in U[\!]w^{\pm 1}]\!], \ k=1, \dots N.
\]
By  \lemref{factorfacts} (1):
\[
(z-\lambda _{N+1} w)a(z, w)  =(z-\lambda_{N+1} w) \sum_{k=1}^N \tilde{c}_{k}(w)\delta(z-\lambda_k w),
\]
where $\tilde{c}_{k}(w)= \frac{1}{(\lambda_{N+1}-\lambda_k)w}c_k (w) \in U[\!]w^{\pm 1}]\!], \ k=1, \dots N$.
Hence we have
\[
(z-\lambda _{N+1} w)\big(a(z, w)-\sum_{k=1}^N \tilde{c}_{k}(w)\delta(z-\lambda_k w)\big)=0
\]
which from the base case gives us
\[
a(z, w)-\sum_{k=1}^N \tilde{c}_{k}(w)\delta(z-\lambda_k w)=\tilde{c}_{N+1}(w)\delta(z-\lambda_{N+1} w),
\]
which proves the induction step and thus the lemma.  \end{proof}
\begin{proof}[Proof of Theorem \ref{mainresult}:]
We will prove the theorem in two parts: first we will prove that $N$-point locality implies
\begin{equation*}
a(z,w)=\sum_{k=1}^N\sum_{l=0}^{n_k-1}c_{kl}(w)\partial_{\lambda_kw}^{(l)}\delta(z-\lambda_kw),
\end{equation*}
where $c_{kl}(w)\in U[\!] w^{\pm 1}]\!]$.
In addition,  we will obtain the particular form of the coefficients $c_{kl}(w)$.
For the first part we will use multi-induction on $n_1, n_2, \dots , n_N$. The base of the induction, $n_1=n_2=\dots n_N=1$,  is Lemma \ref{N_i=1}. Now assume the result is true for some $N$-tuple  $n_1, n_2, \dots, n_m, \dots , n_N$, i.e.,
\[
(z-\lambda _1 w)^{n_{1}}\cdots (z-\lambda _m w)^{n_{k}}\cdots (z-\lambda _N w)^{n_{N}} a(z, w) =0
\]
implies
\begin{equation*}
a(z,w)=\sum_{k=1}^N\sum_{l=0}^{n_k-1}c_{kl}(w)\partial_{\lambda_kw}^{(l)}\delta(z-\lambda_kw),
\end{equation*}
with $c_{kl}(w)\in U[\!] w^{\pm 1}]\!]$.

Suppose now we have for  the tuple $n_1, n_2, \dots, n_m +1, \dots , n_N$
\[
(z-\lambda _1 w)^{n_{1}}\cdots (z-\lambda _k w)^{n_{m}+1}\cdots (z-\lambda _N w)^{n_{N}} a(z, w) =0.
\]
We can regroup (associativity applies here):
\[
(z-\lambda _1 w)^{n_{1}}\cdots (z-\lambda _m w)^{n_{k}}\cdots (z-\lambda _N w)^{n_{N}} \big((z-\lambda_m w)a(z, w)\big) =0.
\]
Hence from the induction hypothesis we have
\begin{equation}
\label{eqn:indhyp}
(z-\lambda_m w)a(z, w)=\sum_{k=1}^N\sum_{l=0}^{n_k-1}c_{kl}(w)\partial_{\lambda_kw}^{(l)}\delta(z-\lambda_kw).
\end{equation}
Now from the factoring properties Lemma \ref{factorfacts} we can factor a $(z-\lambda_m w)$ from each of
$\partial_{\lambda_kw}^{(l)}\delta(z-\lambda_kw)$---with almost no expense for $k\ne m$, and with expense of raising the order of the derivative  $\partial_{\lambda_m w}^{(l)}\delta(z-\lambda_m w)$ for $k=m$:
\begin{align*}
\sum_{k=1}^N\sum_{l=0}^{n_k-1}c_{kl}(w)\partial_{\lambda_kw}^{(l)}\delta(z-\lambda_kw)& =(z-\lambda_m w)\big(\sum_{k\neq m, k=1}^N\sum_{l=0}^{n_k-1}\tilde{c}_{kl}(w)\partial_{\lambda_kw}^{(l)}\delta(z-\lambda_kw)  \\
& \hspace{2.4cm}+\sum_{l=1}^{n_m}\tilde{c}_{ml}(w)\partial_{\lambda_m w}^{(l)}\delta(z-\lambda_m w)\big).
\end{align*}
We can rewrite \eqref{eqn:indhyp} as
\[
(z-\lambda_m w)\Big(a(z, w)-\sum_{k\neq m, k=1}^N\sum_{l=0}^{n_k-1}\tilde{c}_{kl}(w)\partial_{\lambda_kw}^{(l)}\delta(z-\lambda_kw) -\sum_{l=1}^{n_m}\tilde{c}_{ml}(w)\partial_{\lambda_m w}^{(l)}\delta(z-\lambda_m w)\Big)=0
\]
which from Lemma \ref{N_i=1} gives us
\[
a(z, w)-\sum_{k\neq m, k=1}^N\sum_{l=0}^{n_k-1}\tilde{c}_{kl}(w)\partial_{\lambda_kw}^{(l)}\delta(z-\lambda_kw) -\sum_{l=1}^{n_m}\tilde{c}_{ml}(w)\partial_{\lambda_m w}^{(l)}\delta(z-\lambda_m w)=c_{m0}\delta(z-\lambda_m w).
\]
This is precisely the induction step we needed to prove.

Now we proceed to calculating the precise form of the coefficients from \eqref{eqn:FormCoeff}, which we will do  in two steps.


Suppose
$$
a(z,w)=\sum_{k=1}^N\sum_{l=0}^{n_k-1}c_{kl}(w)\partial_{\lambda_kw}^{(l)}\delta(z-\lambda_kw).
$$
First we write
$$
\prod_{i=1,i\neq j}^N(z-\lambda_iw)^{n_i}=p_{0,j}(w)+p_{-1,j}(w)(z-\lambda_jw)+\cdots +p_{-n_j+1,j}(w)(z-\lambda_jw)^{n_j-1}+O(z-\lambda_jw)^{n_j}
$$
where $\displaystyle{p_{-i,j}(w)=
\partial_z^{(i)}\left(\prod_{r=1,r\neq j}^N(z-\lambda_rw)^{n_r}\right)\Big|_{z=\lambda_jw}}$.

Then
\begin{align*}
&\prod_{i=1,i\neq j}^N(z-\lambda_iw)^{n_i}a(z,w)=\sum_{k=1}^N\sum_{l=0}^{n_k-1}c_{kl}(w)\prod_{i=1,i\neq j}^N(z-\lambda_iw)^{n_i}\partial_{\lambda_kw}^{(l)}\delta(z-\lambda_kw)  \\
& =\sum_{l=0}^{n_j-1}c_{jl}(w)\prod_{i=1,i\neq j}^N(z-\lambda_iw)^{n_i}\partial_{\lambda_jw}^{(l)}\delta(z-\lambda_jw)  \\
& =\left(\sum_{l=0}^{n_j-1} p_{-l,j}(w)(z-\lambda_jw)^{l}\right)\cdot\left(\sum_{l=0}^{n_j-1}c_{jl}(w)\partial_{\lambda_jw}^{(l)}\delta(z-\lambda_jw)\right)
\end{align*}


Then we get the linear system
\begin{equation}
\label{eqn:linearsystem}
\sum_{l=r}^{n_j-1}c_{jl}(w)p_{r-l, j}(w) =\text{Res}_z\left( (z-\lambda_jw)^{r}\prod_{i=1,i\neq j}^N(z-\lambda_iw)^{n_i}a(z,w)\right).
\end{equation}
This linear system is upper-triangular, and has a unique solution for the coefficients $c_{jl}$. Our formula \eqnref{eqn:FormCoeff} for the coefficients follows from the following consideration:
Let $P(t)=p_{0, j} +p_{-1, j}t+\dots p_{1-n_j, j}t^{n_j-1}$ be the polynomial in a formal variable $t$, where
$p_{-i, j}=p_{-i, j}(w)$ are the partial derivatives defined above.
Let also $A_r:=A_r(w)=\text{Res}_z\left( (z-\lambda_jw)^{n_j-r-1}\prod_{i=1,i\neq j}^N(z-\lambda_iw)^{n_i}a(z,w)\right)$ and let $A(t)=A_0 +A_1t+\dots +A_{n_j-1}t^{n_j-1}$ be the polynomial in the formal variable $t$. Consider the problem of finding a power series (or polynomial) $Q(t)$ indexed as
\[
Q(t)=q_{ n_j-1} +q_{n_j-2}t +\dots q_{0}t^{n_j-1} +\dots
\]
 satisfying the equation
\begin{equation}
\label{eqn:polynsystem}
P(t)\cdot Q(t) =A(t).
\end{equation}
The coefficients $q_{l}$ for $l\leq n_j-1$ will satisfy the same linear system \eqnref{eqn:linearsystem}, and thus $q_{l}=c_{jl}$ for $l=0, \dots , n_j-1$. On the other hand we can solve  \eqnref{eqn:polynsystem} (and thus \eqnref{eqn:linearsystem}) by just inverting the power series:
\[
Q(t)=A(t)\cdot \frac{1}{P(t)}.
\]
The inversion $1/P(t)$ satisfies Faa di Bruno formula for $f(P(t))$, $f(t)=1/t$:
\[
\frac{1}{P(t)}=\frac{1}{p_{0,j}} +\sum_{n\geq 1}\sum_{k=1}^n\frac{(-1)^k k!}{p_{0,j}^{k+1}n!} B_{n,k}(p_{-1,j}, p_{-2,j}, \dots , p_{k-n-1,j})t^n,
\]
where $B_{n,k}(x_1, x_2, \dots, x_{n-k+1})$ are the Bell polynomials (see e.g. \cite{Bell}, \cite{Comtet}).
We denote the coefficients of this power series by
\[
P_{0, j}(w)=\frac{1}{p_{0,j}}, \quad P_{n, j}(w)=\sum_{k=1}^n\frac{(-1)^k k!}{p_{0,j}^{k+1}n!} B_{n,k}(p_{-1,j}, p_{-2,j}, \dots , p_{k-n-1,j}), \ n\geq 1.
\]
Hence the solution to the power series problem \eqnref{eqn:polynsystem} is given by
\[
Q(t)=A(t)\cdot \frac{1}{P(t)}=A(t)\cdot \sum_{n\geq 0}P_{n, j}(w) t^n;
\]
and thus
\begin{align*}
c_{j, n_j-r}=&q_{n_j-r} =\sum_{i=0}^{r-1}A_iP_{i,j}(w) \\
&\hspace{1.2cm} =\text{Res}_z\left( \sum_{i=0}^{r-1}(z-\lambda_jw)^{n_j-r+i}P_{i,j}(w)\prod_{i=1,i\neq j}^N(z-\lambda_iw)^{n_i}a(z,w)\right).
\end{align*}
Note that since in the denominators of $P_{n, j}(w)$ we have
\[
p_{0,j}= \prod_{i=1,i\neq j}^N(\lambda_j-\lambda_i)^{n_i}w^{n_i}=\text{constant}\cdot w^{\sum_{i=1,i\neq j}^N n_i},
\]
the formula above ensures that the coefficients $c_{j, n_j-r}(w)$ are distributions,  $c_{j, n_j-r}(w)\in U[\!] w^{\pm 1}]\!]$.
This finishes the proof of the  ``if" part of Theorem \ref{mainresult}. The converse  follows directly from the lemmas in Section \ref{sec:notation}.  The uniqueness follows from \lemref{lem:deltaIndep}.
\end{proof}
\begin{remark}\label{remark:CoefZero}
 If a distribution $a(z, w)$ is $N$-point local with exponents $n_1, n_2, \dots, n_N$ as in Definition \ref{defn:locality}, then $a(z, w)$ is also $N$-point local with exponents $m_1, m_2, \dots, m_N$ where $m_i\geq n_i$ for $i=1, \dots, N$. From the uniqueness of the coefficients it follows that  $c_{kl}(w)=0$ for $l\geq n_k$.
\end{remark}
 \ \\

\section{$N$-point local fields, operator product expansions (OPE's), normal ordered products  and their properties}
\label{section:NormalOrdProd}

\begin{defn} \label{defn:parity}
We say that a distribution $a(z)\in U[\!] z^{\pm 1}]\!]$ is {\bf even} and $N$-point self-local at $\lambda_1, \lambda_2, \dots, \lambda_N$  if there exist $n_1, n_2, \dots  , n_N$ such that
\begin{equation}
(z-\lambda _1 w)^{n_{1}}(z-\lambda _2 w)^{n_{2}}\cdots (z-\lambda _N w)^{n_{N}}[a(z),a(w)] =0.
\end{equation}
In this case we set the {\bf parity} $p(a(z))$ of $a(z)$ to be $0$.
\\
We set $\{a, b\}  =ab +ba$.We say that a distribution $a(z)\in U[\!] z^{\pm 1}]\!]$ is $N$-point self-local at $\lambda_1, \lambda_2, \dots, \lambda_N$  and {\bf odd} if there exist $n_1, n_2, \dots , n_N$ such that
\begin{equation}
(z-\lambda _1 w)^{n_{1}}(z-\lambda _2 w)^{n_{2}}\cdots (z-\lambda _N w)^{n_{N}}\{a(z),a(w)\}=0.
\end{equation}
In this case we set the {\bf parity} $p(a(z))$ to be $1$. For brevity we will just write $p(a)$ instead of $p(a(z))$.
\end{defn}
\begin{remark}
In physics the notion of parity of a distribution is tied to its operator product expansion (see \lemref{lem:OPE} below) which follows from its self-locality, and it was the approach we took here as well.
\end{remark}
\begin{remark}
If two elements $a(z), b(z)\in U[\!] z^{\pm 1}]\!]$ commute (one of the elements has parity 0) or  supercommute (both elements have parity 1), then they can be considered to be local at any $\lambda \in \mathbb{C}$. We call such points of locality removable.
\end{remark}
We want to investigate the simplest examples of $N$-point self-local distributions. Of course we start with $N=1$. If the point of locality is at $\lambda =1$, then this case is the usual locality at $z=w$ which is amply studied in the theory of vertex super algebras. On the other hand the cases of other points of locality are not well understood.
\begin{lem}
\label{lem:simlestExamples}
Let  $a(z)\in U[\!] z^{\pm 1}]\!]$ be $1$-point self-local at $\lambda$, of order 1, i.e.,
\begin{align*}
& (z-\lambda w)\left(a(z)a(w) -(-1)^{p(a)^2} a(w)a(z)\right)=0, \  \text{but} \ \
a(z)a(w) -(-1)^{p(a)^2} a(w)a(z)\neq 0.
\end{align*}
Then $\lambda ^2=1$.
\end{lem}
\begin{proof}
From Theorem \ref{mainresult} we have that there exists $c(w)\in U[\!] w^{\pm 1}]\!]$ such that
 \[
 a(z)a(w) -(-1)^{p(a)^2}a(w)a(z) =c(w)\delta (z-\lambda w).
 \]
Now if we interchange the formal variables $z$ and $w$ we get
\begin{align*}
0=&(w-\lambda z)\left(a(w)a(z) -(-1)^{p(a)^2}a(z)a(w)\right)\\
& = -(z-\lambda ^{-1}w)(-1)^{p(a)^2}\lambda \left(a(z)a(w) -(-1)^{p(a)^2}a(w)a(z)\right),
\end{align*}
  which means
 that exists $d(w)\in U[\!] w^{\pm 1}]\!]$ such that
 \[
 a(z)a(w) -(-1)^{p(a)^2}a(w)a(z) = d(w)\delta (z-\lambda^{-1} w).
 \]
 Comparing the two, we get that
 \[
 c(w)\delta (z-\lambda w) =d(w)\delta (z-\lambda^{-1} w).
 \]
 Now since we have that $c(w) \neq 0$ from $a(z)a(w) -(-1)^{p(a)^2} a(w)a(z)\neq 0$,  Lemma \ref{lem:deltaIndep} gives us $\lambda ^2=1$.
 \end{proof}
  There are of course a variety of cases even when $\lambda ^2=1$, depending on the coefficients $c(w)\in U[\!] w^{\pm 1}]\!]$. There are  also restriction on the coefficients $c(w)$ placed by the parity. The easiest of these cases happen when  the coefficients $c(w)$ are simplest: the case of $\lambda =1, \quad  c(w)=1$ is only possible for odd parity, and describes the neutral free fermion (the free fermion of type D), see e.g. \cite{Kac}, \cite{AngTVA}, also  Section \ref{section: examples}. We  will use this example in Section \ref{section: examples}, as even though it is in some sense the easiest of all examples, it gives rise to the representations of several important infinite dimensional Lie algebras, and also is a generating field for the boson-fermion correspondence of type D-A (see \cite{AngTVA}).  The case of  $\lambda =-1, c(w) =1$ is only possible for even parity, and describes the  free boson of type C (considered in detail in \cite{OrlovLeur} as it gives rise to the CKP hierarchy). The case of $\lambda =-1$ and odd parity only allows for odd coefficients $c(w)$, thus the easiest such case is $c(w) =w$; though for reasons of convenience we rescale and use
  $\lambda =-1, c(w) =-2w$. This is the case of the free fermion of type B (see \cite{Ang-Varna2}, \cite{AngTVA}), and is described in Section \ref{section: examples}.
 \begin{remark}
As a corollary to the proof of the Lemma above it is easy to see that if a distribution $a(z)$ has a non-removable point of locality $\lambda \in \mathbb{C}, \ \lambda \neq \pm 1$, then $a(z)$ also has $\lambda ^{-1}$ as a non-removable point of locality.
\end{remark}
\begin{defn}
 Let $a(z)\in U[\!] z^{\pm 1}]\!]$ and $b(z)\in U[\!] z^{\pm 1}]\!]$. We say that $a(z)$ and $b(z)$ are {\it $N$-point mutually local} at $\lambda_1, \lambda_2, \dots, \lambda_N$ if there exist $n_1, n_2, \dots , n_N \in \mathbb{Z}_{\geq 0}$ such that
\begin{equation}
(z-\lambda _1 w)^{n_{1}}(z-\lambda _2 w)^{n_{2}}\cdots (z-\lambda _N w)^{n_{N}}\left(a(z)b(w)-(-1)^{p(a)p(b)}b(w)a(z)\right)=0.
\end{equation}
\end{defn}
\begin{remark}
\label{remark:productLoc}
If two elements $a(z), b(z)\in U[\!] z^{\pm 1}]\!]$ are mutually local at $\lambda_1, \lambda_2, \dots, \lambda_N$, then if we choose $n=max\{ n_1, n_2, \dots, n_N\}$ we have
\begin{equation}
((z-\lambda _1 w)(z-\lambda _2 w)\cdots (z-\lambda _N w))^n\left(a(z)b(w)-(-1)^{p(a)p(b)}b(w)a(z)\right)=0.
\end{equation}
\end{remark}
We will use the following notation:
\begin{defn}\label{prod}
 $\Pi_{z,w}=\prod_{i=1}^N(z-\lambda_iw)$.
\end{defn}
 Then we can restate the $N$-point locality  as follows: There exists $M\in \mathbb{N}$, such that
\begin{equation}
\Pi_{z,w}^M\left(a(z)b(w)-(-1)^{p(a)p(b)}b(w)a(z)\right)=0.
\end{equation}

Let $a(z), b(z) \in U[\!] z^{\pm 1}]\!]$, we will use the  indexing $a(z)=\sum _{n\in \mathbb{Z}} a_n z^{-n-1}$. Recall (notation \eqnref{eqn:plusonfield}) that with this indexing
\begin{equation}
a(z)_-:=\sum_{n\geq 0}a_nz^{-n-1},\quad a(z)_+:=\sum_{n<0}a_nz^{-n-1}.
\end{equation}
\begin{defn} {\bf (Normal ordered product)}
Let $a(z), b(z) \in U[\!] z^{\pm 1}]\!]$. Define
\begin{equation}
:a(z)b(w):=a(z)_+b(w)+(-1)^{p(a)p(b)}b(w)a_-(z).
\end{equation}
One calls this the normal ordered product of $a(z)$ and $b(w)$.
\end{defn}
\begin{remark}
The indexing we adopted above is the traditional vertex algebra indexing, and the definition of normal ordered product above is also the traditional one, and can be found in many books on the subject, e.g., \cite{FLM}, \cite{Kac}, \cite{LiLep}. We include it for completeness, as many of the usual properties of the normal ordered product remain valid after appropriate modification in important  cases of $N$-point locality.
\end{remark}
\begin{lem} {\bf (Operator Product Expansions (OPE's))}\label{lem:OPE}

 Suppose $a(z)$, $b(w)$ are {\it $N$-point mutually local}. Then we have
 \begin{equation}
 \label{eqn:OPEpolcor}
 a(z)b(w) =i_{z, w} \sum_{j=1}^N\sum_{k=0}^{n_j-1}\frac{c_{jk}(w)}{(z-\lambda_j w)^{k+1}} + :a(z)b(w):,
 \end{equation}
 where $c_{jk}(w)\in U[\!] w^{\pm 1}]\!]$, $j=1, \dots, N; k=0, \dots , n_j-1$, are the coefficients from Theorem \ref{mainresult}.
 \end{lem}
\begin{proof}
The $N$-point mutual locality of the distributions  $a(z)$, $b(w)$ implies from Theorem \ref{mainresult} that there exists $c_{jk}(w)\in U[\!] w^{\pm 1}]\!]$, $j=1, \dots, N; k=0, \dots , n_j-1$ such that
\begin{equation}
a(z)b(w)-(-1)^{p(a)p(b)}b(w)a(z)=\sum_{j=1}^N\sum_{k=0}^{n_j-1}c_{jk}(w)\partial_{\lambda_j w}^{(k)}\delta(z-\lambda_j w)
\end{equation}
We can separate the nonnegative powers of $z$ from the negative powers and we will get
\begin{equation}
a(z)_\pm b(w)-(-1)^{p(a)p(b)}b(w)a(z)_\pm=\sum_{j=1}^N\sum_{k=0}^{n_j-1}c_{jk}(w)\partial_{\lambda_j w}^{(k)}\delta(z-\lambda_j w)_{z, \pm},
\end{equation}
and from  \eqref{deltaPM1} we have then
\begin{equation}
a(z)_\pm b(w)-(-1)^{p(a)p(b)}b(w)a(z)_\pm=\begin{cases}
-i_{w,z} \sum_{j=1}^N\sum_{k=0}^{n_j-1}\frac{c_{jk}(w)}{(z-\lambda_j w)^{k+1}} &\quad\text{for subscript $+$} \\
 i_{z, w} \sum_{j=1}^N\sum_{k=0}^{n_j-1}\frac{c_{jk}(w)}{(z-\lambda_j w)^{k+1}} &\quad\text{for subscript $-$}  \end{cases}
\end{equation}
which we can rewrite as
\begin{equation}
a(z)b(w)-:a(z)b(w):=i_{z, w} \sum_{j=1}^N\sum_{k=0}^{n_j-1}\frac{c_{jk}(w)}{(z-\lambda_j w)^{k+1}}.
\end{equation}
\end{proof}
\begin{notation}[\bf Operator Product Expansion (OPE) and contraction]
Since for the commutation relations only the singular part of the \eqnref{eqn:OPEpolcor}  matters, by abusing notation we will write the above OPE as
 \begin{equation}
 a(z)b(w) \sim  \sum_{j=1}^N\sum_{k=0}^{n_j-1}\frac{c_{jk}(w)}{(z-\lambda_j w)^{k+1}}.
 \end{equation}
 The $\sim $ signifies that we have only written the singular part, and also we have omitted writing explicitly the expansion $i_{z, w}$, which we do acknowledge  tacitly. Often also the following notation is used for short:
 \begin{equation}\label{contraction}
\lfloor
ab\rfloor=a(z)b(w)-:a(z)b(w):= [a(z)_-,b(w)],
\end{equation}
i.e.,  the {\it contraction} of any two formal distributions
$a(z)$ and $b(w)$ is in fact also the $i_{z, w}$ expansion of the singular part of the OPE of the two distributions $a(z)$ and $b(w)$.
 \end{notation}
\begin{remark} {\bf ($\mathbf{(j, k)}$ -products $\mathbf{a(w)_{(j, k)} b(w)}$)}
\label{remark:(j, k)-products}
In the case of the usual 1 point locality (at $z=w$) the coefficients $c_{jk}(w)$ require a single index, $k$, as the OPE expansion in that case is:
 \begin{equation*}
 a(z)b(w) =i_{z, w} \sum_{k=0}^{n-1}\frac{c_{k}(w)}{(z- w)^{k+1}} + :a(z)b(w):.
 \end{equation*}
These coefficients $c_k(w)$ are used to define the  $k$-products  $a(w)_k b(w): =c_k(w)$, for $k=0, 1, \dots n-1$. In the case of $N$-point locality we can define $(j, k)$ -products  $a(w)_{(j, k)} b(w): =c_{jk}(w)$ for $j=1, \dots, N; k=0, \dots , n_j-1$. We can further define $c_{jk}(w):=0$ for $j=1, \dots, N; k\geq n_j$, see Remarks \ref{remark:CoefZero} and \ref{remark:productLoc}.
\end{remark}
We would like to define $(j, k)$ -products  $a(w)_{(j, k)} b(w)$ also for $k< 0; j=1, \dots, N$, and  as in the case of 1-point locality we will use normal ordered products to define these new products.
The definition of a normal ordered product $:a(z)b(w):$ gives an element of $U[\!] z^{\pm 1}, w^{\pm 1}]\!]$. If $a(z)$ and $b(z)$ are arbitrary elements of $U[\!] z^{\pm 1}]\!]$, we cannot apply the subsitution  $\lambda z $ for $w$, (or $\lambda w$ for $z$).  i.e., we cannot define elements $:a(z)b(\lambda z):$. But, if we impose the following important restriction on the elements of $U[\!] z^{\pm 1}]\!]$ we consider, we can define $:a(z)b(\lambda z):$ and $:a(\lambda  z)b(z):$ for any $\lambda \in \mathbb{C}$:
\begin{defn}
\begin{bf} (Field)\end{bf}\label{defn:field-fin}
 A field $a(z)$ on a vector space $V$ is a series of the form
\begin{equation}
a(z)=\sum_{n\in \mathbf{Z}}a_{(n)}z^{-n-1}, \ \ \ a_{(n)}\in
\End(V), \ \
\end{equation}
 such that for any $v\in V$ there exists $n\geq  0$ with $a_{(m)}v=0$ for $m\geq n$.
\end{defn}
\begin{remark} \label{remark:fieldindexing}
A field $a(z)$ is a distribution in $End(V)[\!] z^{\pm 1}]\!]$ which obeys the above finiteness condition. The coefficients $a_{(n)}, \ n\in \mathbb{Z}$ are usually called modes. The indexing above is  generally used in vertex algebras, as in a  vertex algebra  the nonnegative modes (the modes in front of negative powers of $z$)  annihilate a special vector called  vacuum vector (hence are called annihilation operators), and the negative modes (the modes in front of the positive powers of $z$) are the creation operators.
\end{remark}
\begin{remark}
Let  $a(z), b(z) \in End(V)[\!] z^{\pm 1}]\!]$ be fields on a vector space $V$. Then
$:a(z)b(\lambda z):$ and $:a(\lambda z)b( z):$ are well defined elements of  $End(V)[\!] z^{\pm 1}]\!]$, and are also fields in $V$ for any $\lambda \in \mathbb{C}^*$.
\end{remark}
The proof of this fact is the same as in  the usual case of fields in a vertex algebra.
\begin{lem}\label{lem:normalprodexpansion}{\bf Taylor expansion formula for normal ordered products}
Let  $a(z), b(z) \in End(V)[\!] z^{\pm 1}]\!]$ be fields on a vector space $V$, $\lambda \in \mathbb{C}^*$. Then
\begin{equation}
i_{z, z_0}:a(\lambda z +z_0)b( z): =\sum _{k\geq 0}\Big(:(\partial_{\lambda z} ^{(k)}a(\lambda z))b(z):\Big) z_0^k
\end{equation}
\end{lem}
Note that the right hand side if this equality doesn't make sense without the normal ordered product, as without the normal ordered product $a(\lambda z))b(z)$ may not be a well defined distribution at all.
\begin{proof}
From Proposition 2.4 in \cite{Kac} (for instance) we have that
\[
i_{z, z_0}a(\lambda z +z_0)= \sum _{k\geq 0}(\partial_{\lambda z} ^{(k)}a(\lambda z))z_0^k;
\]
and we can again separate the nonnegative powers of $z$ from the negative powers of $z$ on both sides and we will get
\[
\Big(i_{z, z_0}a(\lambda z +z_0)\Big)_{z, \pm}= \sum _{k\geq 0}(\partial_{\lambda z} ^{(k)}a(\lambda z))_{z, \pm}z_0^k.
\]
 But also, we have
\begin{align*}
i_{z, z_0}\Big(a(\lambda z +z_0)_{\lambda z+z_0, \pm}\Big)&=\begin{cases}\sum_{n<0} a_ni_{z, z_0}(\lambda  z+z_0)^{-n-1},\quad \text{ for $+$,}  \\
\sum_{n\geq 0} a_ni_{z, z_0}(\lambda  z+z_0)^{-n-1},\quad \text{ for $-$ }
\end{cases} \\
& =\Big(i_{z, z_0}a(\lambda z +z_0)\Big)_{z, \pm},
\end{align*}
since in $i_{z, z_0}(\lambda z+z_0)^n$ the negative powers of $z$ are only occur when $n<0$.
Thus
\begin{align*}
i_{z, z_0}\Big(a(\lambda z +z_0)_{\lambda z+z_0, \pm}\Big)
& =\sum _{k\geq 0}(\partial_{\lambda z} ^{(k)}a(\lambda z))_{z, \pm}z_0^k,
\end{align*}
and we have
\begin{align*}
& i_{z, z_0}:a(\lambda z +z_0)b(\lambda z):\\
& \hspace{1cm} =i_{z, z_0}\Big(a(\lambda z +z_0)_{\lambda z+z_0, +}b(z) +(-1)^{p(a)p(b)}b(z)a(\lambda z +z_0)_{\lambda z+z_0, -}\Big) \\
&  \hspace{1cm} = \sum _{k\geq 0}(\partial_{\lambda z} ^{(k)}a(\lambda z))_{z, +}b(z)z_0^k +(-1)^{p(a)p(b)} \sum _{k\geq 0}b(z)\partial_{\lambda z} ^{(k)}(a(\lambda z))_{z, -}z_0^k \\
&  \hspace{1cm} = \sum _{k\geq 0}\Big( :(\partial_{\lambda z} ^{(k)}a(\lambda z))b(z): \Big) z_0^k.
\end{align*}
\end{proof}
\begin{defn}
\begin{bf} ($(j, k)$ -products of fields)\end{bf}\label{defn:products-of-fields}
Let $a(z), b(z) \in End(V)[\!] z^{\pm 1}]\!]$ be fields on a vector space $V$. Define
\begin{align}
a(w)_{(j,-1)}b(w)&:=:a(\lambda_j w)b(w): \quad \text{for} \quad j=1, \dots, N\\
a(w)_{(j,-n-1)}b(w)&:=:(\partial_{\lambda_jw}^{(n)}a(\lambda_j w))b(w): \quad \text{for} \quad j=1, \dots, N, n>0 \\
a(w)_{(j, n)}b(w)&=c_{j,n}(w) \quad \text{for} \quad j=1, \dots, N, n > 0
\end{align}
\end{defn}
It is easy to check that these products of fields are actually well defined and are in fact fields on $V$.
\begin{remark}
We chose to define products of fields $a(w)_{(j, n)}b(w)$ in two separate ways: by the OPE coefficients for $n\geq 0$, and by the normal ordered products for $n<0$. We chose this definition because both the OPE's and the normal ordered products are very important concepts in mathematical physics and quantum field theory. The OPE's are used in quantum field theory to define (describe) the commutation relations and the parity of the particles (bosons vs fermions are defined by the summetry/antisymmetry of the OPE's). The normal ordered products give one of the most important operations used to produce new fields out of the given fields, and as such are used extensively both in quantum field theory and representation theory, as we will see in the Section \ref{section: examples}.

In the paper \cite{Li2} the author defines products of fields for  ``compatible" pairs of vertex operators (see Definition 3.4 in \cite{Li2}). The definition of compatible pairs in \cite{Li2} is more general than our definition of $N$-point local fields. However, even when the two coincide our definition of products (given below) differs. In particular, in the case of a single pole in the OPE the two definitions may coincide for some of the products, but for multiple poles the definitions differ substantially (see the Appendix \secref{section:appendix} for examples). The definition below incorporates parity, enabling us to apply it for  the boson-fermion correspondences of type B, C and D-A.
\end{remark}
\begin{lem}
\label{lem:ResFormulasForProducts}
{\bf (Residue formulas for the products of fields)}\\
The following formulas hold for $n\geq 0$, any $j=1, \dots, N$:
\begin{align*}
&a(w)_{(j, n)}b(w)  \\
& \hspace {0.1cm} =\text{Res}_z\left( \sum_{i=0}^{M-n-1}P_{i,j}(w)\frac{\prod_{i=1,i\neq j}^N(z-\lambda_iw)^{M}}{(z-\lambda_jw)^{-n-i}}\left(a(z)b(w)-(-1)^{p(a)p(b)}b(w)a(z)\right)\right);
\end{align*}
\begin{align*}
&a(w)_{(j,-n-1)}b(w) \\
& \hspace {0.1cm} =\text{Res}_z \left(a(z)b(w)i_{z, w}\frac{1}{(z-\lambda _jw)^{n+1}}-(-1)^{p(a)p(b)}b(w)a(z)i_{w, z}\frac{1}{(z-\lambda _jw)^{n+1}}\right).
\end{align*}
\end{lem}
\begin{proof}
The formula for $a(w)_{(j, n)}b(w)$ and $n\geq 0$ was established previously in \thmref{mainresult}. The formula for $a(w)_{(j,-n-1)}b(w)$ and $n\geq 0$ follows directly from the definition of normal ordered products and the Cauchy formulas, Lemma \ref{lem:CauchyForm}.
\end{proof}
The goal of this paper is to determine what type of structure is generated from a given number of fields which are mutually local at several points, if we incorporate all the descendants of those fields: descendants obtained by using normal order products, operator product expansions, derivatives and ``substitutions" at the point of locality (i.e., at $z=\lambda_i w$).
For the case of usual 1-point locality the structure that one obtains from allowing these operations (with imposing some additional auxiliary conditions) is a super vertex algebra structure (this is called the ``Reconstruction Theorem", see e.g. \cite{MR1334399}, \cite{Kac}, \cite{LiLep}). We will prove that if one allows locality at several, but finitely many points, one obtains instead the structure of a twisted vertex algebra.

\begin{defn}\label{defn:fielddesc} \begin{bf}(The Field Descendants Space  $\mathbf{\mathfrak{FD} \{a^0 (z), a^1 (z),  \dots a^p(z)\} }$)\end{bf} \\
Let $a^0 (z), a^1 (z), \dots a^p(z)$ be given fields on a vector space $W$, which are self-local and pairwise local with points of locality $\lambda_1,\dots, \lambda_N$. Denote by $\mathfrak{FD} \{a^0 (z), a^1(z), \dots a_p(z)\}$ the subspace of all fields on $W$ obtained from the fields $a^0 (z), a^1(z), \dots a^p(z)$ as follows:
\begin{enumerate}
\item $Id_W, a^0 (z), a^1(z), \dots a^p(z)\in \mathfrak{FD} \{ a^0 (z), a^1 (z), \dots , a^p(z) \}$;
\item  If $d(z)\in \mathfrak{FD} \{ a^0 (z), a^1 (z), \dots , a^p(z) \}$, then $\partial_z (d(z))\in \mathfrak{FD} \{ a^0 (z), \dots , a^p(z) \}$;
\item  If $d(z)\in \mathfrak{FD} \{ a^0 (z), a^1 (z), \dots , a^p(z) \}$, then $d(\lambda_i z)$ are also elements of\\  \mbox{ $\mathfrak{FD} \{a^0 (z), a^1(z), \dots a^p(z)\}$} for $i=1,\dots, N$;
\item
If $d_1(z), d_2 (z)$ are both in \mbox{$\mathfrak{FD} \{a^0 (z), a^1(z), \dots a^p(z)\}$}, then $:d_1(z)d_2(z):$ are also elements  $\mathfrak{FD} \{a^0 (z), a^1(z), \dots a^p(z)\}$, as well as all products $d_1(z)_{(j, k)}d_2(z)$ for any   $j=1,  \dots,  N$, \ $k\in \mathbb{Z}$;
\item all finite linear combinations of fields in $ \mathfrak{FD} \{ a^0 (z), a^1 (z), \dots , a^p(z) \}$ are still in $ \mathfrak{FD} \{ a^0 (z), a^1 (z), \dots , a^p(z) \}$.
\end{enumerate}
\end{defn}

\begin{remark}
The vector space $\mathfrak{FD} \{a^0 (z), a^1(z), \dots a_p(z)\}$ is given structure of a super vector space by the notion of parity of a field (see Definition \ref{defn:parity}).

\begin{remark}
If $a^0(z),\dots,a^p(z)$ have points of locality $\lambda_1,\dots, \lambda_N$, then the fields in $\mathfrak{FD} \{a^0 (z), \dots, a_p(z)\}$ will have points of locality $\lambda$ where $\lambda$ is in the multiplicative subgroup of $\mathbb C^*$ generated by $\lambda_1,\dots ,\lambda_N$.   If we want the total number of points of locality to be finite, then this subgroup needs to be a finite cyclic subgroup of $\mathbb C^*$.
\end{remark}

\end{remark}

For the following lemma, and the remainder of the paper, we will assume that the points of locality $\lambda_1, \lambda_2, \dots, \lambda_N \in \mathbb{C}$ form a multiplicative group (and hence are roots of unity). The proof of the lemma below uses only the group properties of the $\lambda_i$s, and not that they are specifically roots of unity. We present the proof in this form in the interest of generality.
\begin{lem}[\bf Dong's Lemma]
\label{lem:Dong}
 Suppose $a(z)$, $b(z)$, $c(z)$ are $N$-point mutually local fields on a vector space $V$, with points of locality  $\lambda_1, \lambda_2, \dots, \lambda_N \in \mathbb{C}$ being the distinct elements of a finite multiplicative group.
 Let $a(w)_{(j, k)}b(w)$, $j=1, \dots, N, k\in \mathbb{Z}$, be the products of fields defined in \ref{defn:products-of-fields}. Then the field $c(z)$ is $N$-point mutually local to  $a(w)_{(j, k)}b(w)$ for any  $j=1, \dots, N$, $k\in \mathbb{Z}$.
\end{lem}
\begin{proof}  Our proof is a generalization of the argument given in \cite{Kac} and \cite{LiLep}.
 Using the notation \defnref{prod} by hypothesis there exists $r\in\mathbb N$ such that
\begin{align}
\label{eqn:a-b-loc}
\Pi_{z_1, z_2}^ra(z_1)b(z_2)&=(-1)^{p(a)p(b)}\Pi_{z_1,z_2}^rb(z_2)a(z_1)  \\
\Pi_{z_2, z_3}^rb(z_2)c(z_3)&=(-1)^{p(b)p(c)}\Pi_{z_2,z_3}^rc(z_3)b(z_2) \\
\Pi_{z_1, z_3}^ra(z_1)c(z_3)&=(-1)^{p(a)p(c)}\Pi_{z_1,z_3}^rc(z_3)a(z_1)
\end{align}
We will prove the lemma separately for $k\geq 0$ and $k< 0$, as we will use the Residue formulas from Lemma \ref{lem:ResFormulasForProducts}.
Let $k\geq 0$.  Recalling Lemma \ref{lem:ResFormulasForProducts},  we need to show that there is an $M\gg 0$ such that
\begin{equation}
\Pi^M _{z_2 ,z_3}A=(-1)^{p(c)(p(a)+p(b))}\Pi^M_{z_2, z_3}B\label{mutuallylocal}
\end{equation}
where
\begin{align*}
A&:= (z_1-\lambda_jz_2)^{k-n}\Pi_{z_1, z_2}^n\left(a(z_1)b(z_2)c(z_3)-(-1)^{p(a)p(b)}b(z_2)a(z_1)c(z_3)\right), \\
B&:= (z_1-\lambda_jz_2)^{k-n}\Pi_{z_1, z_2}^n\left(c(z_3)a(z_1)b(z_2)-(-1)^{p(a)p(b)}c_(z_3)b(z_2)a(z_1)\right).
\end{align*}
Note that in the above expressions when we write $(z_1-\lambda_jz_2)^{k-n}\Pi_{z_1, z_2}^n$ it is short for $i_{z_1, z_2}(z_1-\lambda_jz_2)^{k-n}\Pi_{z_1, z_2}^n$, which equals  $i_{z_2, z_1}(z_1-\lambda_jz_2)^{k-n}\Pi_{z_1, z_2}^n$, as these are positive binomial powers.
Indeed taking $\text{Res}_{z_1}$ and multiplying \eqnref{mutuallylocal} by \\ $z_2^{-(N-1)n}\prod_{i=1,i\neq j}^N(\lambda_j-\lambda_i)^{-n}$ we would get
$$
\Pi^M _{z_2 ,z_3}a(z_2)_{(j, k)}b(z_2)c(z_3)=\Pi^M _{z_2, z_3}c(z_3)a(z_2)_{(j, k)}b(z_2)
$$
for $k\geq 0$.
To prove this one starts with the observation that for any $i, j=1, \dots, N$
\begin{align*}
(z_2-\lambda_iz_3)^{2r}& =\sum_{k=0}^{2r}\binom{2r}{k}(z_2-\lambda_jz_1)^{2r-k}(\lambda _jz_1-\lambda_iz_3)^k\\
& =\sum_{k=0}^{2r}\binom{2r}{k}\lambda _j^{2r-k}(z_2-\lambda_jz_1)^{2r-k}(z_1-\frac{\lambda _i}{\lambda _j}z_3)^k\\
& =\sum_{k=0}^{2r}\binom{2r}{k}\lambda _j^{4r-2k}(z_1-\lambda_j^{-1}z_2)^{2r-k}(z_1-\frac{\lambda _i}{\lambda _j}z_3)^k
\end{align*}
so that for any $i=1, \dots, N$
\begin{align*}
&(z_2-\lambda_iz_3)^{2rN}  \\
&\quad =\sum_{k_j=0,1\leq j\leq N}^{2r}\binom{2r}{k_1}\cdots \binom{2r}{k_N}\prod_{j=1}^N\lambda _j^{4r-2k_j}(z_1-\lambda_j^{-1}z_2)^{2r-k_j}\prod_{j=1}^N(z_1-\frac{\lambda _i}{\lambda _j}z_3)^{k_j}
\end{align*}
and
\begin{align*}
&\Pi_{z_2, z_3}^{2rN}=(z_2-\lambda_1z_3)^{2rN} \cdots(z_2-\lambda_Nz_3)^{2rN}  \\
&=\prod_{i=1}^N \left(\sum_{k_{i,j}=0, 1\leq j\leq N}^{2r}\binom{2r}{k_{i, 1}}\cdots \binom{2r}{k_{i, N}}\prod_{j=1}^N\lambda _j^{4r-2k_{i, j}}(z_1-\lambda_j^{-1}z_2)^{2r-k_{i, j}}\prod_{j=1}^N(z_1-\frac{\lambda _i}{\lambda _j}z_3)^{k_{i, j}}\right)\\
&\quad =\sum_{k_{i,j}=0, 1\leq i, j\leq N}^{2r}\prod_{i, j=1}^N\binom{2r}{k_{i, j}}\cdot \prod_{j=1}^N\lambda _j^{(4rN-2\sum_{i=1}^Nk_{i, j})}(z_1-\lambda_j^{-1}z_2)^{(2rN-\sum_{i=1}^Nk_{i, j})}\cdot \\
& \hspace{10cm} \cdot \prod_{i, j=1}^N(z_1-\frac{\lambda _i}{\lambda _j}z_3)^{k_{i, j}}
\end{align*}
Due to the group properties of the points of locality-- by the conditions of the lemma $\lambda_1, \lambda_2, \dots, \lambda_N \in \mathbb{C}$ are the distinct elements of a finite multiplicative group--we have that
\begin{equation*}
\prod_{i, j=1}^N(z_1-\frac{\lambda _i}{\lambda _j}z_3)^{k_{i, j}}= \prod _{l=1}^N (z_1-\lambda _lz_3)^{\sum _{\frac{\lambda _i}{\lambda _j}=\lambda _l}k_{i, j}}
\end{equation*}
In the above sum expansion of $\Pi_{z_2, z_3}^{2rN}$ then there two types of summands: the ones where for each $l=1, \dots, N$ the sum $\sum _{\frac{\lambda _i}{\lambda _j}=\lambda _l} k_{i, j}$ is at least $r$; and the summands where exists $l\in \{1, 2, \dots, N\}$ such that $\sum _{\frac{\lambda _i}{\lambda _j}=\lambda _l} k_{i, j} <r$. Now for the first type of summands we can factor out of each such summand
\begin{equation*}
\prod _{l=1}^N (z_1-\lambda _lz_3)^r=\Pi_{z_1, z_3}^r
\end{equation*}
Now we take $M=2rN +r$.
We have
\begin{align*}
&\Pi _{z_2, z_3}^{r}\Pi _{z_2, z_3}^{2rN}  \left(a(z_1)b(z_2)c(z_3)-(-1)^{p(a)p(b)}b(z_2)a(z_1)c(z_3)\right) \\
& =\Pi _{z_2, z_3}^{2rN} \Pi _{z_2, z_3}^{r}  \left((-1)^{p(c)p(b)}a(z_1)c(z_3)b(z_2)-(-1)^{p(a)p(b)}b(z_2)a(z_1)c(z_3)\right)
\end{align*}
We interchanged $b(z_2)$ and $c(z_3)$ in the second summand due to the presence of $ \Pi _{z_2, z_3}^{r}$. Further, the presence of $ \Pi _{z_1, z_3}^{r}$ in the first type of summands will allow us also to interchange $a(z_1)$ and $c(z_3)$ and to get
\begin{align*}
& \Pi _{z_2, z_3}^{2rN} \Pi _{z_2, z_3}^{r}  \left((-1)^{p(c)p(b)}a(z_1)c(z_3)b(z_2)-(-1)^{p(a)p(b)}b(z_2)a(z_1)c(z_3)\right)\\
& =\Pi _{z_2, z_3}^{2rN} \Pi _{z_2, z_3}^{r}  \left((-1)^{p(c)(p(a)+p(b))}c(z_3)a_(z_1)b(z_2)-(-1)^{p(a)(p(b)+p(c))}b(z_2)c(z_3)a(z_1)\right)\\
& =\Pi _{z_2, z_3}^{2rN} \Pi _{z_2, z_3}^{r}  \big((-1)^{p(c)(p(a)+p(b))}c(z_3)a_(z_1)b(z_2)\\
&\hspace{5.5cm} -(-1)^{p(a)(p(b)+p(c))+p(b)p(c)}c(z_3)b(z_2)a(z_1)\big).
\end{align*}
Thus on those first type of summands we have
\begin{align*}
&\Pi _{z_2, z_3}^{r}\Pi _{z_2, z_3}^{2rN}  \left(a(z_1)b(z_2)c(z_3)-(-1)^{p(a)p(b)}b(z_2)a(z_1)c(z_3)\right) \\
& =(-1)^{p(c)(p(a)+p(b))}\Pi _{z_2, z_3}^{r}\Pi _{z_2, z_3}^{2rN} \left(c(z_3)a(z_1)b(z_2)-(-1)^{p(a)p(b)}c(z_3)b(z_2)a(z_1)\right).
\end{align*}
Now there remains to consider what happens with the second type of summands, those for  which
there exists an  $l_0\in \{1, \dots, N\}$ such that the sum $\sum _{\frac{\lambda _i}{\lambda _j}=\lambda _{l_0}} k_{i, j}$ is {\bf less} than $r$.
For {\bf any} given $l$,  and any given $j$ is fixed, we have a  choice of $i$ for which $\frac{\lambda _i}{\lambda _{j}}=\lambda _l$ as they lie in a group. Thus for  {\bf any}  $j\in \{1, \dots, N\}$ exists an $i_j$ such that $k_{i_{j}, j} <r$ (since at least for one choice of $i_j$ we will have $k_{i_{j}, j}$ be a summand in $\sum _{\frac{\lambda _i}{\lambda _j}=\lambda _{l_0}} k_{i, j}<r$). Thus for any $j\in \{1, \dots, N\}$  we have
\[
\sum_{i=1}^Nk_{i, j}= \sum_{i=1, i\neq i_{j}}^Nk_{i, j} +k_{i_{j}, j}  < (N-1)2r +r=(2N-1)r,
\]
since each of the $k_{i, j}$ for $i\neq i_j$ are not greater than $2r$.
Hence $2Nr-\sum_{i=1}^N k_{i, j}> r$ for any $j=1, \dots, N$, and we can factor
\begin{equation*}
\prod _{l=1}^N (z_1-\lambda _lz_2)^{r}=\Pi_{z_1, z_2}^{r} \quad \text{from}\quad \prod_{j=1}^N(z_1-\lambda_j^{-1}z_2)^{(2rN-\sum_{i=1}^Nk_{i, j})}.
\end{equation*}
Thus the second type of summands of the binomial sum expansion of $\Pi _{z_2, z_3}^{2rN}$ will all have a factor of
$\Pi_{z_1, z_2}^{2Nr}$. Thus those summands will allow us to interchange $a(z_1)$ and $b(z_2)$ and will annihilate both \begin{equation*}
\Pi _{z_2, z_3}^{r}\Pi _{z_2, z_3}^{2rN}  \left(a(z_1)b(z_2)c(z_3)-(-1)^{p(a)p(b)}b(z_2)a(z_1)c(z_3)\right)
\end{equation*}
and
\begin{equation*}
\Pi _{z_2, z_3}^{r}\Pi _{z_2, z_3}^{2rN}  \left(c(z_3)a(z_1)b(z_2)-(-1)^{p(a)p(b)}c(z_3)b(z_2)a(z_1)\right).
\end{equation*}
Thus as a result, for both type of summands of the binomial sum expansion of $\Pi _{z_2, z_3}^{2rN}$ we have
\begin{align*}
&\Pi _{z_2, z_3}^{r}\Pi _{z_2, z_3}^{2rN}  \left(a(z_1)b(z_2)c_(z_3)-(-1)^{p(a)p(b)}b(z_2)a(z_1)c(z_3)\right)\\
& =\Pi _{z_2, z_3}^{r}\Pi _{z_2, z_3}^{2rN} (-1)^{p(c)(p(a)+p(b))} \left(c(z_3)a(z_1)b(z_2)-(-1)^{p(a)p(b)}c(z_3)b(z_2)a(z_1)\right).
\end{align*}
Hence we can multiply with the positive-powered binomial product  \mbox{$(z_1-\lambda_jz_2)^{k-n}\Pi_{z_1, z_2}^n$} both sides and get
\begin{equation*}
\Pi^{2Nr +r} _{z_2, z_3}A=(-1)^{p(c)(p(a)+p(b))}\Pi^{2Nr +r}_{z_2, z_3}B,
\end{equation*}
and we have have proved that
the field $c(z)$ is $N$-point mutually local to  $a(w)_{(j, k)}b(w)$ for any  $j=1, \dots, N$, $k\geq 0$.

Now for the products $a(w)_{(j, k)}b(w)$ when $k<0$. In this case the proof is very similar, but since our residue formula for $a(w)_{(j,-n-1)}b(w)$ contains the negative binomial power  \mbox{$\frac{1}{(z-\lambda _jw)^{n+1}}$}, we will need to pick a larger power $M=2N(r+n+1) +r$. In this case we take
\begin{align*}
A&:= \Big(a(z_1)b(z_2)c(z_3)i_{z_1, z_2}\frac{1}{(z_1-\lambda _jz_2)^{n+1}}\\
& \hspace{4cm} -(-1)^{p(a)p(b)}b(z_2)a(z_1)c(z_3)i_{z_2, z_1}\frac{1}{(z_1-\lambda _j z_2)^{n+1}}\Big), \\
B&:= \Big(c(z_3)a(z_1)b(z_2)i_{z_1, z_2}\frac{1}{(z_1-\lambda _j z_2)^{n+1}}\\
& \hspace{4cm} -(-1)^{p(a)p(b)}c(z_3)b(z_2)a(z_1)i_{z_2, z_1}\frac{1}{(z_1-\lambda _j z_2)^{n+1}}\Big).
\end{align*}
The larger power $M=2N(r+n+1) +r$ will ensure that again we will have
\begin{align*}
&\Pi_{z_2, z_3}^{r}\Pi _{z_2, z_3}^{2(r+n+1)N}  \Big(a(z_1)b(z_2)c_(z_3)i_{z_2, z_3}\frac{1}{(z_1-\lambda _j z_2)^{n+1}} \\
& \hspace{4cm} -(-1)^{p(a)p(b)}b(z_2)a(z_1)c(z_3)i_{z_2, z_1}\frac{1}{(z_1-\lambda _j z_2)^{n+1}} \Big)
\\
& =\Pi _{z_3, z_3}^{r}\Pi _{z_2, z_3}^{2(r+n+1)N} (-1)^{p(c)(p(a)+p(b))} \Big(c(z_3)a(z_1)b(z_2)i_{z_{1}, z_{2}}\frac{1}{(z_1-\lambda _j  z_2)^{n+1}} \\
& \hspace{4cm} -(-1)^{p(a)p(b)}c(z_3)b(z_2)a(z_1)i_{z_2, z_1}\frac{1}{(z_1-\lambda _j z_2)^{n+1}}\Big),
\end{align*}
as we have
\[
\Pi _{z_1, z_2}^{(r+n+1)}i_{z_1, z_2}\frac{1}{(z_1-\lambda _jz_2)^{n+1}}= \Pi _{z_1, z_2}^{(r+n+1)}i_{z_2, z_1}\frac{1}{(z_1-\lambda _jz_2)^{n+1}},
\]
and we can factor $\Pi _{z_1, z_2}^r$ out of both sides.
The proof for the products $a(w)_{(j,-n-1)}b(w)$ is then finished in the same way as for the positive (OPE) products
$a(w)_{(j, n)}b(w)$.
\end{proof}
\begin{cor}
Let $a^0 (z), a^1 (z), \dots a^p(z)$ be given fields on a vector space $W$, which are self-local and pairwise local with points of locality $\epsilon ^i$, $i=1, \dots, N$, where $\epsilon$ is a primitive root of unity. Then any two fields in $\mathfrak{FD} \{a^0 (z), a^1(z), \dots a_p(z)\}$ are self and mutually $N$-point local.
\end{cor}
\begin{proof}
The proof follows immediately from the results of \secref{section:NormalOrdProd}, and in particular Dong's Lemma.
\end{proof}
The Dong's Lemma above assures us that if we start with $N$-point mutually local fields with points of locality the roots of unity, then all the products of these fields, in particular  their normal ordered products, will stay $N$-point mutually local. Thus we can calculate OPE's \eqnref{eqn:OPEpolcor} between these normal ordered products. The next two results concern such OPE's.
 The first is a very useful theorem in representation theory, which we will use in \secref{section:reptheory} to prove that appropriate examples of normal ordered products  give rise to representations of the Lie algebras  $b_{\infty}$,  $c_{\infty}$ and  $d_{\infty}$.
\begin{thm}[Wick's Theorem, \cite{MR85g:81096}, \cite{MR99m:81001} or
\cite{Kac} ]  Let  $a^i(z)$ and $b^j(z)$ be formal
distributions with coefficients in the associative algebra
 $\End(\mathbb C[\mathbf x]\otimes \mathbb C[\mathbf y])$,
 satisfying
\begin{enumerate}
\item $[ \lfloor a^i(z)b^j(w)\rfloor ,c^k(x)_\pm]=[ \lfloor
a^ib^j\rfloor ,c^k(x)_\pm]=0$, for all $i,j,k$ and
$c^k(x)=a^k(z)$ or
$c^k(x)=b^k(w)$.
\item $[a^i(z)_\pm,b^j(w)_\pm]=0$ for all $i$ and $j$.
\item The products
$$
\lfloor a^{i_1}b^{j_1}\rfloor\cdots
\lfloor a^{i_s}b^{i_s}\rfloor:a^1(z)\cdots a^M(z)b^1(w)\cdots
b^N(w):_{(i_1,\dots,i_s;j_1,\dots,j_s)}
$$
have coefficients in
$\End(\mathbb C[\mathbf x]\otimes \mathbb C[\mathbf y])$ for all subsets
$\{i_1,\dots,i_s\}\subset \{1,\dots, M\}$, $\{j_1,\dots,j_s\}\subset
\{1,\cdots N\}$. Here the subscript
${(i_1,\dots,i_s;j_1,\dots,j_s)}$ means that those factors $a^i(z)$,
$b^j(w)$ with indices
$i\in\{i_1,\dots,i_s\}$, $j\in\{j_1,\dots,j_s\}$ are to be omitted from
the product
$:a^1\cdots a^Mb^1\cdots b^N:$ and when $s=0$ we do not omit
any factors.
\end{enumerate}
Then
\begin{align*}
:&a^1(z)\cdots a^M(z)::b^1(w)\cdots
b^N(w):= \\
  &\sum_{s=0}^{\min(M,N)}\sum_{\substack{i_1<\cdots<i_s,\\
j_1\neq \cdots \neq j_s}}\pm \lfloor a^{i_1}b^{j_1}\rfloor\cdots
\lfloor a^{i_s}b^{j_s}\rfloor
:a^1(z)\cdots a^M(z)b^1(w)\cdots
b^N(w):_{(i_1,\dots,i_s;j_1,\dots,j_s)}.
\end{align*}
Here the plus or minus sign is determined as follows:  each permutation of an adjacent odd field changes the sign.
\end{thm}
The proof of Wick's theorem is very similar to the proof in the case of usual locality, thus we will not give it here.

The OPE  of the $N$-point local fields  \eqnref{eqn:OPEpolcor} very much resembles the Partial Fraction Theorem one uses in analysis.
To make this more precise, we need some notations.
Denote by $\mathbf{F}^N_{\lambda}(z, w)$ the space of rational functions in the  formal variables $z, w$ with only poles at $z=0, w=0, \ z= \lambda_i w$, $i=1, \dots , N$.   In other words
\begin{equation}
\mathbf{F}^N_{\lambda}(z, w)=\mathbb C[z,w][z^{-1},w^{-1},(z-w)^{-1},(z-\lambda_1 w)^{-1},\cdots, (z-\lambda_{N}w)^{-1}].
\end{equation}
In particular, if $\epsilon$ is  a primitive $N$th root of unity, denote by $\mathbf{F}^N_{\epsilon}(z, w)$
 the space of rational functions in the  formal variables $z, w$ with only poles at $z=0, w=0, \ z= \epsilon^i w$, $i=1, \dots , N$.   In other words
\begin{equation}
\mathbf{F}^N_{\epsilon}(z, w)=\mathbb C[z,w][z^{-1},w^{-1},(z-w)^{-1},(z-\epsilon w)^{-1},\cdots, (z-\epsilon^{N-1}w)^{-1}].
\end{equation}
 $\mathbf{F}^N_{\lambda}(z, w)$ the space of rational functions in the  formal variables $z, w$ with only poles at $z=0, w=0, \ z= \lambda_i w$, $i=1, \dots , N$.   In other words
\begin{equation}
\mathbf{F}^N_{\lambda}(z, w)=\mathbb C[z,w][z^{-1},w^{-1},(z-w)^{-1},(z-\lambda_1 w)^{-1},\cdots, (z-\lambda_N w)^{-1}].
\end{equation}
Also, denote $\mathbf{F}^N_{\lambda}(z, w)^{+, w}$ the space of rational functions in the formal variables $z, w$ with only poles at $z=0,  \ z= \lambda_i w$, $i=1, \dots , N$.   In other words
\begin{equation}
\mathbf{F}^N_{\lambda}(z, w)^{+, w}=\mathbb C[z,w][z^{-1},(z-w)^{-1},(z-\lambda_1 w)^{-1},\cdots, (z-\lambda_N w)^{-1}].
\end{equation}
Note that we do not allow a pole at $w=0$ in $\mathbf{F}^N_{\lambda}(z, w)^{+, w}$, i.e., if $f(z, w)\in \mathbf{F}^N_{\lambda}(z, w)$, then $f(z, 0)$ is well defined. Similarly  if $\epsilon$ is  a primitive $N$th root of unity we use the notation $\mathbf{F}^N_{\epsilon}(z, w)^{+, w}$. \\

Let $\mathbf{F}^N_{\epsilon}(z_1, z_2, \dots , z_l)$ is the space of rational functions in variables $z_1, z_2, \dots, z_l$ with only poles at $z_m=0$, $m=1,\dots, l$, or at $z_i= \epsilon^{k} z_j$, and $i\neq j$, $1\leq k\leq N$.    In other words
\begin{equation}
\mathbf{F}^N_{\epsilon}(z_1, z_2, \dots , z_l)
=\mathbb C[z_i\,|\,1\leq i\leq l][z_m^{-1},(z_i-\epsilon^{k}z_j)^{-1}\,|\,1\leq m\leq l\,,\,i\neq j ].
\end{equation}
And $\mathbf{F}^N_{\epsilon}(z_1, z_2, \dots , z_l)^{+, z_l}$ is the space of rational functions in variables $z_1, z_2, \dots, z_l$ with only poles at $z_m=0$, $m=1,\dots, l-1$, or at $z_i= \epsilon^{k} z_j$, for $i\neq j$.
In other words
\begin{equation}
\mathbf{F}^N_{\epsilon}(z_1, z_2, \dots , z_l)^{+, z_l}
=\mathbb C[z_i\,|\,1\leq l\leq l-1][z_m^{-1},(z_i-\epsilon^{k}z_j)^{-1}\,|\,1\leq m\leq l-1\,,\,i\neq j].
\end{equation}\color{black}
 If $N$ is clear from the context, or the property doesn't depend on the particular value of $N$, we will just write  $\mathbf{F}_{\epsilon}(z, w)$ or $\mathbf{F}_{\epsilon}(z, w)^{+, w}$. In the same way we define $\mathbf{F}^N_{\lambda}(z_1, z_2, \dots , z_l)$ and $\mathbf{F}^N_{\lambda}(z_1, z_2, \dots , z_l)^{+, z_l}$.

Note that Partial Fraction Decomposition in $\mathbb C[z]$ is a consequence of this ring being a Euclidean domain while $\mathbb C[z,w]$ is not a Euclidean domain.  None the less we have
\begin{lem}[Partial Fraction Decomposition]\label{pfd}
Let
$$
f(z,w):=\frac{p(z,w)}{z^mw^n(z-w)^{n_1}\cdots (z-\epsilon^{N-1}w)^{n_N}}\in  \mathbf{F}^N_{\epsilon}(z, w).
$$
where $p(z,w)\in\mathbb C[z,w]$.
Then there exists unique $A_{i,j}(w)\in \mathbb C[w,w^{-1}]$ such that
\begin{align}
f(z,w)&=q(z,w)+\frac{A_{-n-1,0}(w)}{z^n}+\cdots +\frac{A_{-1,0}(w)}{z} \label{f(z,w)} \\
&\hskip 42pt +\frac{A_{-n_1,1}(w)}{(z-w)^{n_1}}+\cdots +\frac{A_{-1,1}(w)}{(z-w)} \notag\\
&\hskip 42pt+\frac{A_{-n_2,\epsilon}(w)}{(z-\epsilon w)^{n_2}}+\cdots +\frac{A_{-1,\epsilon}(w)}{(z-\epsilon w)} \notag \\
&\hskip 72pt \vdots\hskip 100pt \vdots \notag \\
&\hskip 42pt+\frac{A_{-n_{N},\epsilon^{N-1}}(w)}{(z-\epsilon^{N-1} w)^{n_{N}}}+\cdots +\frac{A_{-1,\epsilon^{N-1}}(w)}{(z-\epsilon^{N-1} w)} \notag
\end{align}
where $q(z,w)\in\mathbb C[z,w,w^{-1}]$.
\end{lem}
Note: the Partial Fraction Decomposition can also be given more generally for $\mathbf{F}^N_{\lambda}(z, w)$ without any change, but we will only use the case of $\mathbf{F}^N_{\epsilon}(z, w)$.
\begin{proof}  By the Partial Fraction Decomposition Theorem in one variable (view $z/w$ as the one variable) we have
\begin{align*}
\frac{1}{w^mz^n\prod_{i=0}^{N-1}(z-\epsilon^i w)^{n_i}}&=\frac{w^{-m-n-\sum_{i=0}^{N-1}n_i}}{\left(\frac{z}{w}\right)^n\prod_{i=0}^{N-1}\left(\frac{z}{w}-\epsilon^i \right)^{n_i}}  \\
&=w^M\Bigg(\frac{a_{-n-1,0}w^n}{z^n}+\cdots +\frac{a_{-1,0}w}{z} \\
&\qquad +\frac{a_{-n_1,1}}{\left(\frac{z}{w}-1\right)^{n_1}}+\cdots +\frac{a_{-1,1}}{\left(\frac{z}{w}-1\right)} \\
&\qquad+\frac{a_{-n_2,\epsilon}}{\left(\frac{z}{w}-\epsilon \right)^{n_2}}+\cdots +\frac{a_{-1,\epsilon}}{\left(\frac{z}{w}-\epsilon \right)}  \\
&\qquad  \vdots\hskip 100pt \vdots \\
&\qquad +\frac{a_{-n_{N},\epsilon^{N-1}}}{\left(\frac{z}{w}-\epsilon^{N-1} \right)^{n_{N}}}+\cdots +\frac{a_{-1,\epsilon^{N-1}}}{\left(\frac{z}{w}-\epsilon^{N-1} \right)} \Bigg),
\end{align*}
where $M=-m-n-\sum_{i=0}^{N-1}$.     Multiplying through by $p(z,w)$ we get
\begin{align*}
\frac{p(z,w)}{w^mz^n\prod_{i=0}^{N-1}(z-\epsilon^i w)^{n_i}}
&=w^M\Bigg(\frac{a_{-n-1,0}p(z,w)w^n}{z^n}+\cdots +\frac{a_{-1,0}p(z,w)w}{z} \\
&\qquad +\frac{a_{-n_1,1}p(z,w)}{\left(\frac{z}{w}-1\right)^{n_1}}+\cdots +\frac{a_{-1,1}p(z,w)}{\left(\frac{z}{w}-1\right)} \\
&\qquad+\frac{a_{-n_2,\epsilon}p(z,w)}{\left(\frac{z}{w}-\epsilon \right)^{n_2}}+\cdots +\frac{a_{-1,\epsilon}p(z,w)}{\left(\frac{z}{w}-\epsilon \right)}  \\
&\qquad  \vdots\hskip 100pt \vdots \\
&\qquad +\frac{a_{-n_{N},\epsilon^{N-1}}p(z,w)}{\left(\frac{z}{w}-\epsilon^{N-1} \right)^{n_{N}}}+\cdots +\frac{a_{-1,\epsilon^{N-1}}p(z,w)}{\left(\frac{z}{w}-\epsilon^{N-1} \right)} \Bigg).
\end{align*}
Now the first summation we can rewrite as
\begin{align*}
&\frac{a_{-n-1,0}p(z,w)w^n}{z^n}+\cdots +\frac{a_{-1,0}p(z,w)w}{z} \\
&\quad =q_0(z,w)+\frac{A_{-n-1,0}(w)}{z^n}+\cdots +\frac{A_{-1,0}(w)}{z} ,
\end{align*}
where $q_0(z,w)\in \mathbb C[z,w,w^{-1}]$ and $, A_{-n-1,0}(w),\dots, A_{-1,0}(w)\in\mathbb C[w,w^{-1}]$.
Similarly
\begin{align*}
&\frac{a_{-n_k,\epsilon^k}p(z,w)}{\left(\frac{z}{w}-\epsilon^k \right)^{n_k}}+\cdots +\frac{a_{-1,\epsilon^k}p(z,w)}{\left(\frac{z}{w}-\epsilon \right)}   \\
&\quad =q_k(z,w)+\frac{A_{-n_k,\epsilon^k}(w)}{\left(\frac{z}{w}-\epsilon^k \right)^{n_k}}+\cdots +\frac{A_{-1,\epsilon^k}(w)}{\left(\frac{z}{w}-\epsilon \right)}
\end{align*}
where $q_k(z,w)\in \mathbb C[z,w,w^{-1}]$ and $A_{-n-1,\epsilon^k}(w),\dots, A_{-1,\epsilon^k}(w)\in\mathbb C[w,w^{-1}]$.
Uniqueness follows by induction on multiplying \eqref{f(z,w)} by $z^k$ for various $k$ and setting $z=0$ or multiplying by $(z-\epsilon^kw)^l$ for various $l$ and setting $z=\epsilon^kw$.
\end{proof}
Define
\begin{equation}
\text{\rm Res}_{z=\epsilon^kw}\left(\frac{p(z,w)}{z^mw^n(z-w)^{n_1}\cdots (z-\epsilon^{N-1}w)^{n_N}}\right) =A_{-1,\epsilon^k}(w).
\end{equation}
{\begin{remark} Note that $A_{i,j}(w)\in \mathbb C[w,w^{-1}]$ even if the initial function $f(z,w)\in  \mathbf{F}^N_{\epsilon}(z, w)^{+, w}$, as exemplified by
\[
\frac{1}{(z-w)(z+w)}=\frac{\frac{1}{2w}}{z-w} -\frac{\frac{1}{2w}}{z+w}.
\]
Thus we see that
\[
\text{\rm Res}_{z=\pm w}\frac{1}{(z-w)(z+w)} =\mp \frac{1}{2w}
\]
\end{remark}
Lastly in this section, we want to finish with the following lemma that calculates the  OPE between a field and a normal ordered product that doesn't necessarily obey the conditions of Wick's theorem.}
\begin{lem}\label{lem:tripleOPEcoef}
Let $a(z), b(z), c(z)$ be fields on a vector space $W$ satisfying the hypothesis of Dong's Lemma and suppose $\{\lambda_1,\dots,\lambda_N\}=\{\epsilon^i\,|\,1\leq i\leq N\}$. We have
\begin{align*}
&a(z):b(w)c(w):\sim (-1)^{p(a)p(b)}i_{z, w}\sum_{m=1}^{N}\sum_{p=0}^{M}\frac{:b (w)c^{ac}_{mp}(w):}{(z-\epsilon ^m w)^{p+1}} \\
 & \hspace{3cm} - (-1)^{p(a)p(b)}\Big(i_{z, w}\sum_{m=1}^{N}\sum_{p=0}^{M}\frac{c^{bac}_{mp}(w)}{(z-\epsilon ^m w)^{p+1}}\Big),
\end{align*}
where the triple coefficient $c^{bac}_{mp}(w)$ is defined in the proof below.
\end{lem}
\begin{proof}  First we note by\eqnref{lem:OPE}
\begin{align*}
 b(w)_+a(z) -(-1)^{p(a)p(b)}a (z) b (w)_+=-i_{z, w}\sum_{m=1}^{N}\sum_{p=0}^{M}\frac{c^{ba}_{mp}(z)}{(w-\epsilon ^m z)^{p+1}}.
\end{align*}
We have
\begin{align*}
a (z)&\big( b (w)_+c (w)  +(-1)^{p(b)p(c)}c (w)b (w)_-\big)\\
&= \big (a (z) b (w)_+\big) c (w) +(-1)^{p(b)p(c)}\Big(a(z) c (w)\Big) b (w)_-\\
&= (-1)^{p(a)p(b)} b (w)_+a(z) c (w) - (-1)^{p(a)p(b)} \Big(i_{z, w}\sum_{m=1}^{N}\sum_{p=0}^{M}\frac{c^{ba}_{mp}(z)}{(w-\epsilon ^m z)^{p+1}}\Big) c(w)  \\
& \qquad + (-1)^{p(b)p(c)} \Big(i_{z, w}\sum_{m=1}^{N}\sum_{p=0}^{M}\frac{c^{ac}_{mp}(w)}{(z-\epsilon ^m w)^{p+1}}\Big)b (w)_- +(-1)^{p(b)p(c)} \big(:a(z) c(w):\big) b (w)_-
\end{align*}
\begin{align*}
& = (-1)^{p(a)p(b)}b (w)_+\Big(i_{z, w}\sum_{m=1}^{N}\sum_{p=0}^{M}\frac{c^{ac}_{mp}(w)}{(z-\epsilon ^m w)^{p+1}}\Big) \\
 & \hspace{2cm} + (-1)^{p(b)p(c)} \Big(i_{z, w}\sum_{m=1}^{N}\sum_{p=0}^{M}\frac{c^{ac}_{mp}(w)}{(z-\epsilon ^m w)^{p+1}}\Big)b (w)_- \\
  & \hspace{3cm} - (-1)^{p(a)p(b)} \Big(i_{z, w}\sum_{m=1}^{N}\sum_{p=0}^{M}\frac{c^{ba}_{mp}(z)}{(w-\epsilon ^m z)^{p+1}}\Big)c(w) \\
  &  \hspace{1cm}   +(-1)^{p(a)p(b)} b (w)_+\big(:a (z) c (w):\big) +(-1)^{p(b)p(c)} \big(:a (z) c (w):\big) b (w)_-.
\end{align*}
Now we have to rearrange here the OPE
\[
 \Big(i_{z, w}\sum_{m=1}^{N}\sum_{p=0}^{M}\frac{c^{ba}_{mp}(z)}{(w-\epsilon ^m z)^{p+1}}\Big)c(w):
\]
We use the OPE expansion of $c^{ba}_{mp}(z)$ and $c(w)$
\[
c^{ba}_{mp}(z)c(w)= i_{z, w}\sum_{m=1}^{N}\sum_{p=0}^{M}\frac{c^{c^{ba}_{mp}c}(w)}{(z-\epsilon ^m w)^{p+1}} +:c^{ba}_{mp}(z)c(w):
\]
so we have
\begin{align*}
 \Big(i_{z, w}&\sum_{m=1}^{N}\sum_{p=0}^{M}\frac{c^{ba}_{mp}(z)}{(w-\epsilon ^m z)^{p+1}}\Big)c(w)=\\ & \Big(i_{z, w}\sum_{m_1=1}^{N}\sum_{p_1=0}^{M}\frac{1}{(w-\epsilon ^{m_1} z)^{{p_1}+1}}\Big)\Big(i_{z, w}\sum_{m_2=1}^{N}\sum_{p_2=0}^{M}\frac{c^{c^{ba}_{m_2p_2}}(w)}{(z-\epsilon ^{m_2} w)^{{p_2}+1}}\Big)  \\
 &\hspace{6cm} +i_{z, w}\sum_{m=1}^{N}\sum_{p=0}^{M}\frac{:c^{ba}_{mp}(z)c(w):}{(w-\epsilon ^m z)^{p+1}},
\end{align*}
which we can rearrange further using the Partial Fractions \lemref{pfd} and Dong's \lemref{lem:Dong} to get
\begin{align*}
 \Big(i_{z, w}\sum_{m=1}^{N}& \sum_{p=0}^{M}\frac{c^{ba}_{mp}(z)}{(w-\epsilon ^m z)^{p+1}}\Big)c(w)=i_{z, w}\sum_{m=1}^{N}\sum_{p=0}^{M}\frac{c^{b,ac}_{mp}(w)}{(z-\epsilon ^m w)^{p+1}} +Reg^{bac}(z, w);
\end{align*}
where we denote by $c^{b,ac}_{mp}(w)$ the coefficients in the combined singular part of the partial fraction expansions (after the re-arrangement), and the $Reg^{b,ac}(z, w)$ denotes the regular (nonsingular) part of the partial fraction expansion.
If we use this, we have
\begin{align*}
a (z) &:b (w)c (w):\\
& =(-1)^{p(a)p(b)}i_{z, w}\sum_{m=1}^{N}\sum_{p=0}^{M}\frac{b (w)_+c^{ac}_{mp}(w)+(-1)^{(p(a)+p(c))p(b)}c^{ac}_{mp}(w)b (w)_-}{(z-\epsilon ^m w)^{p+1}} \\
 & \hspace{2cm} + (-1)^{p(a)p(b)} i_{z, w}\sum_{m=1}^{N}\sum_{p=0}^{M}\frac{c^{b,ac}_{mp}(w)}{(w-\epsilon ^m z)^{p+1}}\\
&  \hspace{1cm}   +(-1)^{p(a)p(b)} b (w)_+:a (z) c (w): +(-1)^{p(b)p(c)} :a (z) c (w): b (w)_- \\ \\
&=(-1)^{p(a)p(b)}i_{z, w}\sum_{m=1}^{N}\sum_{p=0}^{M}\frac{:b (w)c^{ac}_{mp}(w):}{(z-\epsilon ^m w)^{p+1}}  - (-1)^{p(a)p(b)} i_{z, w}\sum_{m=1}^{N}\sum_{p=0}^{M}\frac{c^{b,ac}_{mp}(w)}{(z-\epsilon ^m w)^{p+1}}   \\
&  +(-1)^{p(a)p(b)}\Big( b (w)_+:a (z) c (w): +(-1)^{(p(a)+p(c))p(b)} :a (z) c (w):b (w)_-\Big)\\
&\hspace{8.5cm} +Reg^{b,ac}(z, w).
\end{align*}
We can express the last two lines as a Taylor expansion of normal ordered triple products, but that is unimportant, as they do not contribute to the OPE coefficients. The summands contributing to the OPE coefficients are $:b (w)c^{ac}_{mp}(w):$  and $c^{b,ac}_{mp}(w)$.
\end{proof}
We will use the triple OPE coefficient $c^{b,ac}_{mp}(w)$ later on in the paper, hence we formalize the following:
\begin{notation} {\bf (Triple OPE coefficients)}
Let $a(z), b(z), c(z)$ be fields on a vector space $W$ satisfying the hypothesis of Dong's Lemma. Denote by  $c^{a,bc}_{mp}(w)$
the coefficient in the double OPE expansion of
\[
 \Big(i_{z, w}\sum_{m=1}^{N}\sum_{p=0}^{M}\frac{c^{ab}_{mp}(z)}{(w-\epsilon ^m z)^{p+1}}\Big)c(w);
\]
see above, where $c^{ab}_{mp}$ is itself the coefficient  from the OPE expansion \eqnref{eqn:OPEpolcor} of the fields $a(w)$ and $b(z)$, in particular
\[
a(w)_+b(z)-(-1)^{p(a)p(b)}b(z)a(w)_+=-i_{z, w}\sum_{m=1}^{N}\sum_{p=0}^{M}\frac{c^{ab}_{mp}(z)}{(w-\epsilon ^m z)^{p+1}}.
\]
Similarly, we denote
 by  $c^{ab,c}_{mp}(w)$
the coefficient in the double  OPE expansion of
\[
 \Big(i_{z, w}\sum_{m=1}^{N}\sum_{p=0}^{M}c(z)\frac{c^{ab}_{mp}(w)}{(z-\epsilon ^m w)^{p+1}}\Big);
\]
where $c^{bc}_{mp}$ is itself the coefficient  from the OPE expansion \eqnref{eqn:OPEpolcor} of the fields $a(z)$ and $b(w)$, in particular
\[
a(z)_- b(w)-(-1)^{p(a)p(b)}b(w)a(z)_-=i_{z, w}\sum_{m=1}^{N}\sum_{p=0}^{M}\frac{c^{ab}_{mp}(w)}{(z-\epsilon ^m w)^{p+1}}.
\]
\end{notation}
From the lemma above we see that $c^{a,bc}_{mp}(w)$ appears in the OPE expansion of \\ $b(z)(:a(w)c(w):)$, and  $c^{ab,c}_{mp}(w)$ appears in the OPE expansion of $(:a(z)c(z):)b(w):$.
\begin{cor}\label{cor:tripleOPEcoefN}
Let $a(z), b(z), c(z)$ be fields on a vector space $W$ satisfying the hypothesis of Dong's \lemref{lem:Dong}. We have
\begin{align*}
&a(z):b(\epsilon ^i w)c(w):\sim (-1)^{p(a)p(b)}i_{z, w}\sum_{m=1}^{N}\sum_{p=0}^{M}\frac{:b (\epsilon ^iw)c^{ac}_{mp}(w):}{(z-\epsilon ^m w)^{p+1}}  \\
 & \hspace{3cm} - (-1)^{p(a)p(b)}\Big(i_{z, w}\sum_{m=1}^{N}\sum_{p=0}^{M}\frac{\epsilon ^{-(p+1)i}c^{bac}_{[m+i],p}(w)}{(z-\epsilon ^m w)^{p+1}}\Big).
\end{align*}
where in the index $[m+i]$ is the residue of $m+i$ modulo $N$.
\end{cor}

 \ \\

\section{Examples}
\label{section: examples}

We now will consider the simplest examples of $N$-point local fields and give selected examples of elements of $\mathbf{\mathfrak{FD}}$ which can be used to obtain representations of certain Lie algebras. As in the preceding section, the points of locality in all examples are cyclic groups of roots of unity. Our first three examples will have points of locality $\{\lambda_1, , \lambda_2  \} = \{ 1 ,-1 \} $, and the fourth has points of locality $\{\lambda_j = \epsilon^j : j = 0, \dots, N-1\}$ where $\epsilon$ is a primitive root of unity.

\begin{example} {\bf Free neutral fermion of type B and bosonization of type  B}\\
\label{example:B}
Let $\{\lambda_1, , \lambda_2  \} = \{ 1 ,-1 \} $, we will consider elements of the space $\mathbf{\mathfrak{FD}}\{ \phi ^B(z)   \} $, where  $\phi ^B(z)$ is an
 {\bf odd}  self-local field, which we index in the form $\phi ^B (z)=\sum _{n\in \mathbb{Z}} \phi^B_n z^{n}$. By \lemref{lem:OPE}
to describe its self locality, it is necessary and sufficient to give its OPE expansion:
\begin{equation}
\label{eqn:OPE-B}
\phi ^B(z)\phi ^B(w)\sim \frac{-2w}{z+w}.
\end{equation}
This OPE expansion completely determines the commutation relations between the modes $\phi^B_n$, $n\in \mathbf{Z}$:
\begin{equation}
\label{eqn:Com-B}
\{\phi^B_m,\phi^B_n\}=2(-1)^m\delta _{m, -n}1,
\end{equation}
which show that the  modes generate a Clifford algebra $\mathit{Cl_B}$.
\begin{remark} One often (especially in physics)  writes the OPE expansion \eqnref{eqn:OPE-B} in the following form
\begin{equation*}
\phi ^B(z)\phi ^B(w)\sim \frac{z-w}{z+w},
\end{equation*}
as it illustrates at a glance the fact that this field describes a fermion: the OPE $\frac{z-w}{z+w}$ is antisymmetric with exchange of $z$ and $w$.
\end{remark}

With points of locality  $ \{ 1 ,-1 \} $ the following element are  in $\mathbf{\mathfrak{FD}}\{ \phi ^B(z)   \} $:
 $$
 \phi^B(w)_{(2, 0)}\phi ^B (w)=-2w,
 $$
 and for $n>0$
 $$
 \phi^B(w)_{(2, n)}\phi ^B(w)=0.
 $$
 We have $\phi(w)^B_{(1, n)}\phi ^B(w)=0$ for $n\geq 0$. There are no other nontrivial products with $n>0$ involving only the generating field $\phi ^B(z)$ in $\mathbf{\mathfrak{FD}}\{ \phi^B(z)   \} $.
\begin{remark}
\label{remark:shiftsOPE-B}
The OPE \eqnref{eqn:OPE-B} of the field $\phi ^B(z)$ displays a property that is new to $N$-point local fields. Due to the powers of $w$ in the numerator (called ``shifts") we have a less symmetric OPEs on descendants, for example:
\begin{equation}
\label{eqn:OPE-B-derivative}
(\partial_z \phi ^B(z))\phi ^B(w)\sim \frac{2w}{(z+w)^2},  \quad \phi ^B(z)(\partial_w\phi ^B(w))\sim \frac{2w}{(z+w)^2} -\frac{2}{z+w}.
\end{equation}
We want to note that the shifts change with application of the derivative operator: we started with a shift $w^1$ in the OPE of $\phi ^B(z)$ and $\phi ^B(w)$ for the pole at $z=-w$ of order 1, we have no  shift (shift of $1=w^0$)  in the OPE of $\phi ^B(z)$ and $\partial_w \phi ^B(w)$ for the pole at $z=-w$ of order 1. This is important for the algebraic structure of twisted vertex algebra that we establish in Section \ref{sectio:tva}, see Remark \ref{remark:shift}.
\end{remark}
Let  $\mathit{F_B}$ be the highest weight representation of $\mathit{Cl_B}$ with the vacuum vector denoted $|0\rangle $, so that
$\phi^B_n|0\rangle=0 \ \text{for} \  n<0$. We call  $\mathit{F_B}$ the Fock space, and note that  $\phi ^B(z)$ a field on $\mathit{F_B}$ and thus we can form normal ordered products.

The first normal ordered product we consider is the normal ordered product $:\phi ^B(z)\phi ^B(z):$.
\begin{cor}\label{fact:BnormOrd=1}
$:\phi ^B(z)\phi ^B(z): \, =\phi^B(z)_{(1, -1)}\phi ^B(z)=1$.
\end{cor}
\begin{proof}
By definition
\begin{align*}
:\phi ^B(z)\phi ^B(z): &=\phi ^B(z)_+\phi ^B(z)+(-1)^{p(\phi ^B)^2}\phi ^B(z)\phi ^B_-(z) \\
& \hspace{2cm}=\big(\sum _{n\geq 0} \phi ^B_n z^{n} \big)\big(\sum _{n\in \mathbb{Z}} \phi ^B_n z^{n}\big) -\big(\sum _{n\in \mathbb{Z}} \phi ^B_n z^{n} \big)\big(\sum _{n<0} \phi ^B_n z^{n}\big) \\
& \hspace{2cm} =\big(\sum _{n\geq 0} \phi ^B_n z^{n} \big)\big(\sum _{n\geq 0} \phi ^B_n z^{n} \big)-\big(\sum _{n< 0} \phi ^B_n z^{n} \big)\big(\sum _{n< 0} \phi ^B_n z^{n} \big)  \\
& \hspace{2cm} =\sum_{k\geq 0}\sum _{m+ n=k} \phi ^B_m \phi ^B_n z^{m+n} - \sum_{k< 0}\sum _{m+ n=k} \phi ^B_m \phi ^B_n z^{m+n}.
\end{align*}
Note that the commutation relations \eqnref{eqn:Com-B} imply $\phi ^B_m\phi ^B_n + \phi ^B_n\phi ^B_m=0$ if both  $m, n >0$, or if both  $m, n <0$; also $\phi ^B_0\phi ^B_0 =1$.
Thus all but one of the terms above equals 0,
\begin{equation*}
:\phi ^B(z)\phi ^B(z): =\sum_{k\geq 0}\sum _{m+ n=k} \phi ^B_m \phi ^B_n z^{m+n} - \sum_{k< 0}\sum _{m+ n=k} \phi ^B_m \phi ^B_n z^{m+n}=\phi ^B_0\phi ^B_0 =1.
\end{equation*}
\end{proof}

The two-variable normal ordered product $:\phi ^B(z)\phi ^B(w):$  is very important to the representation theory of infinite rank Lie algebras, as it gives a representation of the algebra
$b_{\infty}$ (\cite{DJKM-4}, \cite{YouBKP}, \cite{OrlovVanDeLeur}, see also next section). Note that from the OPE expansion  \lemref{lem:OPE} we have a shorter formula for the normal ordered product $:\phi ^B(z)\phi ^B(w):$\,:
\[
:\phi ^B(z)\phi ^B(w): =\phi ^B(z)\phi ^B(w) +i_{z, w} \frac{2w}{z+w}.
\]
This formula allows us to calculate commutation relations and will be used in the next \secref{section:reptheory}. In this case we have for the contraction $\lfloor
\phi^B \phi^B \rfloor$ of the fields
$\phi ^B(z), \phi ^B(w)$ that  \ $\lfloor
\phi^B \phi^B \rfloor =i_{z, w} \frac{2w}{z+w}$, but even though this contraction satisfies condition (1) of Wick's theorem, condition (2) is not satisfied due to the presence of the shift (see Remark \ref{remark:shiftsOPE-B}). Thus   Wick's theorem cannot be used in this case, but instead we will use a direct computation instead (see \secref{section:reptheory}).

Other elements of $\mathbf{\mathfrak{FD}}\{ \phi^B(z) \} $  are found by (3) substituting  at points of locality, in this case at $\lambda =-1$. Thus it is natural to consider the field $\phi ^B(-z)$, and all of its descendants. Of particular interest is then the normal ordered product $:\phi ^B(z)\phi ^B(-z):$, which yields the following:
\begin{lem} ( \cite{DJKM-4}, \cite{Ang-Varna2}, \cite{AngTVA})
The field $h^B(z)$ given by:
\begin{equation}
\label{eqn:normal-order-h-B}
h^B(z)= \frac{1}{4}(:\phi ^B(z)\phi ^B(-z): -1).
\end{equation}
has only  odd-indexed modes,   $h^B (z)=\sum _{n\in \mathbb{Z}} h_{2n+1} z^{-2n-1}$ and moreover it is  a Heisenberg field. I.e., it has OPE with itself given by:
\begin{equation}
\label{eqn:HeisOPEsB}
h^B (z)h^B (w)\sim \frac{zw(z^2 +w^2)}{2(z^2 -w^2)^2}\sim \frac{1}{4}\frac{w^2}{(z-w)^2} + \frac{1}{4}\frac{w}{(z-w)} - \frac{1}{4}\frac{w^2}{(z+w)^2} + \frac{1}{4}\frac{w}{(z+w)},
\end{equation}
and its  modes, $h_n, \ n\in 2\mathbb{Z}+1$, generate a \textbf{twisted} Heisenberg algebra $\mathcal{H}_{\mathbb{Z}+1/2}$ with relations $[h_m,h_n]=\frac{m}{2}\delta _{m+n,0}1$, \ $m,n$ are odd  integers.
 \end{lem}
 We have the product $\phi ^B(w)_{(2, -1)}\phi ^B(w)=-1-4h^B (w)$, and the OPE yields the $(j,n)$-products for $h^B (w)_{(j, n)}h^B (w)$ for $n\geq 0$:
 \begin{align}
& h^B (w)_{(1, 0)}h^B (w)= \frac{1}{4}w, \quad h^B (w)_{(2, 0)}h^B (w)= \frac{1}{4}w \\
& h^B (w)_{(1, 1)}h^B (w)= \frac{1}{4}w^2, \quad h^B (w)_{(2, 1)}h^B (w)= -\frac{1}{4}w^2 \\
& h^B (w)_{(1, n)}h^B (w)= 0, \quad h^B (w)_{(2, n)}h^B (w)=0 \quad \text{for} \quad n\geq 2
 \end{align}
 The negative products $h^B (w)_{(j, n)}h^B (w)$ are also very important, as they give rise to representations of the Virasoro algebra, but we will not consider them in this paper.

The twisted Heisenberg algebra $\mathcal{H}_{\mathbb{Z}+1/2}$   has (up-to isomorphism) only one irreducible highest weight module $B_{1/2}\cong \mathbb{C}[x_1, x_3, \dots , x_{2n+1}, \dots ]$.
The fermionic space of states $\mathit{F_B}$ decomposes as $\mathit{F_B} =B_{1/2} \oplus B_{1/2}$ (\cite{Ang-Varna2}, \cite{AngTVA}) and thus cannot be realized as a twisted module for a vertex super algebra.
\end{example}

\label{example:C}
\begin{example} {\bf Free twisted boson of type C}\\
For $\{\lambda_1, , \lambda_2  \} = \{ 1 ,-1 \} $, we now consider elements of the space $\mathbf{\mathfrak{FD}}\{ \phi ^C (z) \} $, where $\phi ^C (z)$ is an
{\bf even}  self-local field (see \cite{DJKM6} and \cite{OrlovLeur}) with OPE:
\begin{equation}
\label{equation:OPE-C}
\phi ^C(z)\phi ^C(w)\sim \frac{1}{z+w}.
\end{equation}
We index the field as follows: $\phi ^C (z)=\sum _{n\in \mathbf{Z+1/2}} \phi ^C_n z^{n-1/2}$ (the half-integers are commonly used when indexing in this case, see  \cite{DJKM6} and \cite{OrlovLeur}).
In modes, the OPE leads to commutation relations:
\begin{equation}
\label{eqn:Com-C}
[\phi ^C_m, \phi ^C_n]=(-1)^{n-\frac{1}{2}}\delta _{m, -n}1,
\end{equation}
Despite the name of the field ($\phi ^C (z)$) the modes form a Lie algebra, not a Clifford algebra, which we denote by $L_C$.
The two-variable normal ordered product $:\phi ^C(z)\phi ^C(w):$  is very important to the representation theory of infinite rank Lie algebras, as it gives the representation of the algebra
$c_{\infty}$ (see \cite{DJKM6} and \secref{section:reptheory} below). Note that from the OPE expansion \lemref{lem:OPE} we have a shorter formula for the normal ordered product $:\phi ^C(z)\phi ^C(w):$:
\[
:\phi ^C(z)\phi ^C(w): =\phi ^C(z)\phi ^C(w) -i_{z, w} \frac{1}{z+w}.
\]
This formula allows us to calculate commutation relations very easily, and can be used directly instead of Wick's theorem. In this case we have for the contraction
$\lfloor
\phi^C \phi^C \rfloor$ of the fields
$\phi ^C(z), \phi ^C(w)$ that  \ $\lfloor
\phi^C \phi^C \rfloor =i_{z, w} \frac{1}{z+w}$, which we will use for the application of Wick's theorem in \secref{section:reptheory}.

Let $\mathit{F_C}$ be a the highest weight module of $L_C$ with  the vacuum vector $|0\rangle $, such that $\phi ^C_n|0\rangle=0 \ \text{for} \  n < 0$.  Then  $\phi ^C(z)$ is field on $\mathit{F_C}$ (which we call the Fock space) and thus we can define normal ordered products, and compute elements of $\mathbf{\mathfrak{FD}}\{ \phi ^C (z) \} $.
First consider the normal ordered product $:\phi ^C(z)\phi ^C(-z):$.
We have
\begin{lem} (\cite{DJKM6})
The field $h^C(z)$ given by:
\begin{equation}
\label{eqn:normal-order-h-C}
h^C(z)= \frac{1}{2}(:\phi ^C(z)\phi ^C(-z): -1).
\end{equation}
has only  odd-indexed modes,   $h^C (z)=\sum _{n\in \mathbb{Z}} h_{2n+1} z^{-2n-1}$ and moreover it is  a Heisenberg field. I.e., it has OPE with itself given by:
\begin{equation}
\label{eqn:HeisOPEsC}
h^C (z)h^C (w)\sim -\frac{(z^2 +w^2)}{2(z^2 -w^2)^2}\sim -\frac{1}{4}\frac{1}{(z-w)^2} - \frac{1}{4}\frac{1}{(z+w)^2} ,
\end{equation}
and its  modes, $h_n, \ n\in 2\mathbb{Z}+1$, generate a \textbf{twisted} Heisenberg algebra $\mathcal{H}_{\mathbb{Z}+1/2}$ with relations $[h_m,h_n]=-\frac{m}{2}\delta _{m+n,0}1$, \ $m,n$ are odd  integers.
 \end{lem}
 The above result appears in \cite{DJKM6} and in \cite{OrlovLeur}.

We have the product $\phi ^C(w)_{(2, -1)}\phi ^C(w)=2h^C (w)$, and the lemma gives the $(j,n)$-products $h^C (w)_{(j, n)}h^C (w)$ for $n\geq 0$:
 \begin{align}
& h^C (w)_{(1, 0)}h^C (w)= 0, \quad h^C (w)_{(2, 0)}h^C (w)= 0 \\
& h^C (w)_{(1, 1)}h^C (w)= - \frac{1}{4}, \quad h^C (w)_{(2, 1)}h^C (w)= -\frac{1}{4} \\
& h^C (w)_{(1, n)}h^C (w)= 0, \quad h^C (w)_{(2, n)}h^C (w)=0 \quad \text{for} \quad n\geq 2
 \end{align}
\end{example}

\smallskip

 The next pair of examples is based on  a self-local generating field we which  is in fact well known as an example of a generating field for  a super vertex algebra, and thus is 1-point local at $\lambda =1$. If we allow though substitutions at new points of locality, at $\lambda =\epsilon ^k$, $k=0, 1, \dots, N-1$, where $\epsilon$ is a $N$-th root of unity, we get a new and interesting structure which wasn't otherwise possible--- the boson-fermion correspondence of type D-A (\cite{AngTVA}). We start with $N=2$.

\begin{example} \label{example:D} {\bf Free neutral fermion of type D and bosonization of type  D-A}\\
For this example, we again take points of locality $\{\lambda_1, , \lambda_2  \} = \{ 1 ,-1 \} $. We consider a single odd self-local field $\phi ^D(z)$, which we index in the form $\phi ^D(z)=\sum _{n\in \mathbf{Z+1/2}} \phi^D_n z^{-n-1/2}$
The OPE of $\phi ^D(z)$ is given by
\begin{equation}
\label{equation:OPE-D}
\phi ^D(z)\phi ^D(w)\sim \frac{1}{z-w}.
\end{equation}
This OPE completely determines the commutation relations between the modes $\phi^D_n$, $n\in \mathbf{Z +1/2}$:
\begin{equation}
\label{eqn:Com-D}
\{\phi^D_m,\phi^D_n\}=\delta _{m, -n}1.
\end{equation}
and so the modes generate a Clifford algebra $\mathit{Cl_D}$.
We have thus just one nontrivial $(j,n)$, $n$ nonnegative product for the generating field $\phi^D  (w)$: the $\phi^D (w)_{(1, 0)}\phi^D  (w)=1$, and $\phi^D (w)_{(1, n)}\phi^D  (w)=0$ for $n>0$; $\phi^D (w)_{(2, n)}\phi^D  (w)=0$ for $n\geq 0$.

 In this example the Fock space is denoted $\mathit{F_D}$ and is the highest weight module of $\mathit{Cl_D}$ with vacuum  vector $|0\rangle $, such that $\phi^D_n|0\rangle=0 \ \text{for} \  n >0$.  (The space $\mathit{F_D}$ can be given a super-vertex algebra structure, as is known from e.g. \cite{Kac}, \cite{Wang}.)
Note that $\phi^D (z)$ a field on $\mathit{F_D}$ and thus normal ordered products are defined.
Consider the normal ordered product $:\phi^D  (z)\phi^D  (z):$.
\begin{cor}
$:\phi^D  (z)\phi^D  (z): =0 =\phi^D (z)_{(1, -1)}\phi^D  (z)$.
\end{cor}
\begin{proof}
By definition
\begin{align*}
:\phi^D  (z)&\phi^D  (z):\, =\phi^D  (z)_+\phi^D  (z)+(-1)^{p(\phi^D  )^2}\phi^D  (z)\phi^D  _-(z) \\
& =\big(\sum _{n\geq 0} \phi^D _{-n-1/2} z^{n} \big)\big(\sum _{n\in \mathbf{Z}} \phi^D _{-n-1/2} z^{n}\big) -\big(\sum _{n\in \mathbf{Z}} \phi^D _{-n-1/2} z^{n} \big)\big(\sum _{n<0} \phi^D _{-n-1/2} z^{n}\big) \\
&   =\big(\sum _{n> 0} \phi^D _{-n-1/2} z^{n} \big)\big(\sum _{n > 0} \phi^D _{-n-1/2} z^{n} \big)-\big(\sum _{n< 0} \phi^D _{-n-1/2} z^{n} \big)\big(\sum _{n< 0} \phi^D _{-n-1/2} z^{n} \big)  \\
&  =\sum_{\substack{k\geq 0 \\ k\in  \mathbf{Z }}}\sum _{m+ n=k} \phi^D  _{-m-1/2} \phi^D _{-n-1/2} z^{m+n} - \sum_{\substack{k< 0\\ k\in  \mathbf{Z}}}\sum _{m+ n=k} \phi^D  _{-m-1/2} \phi^D _{-n-1/2} z^{m+n}.
\end{align*}
Note that the commutation relations \eqnref{eqn:Com-D} imply $\phi^D _{-m-1/2}\phi^D _{-n-1/2} + \phi^D_{-n-1/2}\phi^D_{-m-1/2}=0$ if both  $m, n >0$, or if both  $m, n <0$.
Thus $:\phi^D  (z)\phi^D  (z): =0$.
\end{proof}
The more general normal ordered product $:\phi^D  (z)\phi^D  (w):$  is again very important, this time  to the representation theory the infinite rank Lie  algebra
$d_{\infty}$ (\cite{Wang}, \cite{WangDual}). Again from \lemref{lem:OPE} we have a shorter formula for the normal ordered product \\ \mbox{$:\phi^D  (z)\phi^D  (w):$}
\[
:\phi^D  (z)\phi^D  (w): =\phi^D  (z)\phi^D  (w) +i_{z, w} \frac{1}{z-w}.
\]
This formula allows us to calculate commutation relations for $d_{\infty}$. In this case we have for the contraction $\lfloor
\phi^D \phi^D \rfloor$ of the fields
$\phi ^D(z), \phi ^D(w)$ that  \ $\lfloor
\phi^D \phi^D \rfloor =i_{z, w} \frac{1}{z-w}$, which we will use for the application of Wick's theorem in \secref{section:reptheory}.

The space  $\mathbf{\mathfrak{FD}}\{ \phi ^D(z) \} $ contains the normal ordered product $:\phi ^D(z)\phi ^D(-z):$ (for which the point of locality $\lambda_2 = -1$ is not removable).  Here we calculate its OPE:
\begin{lem} (\cite{Ang-Varna2}, \cite{AngTVA})
The field $h^D(z)$ given by:
\begin{equation}
\label{eqn:normal-order-h-D}
h^D (z)= \frac{1}{2}:\phi ^D(z)\phi ^D(-z) :
\end{equation}
has only  odd-indexed modes,    $h^D (z)=\sum _{n\in \mathbb{Z}} h_{n} z^{-2n-1}$ and moreover it is  a Heisenberg field. I.e., it has OPE with itself given by:
\begin{equation}
\label{eqn:HeisOPEsD}
h^D (z)h^D (w)\sim \frac{zw}{(z^2-w^2)^2} \sim \frac{1}{4}\frac{1}{(z-w)^2} - \frac{1}{4}\frac{1}{(z+w)^2}
\end{equation}
and its  modes, $h_n, \ n\in \mathbb{Z}$, generate an ordinary, \textbf{untwisted}, Heisenberg algebra $\mathcal{H}_{\mathbb{Z}}$ with relations $[h_m,h_n]=m\delta _{m+n,0}1$, \ $m,n$   integers.
 \end{lem}
We have the product $\phi^D (w)_{(2, -1)}\phi^D  (w)=-2h^D (w)$, and the products $h^D (w)_{(j, n)}h^D (w)$ for $n\geq 0$:
 \begin{align}
& h^D (w)_{(1, 0)}h^D (w)= 0, \quad h^D (w)_{(2, 0)}h^D (w)= 0 \\
& h^D (w)_{(1, 1)}h^D (w)= \frac{1}{4}, \quad h^D (w)_{(2, 1)}h^D (w)= -\frac{1}{4} \\
& h^D (w)_{(1, n)}h^D (w)= 0, \quad h^D (w)_{(2, n)}h^D (w)=0 \quad \text{for} \quad n\geq 2
 \end{align}
 The negative products $h^D (w)_{(j, n)}h^D (w)$ are also very important, as they give rise to representations of the Virasoro algebra, but we will not consider them in this paper. The fact that the space $\mathit{F_D}$ is a representation of the Virasoro algebra with central charge $c=\frac{1}{2}$ can be found in \cite{Wang}, \cite{WangDual}. The Virasoro field is given by $: \phi ^D(w) \partial_w \phi ^D (w):$.

Unlike the twisted Heisenberg algebra, the untwisted Heisenberg algebra  has infinitely many  irreducible highest weight modules, labeled by the action of $h_0$.
The fermionic space of states $\mathit{F_D}$ decomposes as
\begin{equation}
W=F_D \cong \oplus _{i\in \mathbb{Z}} B_i \cong  \mathbb{C}[e^{\alpha}, e^{-\alpha}] \ten \mathbb{C}[x_1, x_2, \dots , x_n, \dots ]=B_D,
\end{equation}
where by $\mathbb{C}[e^{\alpha}, e^{-\alpha}]$ we mean the Laurent polynomials with one variable   $e^{\alpha}$ (\cite{Ang-Varna2}, \cite{AngTVA}).
The isomorphism above is  as  Heisenberg modules, where $e^{n\alpha }$ is identified as the highest weight vector for the irreducible Heisenberg module $B_n$ with highest weight $n$.
\end{example}
\smallskip

We want to finish with the last example, which involves the primitive root of unity of order $N\in \mathbb{N}$.
\begin{example} {\bf Bosonisation of type  D-A  (general order $\mathbf{N}$)}\\
In this example we extend the boson-fermion correspondence of type D-A from the previous example to a general order $N$ as follows. Given a primitive $N$-th root of unity $\epsilon$, choose $\{ \lambda_1, \lambda_2, \dots, \lambda_N\} = \{ 1, \epsilon, \dots, \epsilon^{N-1}\}$.
We again consider the free field $\phi ^D(z)=\sum _{n\in \mathbf{Z+1/2}} \phi^D_n z^{-n-1/2}$, with OPE's with itself given by
\eqref{equation:OPE-D}. Thus the descendent field space  $\mathbf{\mathfrak{FD}}\{ \phi ^D(z) \} $ contains all $T_{\epsilon}^i\phi ^D (z)=\phi ^D (\epsilon ^i z)$, for any $0\le i\le N-1$, with OPE's
\begin{equation*}
\phi ^D (\epsilon^iz)\phi ^D(\epsilon^jw) \sim \frac{1}{\epsilon ^iz- \epsilon ^jw};
\end{equation*}
We note that for the above fields the points of locality$\lambda_j=\epsilon^j$, $j=1, \dots, N-1$ are  not removable in the following
normal ordered products, and in particular we have the following lemma:
\begin{lem} The field $h^D(z)$ given by
\begin{equation}
\label{eqn:HeisOrderN}
h^D(z)=\frac{1}{N}\sum_{i=0}^{N-1} \epsilon^{i-1}  : \phi^D (\epsilon^{i-1}z)\phi^D (\epsilon^{i}z) : =\sum _{n\in \mathbb{Z}}h_n z^{-Nn-1}
\end{equation}
can be indexed as
\begin{equation}
h^D(z)=\sum _{n\in \mathbb{Z}}h_n z^{-Nn-1}
\end{equation}
and  is a Heisenberg field. I.e., it has OPE with itself given by:
\begin{equation}
h^D(z)h^D(w)\sim  \frac{z^{N-1}w^{N-1}}{(z^N -w^N)^2},
\end{equation}
and thus  the commutation relations $[ h^D_m, h^D_n ] = m\delta _{m+n, 0} 1$ for the Heisenberg algebra $\mathcal{H}_{\mathbb{Z}}$  hold.
\end{lem}
\end{example}

\smallskip

\ \\

\section{Applications to representation theory}
\label{section:reptheory}
 In this section we consider the applications of the examples of normal ordered products from the previous section to the representation theory of the three double-infinite rank Lie algebras   $b_{\infty}$,  $c_{\infty}$ and  $d_{\infty}$. The three representations given here also appear in the works of Date-Jimbo-Kashiwara-Miwa (see  \cite{DJKM-4} and \cite{DJKM6}), and  the representation of $d_{\infty}$ is also considered in works of Wang (see \cite{Wang}, \cite{WangKac},  \cite{WangDual}). We present the proofs here for completeness, and as an application of the N-point locality OPEs, the use of which generalizes the techniques already known for super vertex algebras.

 First, we recall the definitions and notations for the double-infinite rank Lie algebras $a_{\infty}$, $b_{\infty}$,  $c_{\infty}$ and  $d_{\infty}$ as in \cite{Kac-Lie}

 The  Lie algebra  $\bar{a}_{\infty}$ is the Lie algebra of infinite matrices
\begin{equation}
\bar{a}_{\infty}=\{(a_{ij}) | \ i, j\in \mathbb{Z}, \ a_{ij} =0 \ \text{for} |i-j|\gg 0 \}.
\end{equation}
As usual denote the elementary matrices by $E_{ij}$.

The algebra $a_{\infty}$ is a central extension of $\bar{a}_{\infty}$ by a central element $c$, $a_{\infty}=\bar{a}_{\infty}\oplus \mathbb{C}c$, with cocycle given by
\begin{equation}
\label{equation:cocycle-a}
 C(A, B) =Trace([J, A]B),
\end{equation}
where the matrix $J=\sum_{i\le 0}E_{ii}$. In particular
\begin{align*}
C(E_{ij},E_{ji})&=-C(E_{ji},E_{ij})=1,\quad \text{if} \enspace i\leq 0,\enspace j\geq 1 \\
C(E_{ij},E_{kl})&=0\quad \text{ in all other cases},
\end{align*}
and we denote $c_{ijkl}=C(E_{ij},E_{kl})c$.

The three   algebras   $\bar{b}_{\infty}$,  $\bar{c}_{\infty}$ and  $\bar{d}_{\infty}$ are all defined as subalgebras of $\bar{a}_{\infty}$, each preserving different bilinear form (\cite{Kac-Lie}).

\subsection{The infinite dimensional Lie algebra $b_{\infty}$} \ \\
The infinite dimensional Lie algebra $\bar{b}_{\infty}$ is the subalgebra of $\bar{a}_{\infty}$ consisting of the infinite matrices preserving the bilinear form $B(v_i, v_j)=(-1)^i \delta_{i, -j}$, i.e.,
\begin{equation}
\bar{b}_{\infty}=\{(a_{ij})\in \bar{a}_{\infty} | \ a_{ij}=(-1)^{i+j-1}a_{-j, -i} \}.
\end{equation}
Denote by $b_{\infty}$ the  central extension of $\bar{b}_{\infty}$ by a central element $c$, $b_{\infty}=\bar{b}_{\infty}\oplus \mathbb{C} \frac{1}{2}c$, where $c$ is the cocycle  for $a_{\infty}$, \eqref{equation:cocycle-a} (see \cite{Kac-Lie}).  Thus the commutation relations for the elementary matrices in $b_\infty$ is
\begin{align*}
[E_{ij},E_{kl}]=\delta_{jk}E_{il}-\delta_{li}E_{kj}+\frac{1}{2}C(E_{ij},E_{kl})c.
\end{align*}
The generators for the algebra $b_\infty$ can be written in terms of these elementary matrices as:
\[
\{ (-1)^j E_{i, -j} -(-1)^i E_{j, -i} ,\enspace c \,|\,  i, j \in \mathbb{Z}\}.
\]
We can arrange the non-central generators in a generating series
\begin{equation}
E^B(z, w) =\sum _{i, j\in \mathbb{Z}} ((-1)^jE_{i, -j}-(-1)^iE_{j, -i})z^{i-1}w^{-j}.
\end{equation}
\begin{lem}
The generating series  $E^B(z,w)$ obeys the following relations:
\begin{equation}
E^B(z,w) = -E^B(w,z)
\end{equation}
and
\begin{align}
[E^B(z_1, w_1),E^B(z_2, w_2)]&=-z_{2}E^B(z_1,w_2)\delta(w_1 +z_{2} )
 + w_2E^B(z_1,z_2)  \delta(w_1+w_2)  \notag \\
  &\quad   + z_2E^B(w_1,w_2)z_1 \delta(z_1 +z_2 )
 - w_2E^B( w_1, z_2)\delta(z_1+w_2)  \notag \\
&\quad + \frac{c}{4}i_{w_1, z_2}\frac{w_1-z_2}{w_1+z_2}i_{z_1, w_2}\frac{z_1-w_2}{z_1+w_2}  -\frac{c}{4}i_{w_2, z_1}\frac{w_2-z_1}{z_1+w_2}i_{z_2, w_1}\frac{z_2-w_1}{z_2+w_1}- \notag \\
&-\frac{c}{4}i_{w_1, w_2}\frac{w_1-w_2}{w_1+w_2}i_{z_1, z_2}\frac{z_1-z_2}{z_1+z_2}+\frac{c}{4}i_{w_2, w_1}\frac{w_2-w_1}{w_1+w_2}i_{z_2, z_1}\frac{z_2-z_1}{z_1+z_2}
\end{align}
\end{lem}
\begin{proof}  We begin with
\begin{align*}
[E^B(z_1, w_1)&, E^B(z_2, w_2)]=\sum_{i, j,k,l\in \mathbb{Z} } [( (-1)^j E_{i,-j} -(-1)^i E_{j,-i} ),( (-1)^l E_{k,-l} -(-1)^k E_{l,-k} )]z_1^{i}z_2^{k}w^{j}_1 w^{l}_2 \\
&=\sum_{i, j,k,l\in \mathbb{Z} }(-1)^{j +l} [ E_{i,-j} , E_{k,-l}]z_1^{i}z_2^{k}w^{j}_1w^{l}_2 - \sum_{i, j,k,l\in \mathbb{Z}}(-1)^{j +k} [ E_{i,-j} , E_{l,-k}]z_1^{i}z_2^{k}w^{j}_1w^{l}_2   \\
&\quad - \sum_{i, j,k,l\in \mathbb{Z}}(-1)^{i+l} [ E_{j,-i} , E_{k,-l}]z_1^{i}z_2^{k}w^{j}_1w^{l}_2   + \sum_{i, j,k,l\in \mathbb{Z}}(-1)^{i +k} [ E_{j,-i},  E_{l,-k} ]z_1^{i}z_2^{k}w^{j}_1w^{l}_2  \\   \\
\end{align*}
\begin{align*}
 &=\sum_{i, j,k,l\in \mathbb{Z}} (-1)^{j+l}(\delta_{-jk}E_{i,-l}-\delta_{i,-l}E_{k,-j})z_1^{i}z_2^{k}w^{j}_1w^{l}_2 \\
&\quad +\frac{1}{2}c \sum _{i\leq 0, j\geq 1} (-1)^{j+i} z_1^{i}z_2^{j}w^{-j}_1w^{-i}_2- \frac{1}{2} c \sum _{j\leq 0, i\geq 1} (-1)^{j+i}z_1^{i}z_2^{j}w^{-j}_1w^{-i}_2\\
 &\quad - \sum_{i, j,k,l\in \mathbb{Z}}(-1)^{j+k} (\delta_{-j,l} E_{i,-k}-\delta_{i,-k}   E_{l, -j})z_1^{i}z_2^{k}w^{j}_1w^{l}_2   \\
 &\quad - \frac{1}{2}c  \sum _{i\leq 0, j\geq 1} (-1)^{j+i} z_1^{i}z_2^{-i}w^{-j}_1w^{j}_2 + \frac{1}{2} c \sum _{i\geq 1,j\leq 0} (-1)^{j+i} \frac{1}{2}z_1^{i}z_2^{-i}w^{-j}_1w^{j}_2 \\
 &\quad - \sum_{i, j,k,l\in \mathbb{Z} }(-1)^{i+l} (\delta_{ -i,k}E_{  j,-l}- \delta_{ j,-l}E_{k, -i} )z_1^{i}z_2^{k}w^{j}_1w^{l}_2   \\
&\quad - \frac{1}{2}c  \sum _{j\leq 0, i\geq 1} (-1)^{j+i} z_1^{-i}z_2^{i}w^{j}_1w^{-j}_2 +  \frac{1}{2} c \sum _{i\leq 1,j\geq 1} (-1)^{j+i} z_1^{-i}z_2^{i}w^{j}_1w^{-j}_2 \\
 &\quad + \sum_{i, j,k,l\in \mathbb{Z}}(-1)^{i+k} (\delta_{-il}  E_{  j, -k}-\delta_{j,-k}   E_{  l,  -i})z_1^{i}z_2^{k}w^{j}_1w^{l}_2  \\
&\quad + \frac{1}{2}c \sum _{i\geq 1, j\leq 0} (-1)^{j+i} z_1^{-i}z_2^{-j}w^{j}_1w^{i}_2- \frac{1}{2} c \sum _{i\leq 0, j\geq 1} (-1)^{j+i}z_1^{-i}z_2^{-j}w^{j}_1w^{i}_2\\
 &=-z_{2}E(z_1,w_2)\delta(w_1 +z_{2} )
 + w_2E(z_1,z_2)  \delta(w_1+w_2)  \\
  &\quad   + z_2E(w_1,w_2)z_1 \delta(z_1 +z_2 )
 - w_2E( w_1, z_2)\delta(z_1+w_2)   \\
&\quad +\frac{1}{2}c \sum _{i\leq 0}(\frac{-z_1}{w_2})^i\sum_{j\geq 1}(\frac{-z_2}{w_1})^j - \frac{1}{2} c \sum _{i\geq 1}(\frac{-z_1}{w_2})^i\sum_{j\leq 0}(\frac{-z_2}{w_1})^j\\
&\quad - \frac{1}{2}c \sum _{i\leq 0}(\frac{-z_1}{z_2})^i\sum_{j\geq 1}(\frac{-w_2}{w_1})^j  + \frac{1}{2} c \sum _{i\geq 1}(\frac{-z_1}{z_2})^i\sum_{j\leq 0}(\frac{-w_2}{w_1})^j \\
&\quad - \frac{1}{2}c \sum _{i\leq 0}(\frac{-w_1}{w_2})^i\sum_{j\geq 1}(\frac{-z_2}{z_1})^j +   \frac{1}{2} c \sum _{i\geq 1}(\frac{-w_1}{w_2})^i\sum_{j\leq 0}(\frac{-z_2}{z_1})^j\\
&\quad + \frac{1}{2}c \sum _{i\leq 0}(\frac{-w_1}{z_2})^i\sum_{j\geq 1}(\frac{-w_2}{z_1})^j - \frac{1}{2} c \sum _{i\geq 1}(\frac{-w_1}{z_2})^i\sum_{j\leq 0}(\frac{-w_2}{z_1})^j\\
 &=-z_{2}E(z_1,w_2)\delta(w_1 +z_{2} )
 + w_2E(z_1,z_2)  \delta(w_1+w_2)  \\
  &\quad   + z_2E(w_1,w_2)z_1 \delta(z_1 +z_2 )
 - w_2E( w_1, z_2)\delta(z_1+w_2)   \\
&\quad +\frac{1}{2}c \sum _{i\geq 0}(\frac{-w_2}{z_1})^i\sum_{j\geq 1}(\frac{-z_2}{w_1})^j - \frac{1}{2} c \sum _{i\geq 1}(\frac{-z_1}{w_2})^i\sum_{j\geq 0}(\frac{-w_1}{z_2})^j\\
&\quad - \frac{1}{2}c \sum _{i\geq 0}(\frac{-z_2}{z_1})^i\sum_{j\geq 1}(\frac{-w_2}{w_1})^j  + \frac{1}{2} c \sum _{i\geq 1}(\frac{-z_1}{z_2})^i\sum_{j\geq 0}(\frac{-w_1}{w_2})^j \\
&\quad - \frac{1}{2}c \sum _{i\geq 0}(\frac{-w_2}{w_1})^i\sum_{j\geq 1}(\frac{-z_2}{z_1})^j +   \frac{1}{2} c \sum _{i\geq 1}(\frac{-w_1}{w_2})^i\sum_{j\geq 0}(\frac{-z_1}{z_2})^j\\
&\quad + \frac{1}{2}c \sum _{i\geq 0}(\frac{-z_2}{w_1})^i\sum_{j\geq 1}(\frac{-w_2}{z_1})^j - \frac{1}{2} c \sum _{i\geq 1}(\frac{-w_1}{z_2})^i\sum_{j\geq 0}(\frac{-z_1}{w_2})^j
\end{align*}
\begin{align*}
 &=-z_{2}E(z_1,w_2)\delta(w_1 +z_{2} )
 + w_2E(z_1,z_2)  \delta(w_1+w_2)  \\
  &\quad   + z_2E(w_1,w_2)z_1 \delta(z_1 +z_2 )
 - w_2E( w_1, z_2)\delta(z_1+w_2)   \\
&\quad +c \sum _{i\geq 1}(\frac{-w_2}{z_1})^i\sum_{j\geq 1}(\frac{-z_2}{w_1})^j - c \sum _{i\geq 1}(\frac{-z_1}{w_2})^i\sum_{j\geq 1}(\frac{-w_1}{z_2})^j\\
&\quad - c \sum _{i\geq 1}(\frac{-z_2}{z_1})^i\sum_{j\geq 1}(\frac{-w_2}{w_1})^j  +  c \sum _{i\geq 1}(\frac{-z_1}{z_2})^i\sum_{j\geq 1}(\frac{-w_1}{w_2})^j \\
&\quad + \frac{1}{2} c\sum_{j\geq 1}(\frac{-z_2}{w_1})^j - \frac{1}{2} c \sum _{i\geq 1}(\frac{-z_1}{w_2})^i - \frac{1}{2} c\sum_{j\geq 1}(\frac{-w_2}{w_1})^j  +\frac{1}{2} c \sum _{i\geq 1}(\frac{-z_1}{z_2})^i\\
&\quad - \frac{1}{2} c\sum_{j\geq 1}(\frac{-z_2}{z_1})^j +\frac{1}{2} c \sum _{i\geq 1}(\frac{-w_1}{w_2})^i+ \frac{1}{2} c\sum_{j\geq 1}(\frac{-w_2}{z_1})^j -\frac{1}{2} c \sum _{i\geq 1}(\frac{-w_1}{z_2})^i
\end{align*}
Finally, we have
\begin{align*}
i_{z, w}\frac{z-w}{z+w} &=1 +2\sum_{j\geq 1}(-1)^jw^{j}z^{-j} =1+2\sum_{j\geq 1} (\frac{-w}{z})^j\\
i_{w, z}\frac{w-z}{z+w} &=1 +2\sum_{j\geq 1}(-1)^jw^{-j}z^{j}=1+2\sum_{j\geq 1} (\frac{-z}{w})^j
\end{align*}
\end{proof}
\begin{prop}
Let the field $\phi ^B (z)$ and the vector space $\mathit{F_B}$ are as in Example \ref{example:B}. Define
\[
\Phi(z, w) =:\phi ^B(z) \phi ^B(w): -1 =\phi ^B(z) \phi ^B(w)-i_{z, w}\frac{z-w}{z+w}
\]
The assignment $E^B(z, w)\to \frac{1}{2}\Phi(z, w)$, \ $c\to Id_{\mathit{F_B}}$ gives a representation of the Lie algebra $b_{\infty}$.
\end{prop}
\begin{proof}
 First, we need to prove $\Phi(w, z)=-\Phi(z, w)$:
\begin{align*}
-\Phi(w, z)&= -\phi ^B(w) \phi ^B(z)+i_{w, z}\frac{-2z}{z+w}+1=\phi ^B(z) \phi ^B(w) +2w\delta (z+w) +i_{w, z}\frac{-2z}{z+w}+1 \\
&\quad =\phi ^B(z) \phi ^B(w) +i_{z, w}\frac{2w}{z+w} -i_{w, z}\frac{2w}{z+w} +i_{w, z}\frac{-2z}{z+w}+1 \\
&\quad =\phi ^B(z) \phi ^B(w) +i_{z, w}\frac{2w}{z+w} -i_{w, z}\frac{2z+2w}{z+w} +1 \\
&\quad =\phi ^B(z) \phi ^B(w) +i_{z, w}\frac{-2w}{z+w}-1=\Phi(z, w)
\end{align*}
For the commutation relations of $:\phi ^B(z) \phi ^B(w):$ one would be tempted to  use Wick's theorem, but Wick's theorem doesn't apply here, due to the fact that $\phi^B_0\phi^B_0=1$ (see Corollary \ref{fact:BnormOrd=1}), i.e., $\phi ^B(z)_+$ doesn't anticommute with itself, which is a requirement for Wick's theorem.
Hence we do direct computation using $\Phi(z, w): =\phi ^B(z) \phi ^B(w)-i_{z, w}\frac{z-w}{z+w}$:
\begin{align*}
&[\Phi(z_1, w_1), \Phi(z_2, w_2)]=[\phi ^B(z_1) \phi ^B(w_1), \phi ^B(z_2) \phi ^B(w_2)] \\
&=\phi ^B(z_1)\{\phi ^B(w_1), \phi ^B(z_2)\}\phi ^B(w_2)-\phi ^B(z_1)\phi ^B(z_2)\{\phi ^B(w_1), \phi ^B(w_2)\}  \\
&\quad +\{\phi ^B(z_1), \phi ^B(z_2)\}\phi ^B(w_2)\phi ^B(w_1) -\phi ^B(z_2)\{\phi ^B(z_1), \phi ^B(w_2)\}\phi ^B(w_1)\\
&=-2z_2\delta(w_1 +z_2)\phi ^B(z_1)\phi ^B(w_2) +2w_2\delta(w_1 +w_2)\phi ^B(z_1)\phi ^B(z_2)  \\
&\quad -2z_2\delta(z_1 +z_2)\phi ^B(w_2)\phi ^B(w_1) +2w_2\delta(z_1 +w_2)\phi ^B(z_2)\phi ^B(w_1)\\
&=-2z_2\delta(w_1 +z_2)\big(\Phi(z_1, w_2) + i_{z_1, w_2}\frac{z_1-w_2}{z_1+w_2}\big)+2w_2\delta(w_1 +w_2)\big(\Phi(z_1, z_2)+i_{z_1, z_2}\frac{z_1-z_2}{z_1+z_2}\big)  \\
&\quad -2z_2\delta(z_1 +z_2)\big(\Phi(w_2, w_1)+i_{w_2, w_1}\frac{w_2-w_1}{w_2+w_1}\big) +2w_2\delta(z_1 +w_2)\big(\Phi(z_2, w_1)+i_{z_2, w_1}\frac{z_2-w_1}{z_2+w_1}\big)\\
&=-2z_2\delta(w_1 +z_2)\Phi(z_1, w_2)+ 2w_2\delta(w_1 +w_2)\Phi(z_1, z_2)  \\
& \quad +2w_2\delta(z_1 +w_2)\Phi(z_2, w_1) -2z_2\delta(z_1 +z_2)\Phi(w_2, w_1) \\
&\quad +\Big(i_{w_1, z_2}\frac{w_1-z_2}{w_1+z_2}-i_{z_2, w_1}\frac{w_1-z_2}{w_1+z_2}\Big)i_{z_1, w_2}\frac{z_1-w_2}{z_1+w_2}  \\
&\quad -\Big(i_{w_1, w_2}\frac{w_1-w_2}{w_1+w_2}-i_{w_2, w_1}\frac{w_1-w_2}{w_1+w_2}\Big)i_{z_1, z_2}\frac{z_1-z_2}{z_1+z_2} \\
&\quad +\Big(i_{z_1, z_2}\frac{z_1-z_2}{z_1+z_2}-i_{z_2, z_1}\frac{z_1-z_2}{z_1+z_2}\Big)i_{w_2, w_1}\frac{w_2-w_1}{w_2+w_1} \\
&\quad -\Big(i_{z_1, w_2}\frac{z_1-w_2}{z_1+w_2}-i_{w_2, z_1}\frac{z_1-w_2}{z_1+w_2}\Big)i_{z_2, w_1}\frac{z_2-w_1}{z_2+w_1}\\
&=-2z_2\delta(w_1 +z_2)\Phi(z_1, w_2)+ 2w_2\delta(w_1 +w_2)\Phi(z_1, z_2)  \\
& \quad + 2w_2\delta(z_1 +w_2)\Phi(z_2, w_1) -2z_2\delta(z_1 +z_2)\Phi(w_2, w_1) \\
&\quad + i_{w_1, z_2}\frac{w_1-z_2}{w_1+z_2}i_{z_1, w_2}\frac{z_1-w_2}{z_1+w_2}  -i_{w_2, z_1}\frac{w_2-z_1}{z_1+w_2}i_{z_2, w_1}\frac{z_2-w_1}{z_2+w_1} \\
&\quad  - i_{w_1, w_2}\frac{w_1-w_2}{w_1+w_2}i_{z_1, z_2}\frac{z_1-z_2}{z_1+z_2}+i_{w_2, w_1}\frac{w_2-w_1}{w_1+w_2}i_{z_2, z_1}\frac{z_2-z_1}{z_1+z_2}\\
&=-2z_2\delta(w_1 +z_2)\Phi(z_1, w_2)+ 2w_2\delta(w_1 +w_2)\Phi(z_1, z_2)  \\
& \quad -2w_2\delta(z_1 +w_2)\Phi(w_1, z_2) +2z_2\delta(z_1 +z_2)\Phi(w_1, w_2) \\
&\quad + i_{w_1, z_2}\frac{w_1-z_2}{w_1+z_2}i_{z_1, w_2}\frac{z_1-w_2}{z_1+w_2}  -i_{w_2, z_1}\frac{w_2-z_1}{z_1+w_2}i_{z_2, w_1}\frac{z_2-w_1}{z_2+w_1} \\
&\quad -i_{w_1, w_2}\frac{w_1-w_2}{w_1+w_2}i_{z_1, z_2}\frac{z_1-z_2}{z_1+z_2}+i_{w_2, w_1}\frac{w_2-w_1}{w_1+w_2}i_{z_2, z_1}\frac{z_2-z_1}{z_1+z_2}
\end{align*}
 A direct comparison then shows that $\frac{1}{2} \Phi(z, w)$  indeed has the required commutation relations of $E^B(z, w)$.\\
\end{proof}

\subsection{The infinite dimensional Lie algebra $c_{\infty}$} \ \\
The infinite dimensional Lie algebra $\bar{c}_{\infty}$ is the subalgebra of $\bar{a}_{\infty}$ consisting of the infinite matrices preserving the bilinear form $C(v_i, v_j)=(-1)^i \delta_{i, 1-j}$, i.e.,
\begin{equation}
\bar{c}_{\infty}=\{(a_{ij})\in \bar{a}_{\infty} | \ a_{ij}=(-1)^{i+j-1}a_{1-j, 1-i} \}.
\end{equation}
The algebra $c_{\infty}$ is a central extension of
$\bar{c}_{\infty}$ by a central element $c$, $c_{\infty}=\bar{c}_{\infty}\oplus \mathbb{C} c$, where $c$ is the cocycle  for $a_{\infty}$, \eqref{equation:cocycle-a} (see \cite{Kac-Lie}). And the  commutation relations for the elementary matrices in $c_\infty$ is
\begin{align*}
[E_{ij},E_{kl}]=\delta_{jk}E_{il}-\delta_{li}E_{kj}+C(E_{ij},E_{kl})c.
\end{align*}
The generators for the algebra $c_\infty$ can be written in terms of these elementary matrices as:
\[
\{ (-1)^j E_{i, j} -(-1)^i E_{1-j, 1-i}, \  i, j \in \mathbb{Z}; \text{and} \ \ c\}.
\]
We can arrange the non-central generators in a generating series
\begin{equation}
E^C(z, w) =\sum _{i, j\in \mathbb{Z}} ((-1)^jE_{ij}-(-1)^iE_{1-j, 1-i})z^{i-1}w^{-j}.
\end{equation}
\begin{lem}
\label{lem:generseriesC}
The generating series  $E^C(z,w)$ obeys the following relations:
\[
E^C(z,w) = E^C(w,z)
\]
and
\begin{align*}
[E^C(z_1, w_1), &E^C(z_2, w_2)]=  E^C(z_1, w_2)\delta(z_2 + w_1) -E^C(z_2, w_1)\delta(z_1+ w_2) \\
&\hspace{1.7cm} - E^C(w_2,w_1)\delta(z_1 + z_2) + E^C(z_1,z_2)\delta(w_2 + w_1)   \\
&+ 2 \iota_{z_1,w_2} \frac{1}{z_1+ w_2} \iota_{w_1,z_2} \frac{1}{w_1  + z_2}c
- 2 \iota_{ w_2, z_1}\frac{1}{ w_2 + z_1} \iota_{z_2,w_1} \frac{1}{z_2 + w_1 }c \\
& + 2 \iota_{z_1,z_2 } \frac{1}{   z_1+ z_2} \iota_{w_1,w_2 } \frac{1}{w_1 + w_2 }c
-  2 \iota_{ z_2,z_1} \frac{1}{   z_2+ z_1} \iota_{ w_2,w_1} \frac{1}{w_2 + w_1 }c.
\end{align*}
\end{lem}
\begin{proof}
We present the proof here for completeness.
The first property is very easy to check. For the second property, we have:
\begin{align*}
&[E^C(z_1, w_1),E^C(z_2, w_2)] \\
& = \sum _{i, j,k,l\in \mathbb{Z}} [((-1)^jE_{ij}-(-1)^iE_{1-j, 1-i}),((-1)^lE_{kl}-(-1)^kE_{1-l, 1-k})]z_1^{i-1}z_2^{k-1}w^{-j}_1w^{-l}_2 \\
& =  \sum _{i, j,k,l\in \mathbb{Z}} (-1)^{j+l}[E_{ij} ,E_{kl}]z_1^{i-1}z_2^{k-1}w^{-j}_1w^{-l}_2  \\
&\quad - \sum _{i, j,k,l\in \mathbb{Z}} (-1)^{j+k}[E_{ij} , E_{1-l, 1-k}]z_1^{i-1}z_2^{k-1}w^{-j}_1w^{-l}_2   \\
&\quad - \sum _{i, j,k,l\in \mathbb{Z}}(-1)^{i+l} [E_{1-j, 1-i},E_{kl}]z_1^{i-1}z_2^{k-1}w^{-j}_1w^{-l}_2  \\
&\quad + \sum _{i, j,k,l\in \mathbb{Z}}(-1)^{i+k} [ E_{1-j, 1-i}, E_{1-l, 1-k}]z_1^{i-1}z_2^{k-1}w^{-j}_1w^{-l}_2 \\
\end{align*}
\begin{align*}
&=\sum _{i, j,l\in \mathbb{Z}}  (-1)^{j+l}E_{il} z_1^{i-1}z_2^{j-1}w^{-j}_1w^{-l}_2-\sum _{i, j,k\in \mathbb{Z}}  (-1)^{j+i}E_{kj}z_1^{i-1}z_2^{k-1}w^{-j}_1w^{-i}_2  \\
&\quad - \sum _{i, j,k \in \mathbb{Z}} (-1)^{j+k}E_{i,1-k} z_1^{i-1}z_2^{k-1}w^{-j}_1w^{j-1}_2 +
 \sum _{i, j,l \in \mathbb{Z}}  (-1)^{j+1-i} E_{1-l, j} z_1^{i-1}z_2^{-i}w^{-j}_1w^{-l}_2   \\
&\quad - \sum _{i, j,l\in \mathbb{Z}} (-1)^{i+l}E_{1-j,l} z_1^{i-1}z_2^{-i}w^{-j}_1w^{-l}_2 +
 \sum _{i, j,k \in \mathbb{Z}}  (-1)^{i+1-j}E_{k,1-i}  z_1^{i-1}z_2^{k-1}w^{-j}_1w^{j-1}_2   \\
&\quad + \sum _{i, j,k\in \mathbb{Z}} (-1)^{i+k}E_{1-j,1-k}z_1^{i-1}z_2^{k-1}w^{-j}_1w^{-i}_2  -   \sum _{i, k,l\in \mathbb{Z}}  (-1)^{i+k}E_{1-l, 1-i} z_1^{i-1}z_2^{k-1}w^{-k}_1w^{-l}_2  \\
&\quad +   2\sum_{l \leq 0 , j \geq 1} (-1)^{j+ l}    z_1^{l-1} z_2^{j-1} w_1^{-j} w_2^{-l}c - 2\sum_{l \geq 1 , j\leq 0} (-1)^{j+ l }   z_1^{l-1} z_2^{j-1} w_1^{-j} w_2^{-l} c \\
&\quad - 2 \sum_{k \geq 1, j \geq 1} (-1)^{j+k}  z_1^{-k} z_2^{k-1} w_1^{-j} w_2^{1-j}c + 2  \sum_{k \leq 0, j \leq 0} (-1)^{j+k}  z_1^{-k} z_2^{k-1} w_1^{-j} w_2^{j-1} c  \\
& =  E^C(z_1, w_2)\delta(z_2 + w_1) -E^C(z_2, w_1)\delta(z_1+ w_2)\\
&\quad - E^C(w_2,w_1)\delta(z_1 + z_2) + E^C(z_1,z_2)\delta(w_2 + w_1)   \\
&\quad -2 \iota_{z_1,w_2} \frac{1}{z_1+ w_2} \iota_{w_1,z_2} \frac{1}{w_1  + z_2}c
+2 \iota_{ w_2, z_1}\frac{1}{ w_2 + z_1} \iota_{z_2,w_1} \frac{1}{z_2 + w_1 }c\\
&\quad - 2 \iota_{z_1,z_2 } \frac{1}{   z_1+ z_2} \iota_{w_1,w_2 } \frac{1}{w_1 + w_2 }c
+2 \iota_{ z_2,z_1} \frac{1}{   z_2+ z_1} \iota_{ w_2,w_1} \frac{1}{w_2 + w_1 }c.
\end{align*}
\end{proof}
\begin{prop}
Recall the field $\phi ^C (z)$ on the highest weight module $\mathit{F_C}$ from Example \ref{example:C}
The assignment $E(z, w)\to i:\phi^C(z)\phi^C (w):$, \ $c\to \frac{1}{2}Id_{\mathit{F_C}}$ gives a representation of the Lie algebra $c_{\infty}$. (Here $i$ is the imaginary unit.)
\end{prop}
\begin{proof}
We will use Wick's theorem to prove this proposition:  First we check
\begin{align*}
[\phi^C_+(z),\phi^C_+(w)]&=\sum_{m\geq 1/2,n\geq 1/2,m,n\in\frac{1}{2}+\mathbb Z}[\phi^C_m,\phi^C_n]z^{m-1/2}w^{n-1/2} \\
&=\sum_{m\geq 1/2,n\geq 1/2,m,n\in\frac{1}{2}+\mathbb Z}(-1)^{n-1/2}\delta_{m,-n}z^{m-1/2}w^{n-1/2}=0, \\
[\phi^C_-(z),\phi^C_-(w)]&=\sum_{m<1/2,n< 1/2,m,n\in\frac{1}{2}+\mathbb Z}[\phi^C_m,\phi^C_n]z^{-m-1}w^{-n-1} \\
&=\sum_{m<1/2,n< 1/2,m,n\in\frac{1}{2}+\mathbb Z}(-1)^{n-1/2}\delta_{m,-n}z^{m-1/2}w^{n-1/2}=0.
\end{align*}
 As a direct application of Wick's Theorem we have
\begin{align*}
:\phi ^C(z_1)&\phi ^C(w_1): :\phi ^C(z_2)\phi ^C(w_2):   \\
&=\lfloor \phi ^C(z_1)\phi ^C(z_2)\rfloor:\phi ^C(w_1)\phi ^C(w_2): +\lfloor \phi ^C(z_1)\phi ^C(w_2)\rfloor:\phi ^C(z_2)\phi ^C(w_1): \\
&\quad+\lfloor \phi ^C(w_1)\phi ^C(z_2)\rfloor:\phi ^C(z_1)\phi ^C(w_2):+\lfloor \phi ^C(w_1)\phi ^C(w_2)\rfloor:\phi ^C(z_1)\phi ^C(z_2): \\
&\quad +\lfloor\phi ^C(z_1)\phi ^C(z_2)\rfloor\lfloor \phi ^C(w_1)\phi ^C(w_2)\rfloor +\lfloor\phi ^C(z_1)\phi ^C(w_2)\rfloor\lfloor \phi ^C(w_1)\phi ^C(z_2)\rfloor
\end{align*}
and hence
\begin{align*}
[i &:\phi ^C(z_1)\phi ^C(w_1):, \ i:\phi ^C(z_2)\phi ^C(w_2)]  \\
&=-\delta(z_1+z_2):\phi ^C(w_1)\phi ^C(w_2): - \delta(z_1+w_2) :\phi ^C(z_1)\phi ^C(w_2): \\
&\quad+\delta( z_2+ w_1) :\phi ^C(z_1)\phi ^C(w_2):+\delta( w_2 +w_1) :\phi ^C(z_1)\phi ^C(z_2): \\
&\quad -\iota_{z_1,z_2}\left(\frac{1}{z_1+z_2}\right)\iota_{w_1,w_2}\left(\frac{1}{w_1+w_2}\right) -\iota_{z_1,w_2}\left(\frac{1}{z_1+w_2}\right)\iota_{w_1,z_2}\left(\frac{1}{w_1+z_2}\right)  \\
&\quad +\iota_{z_2,z_1}\left(\frac{1}{z_2+z_1}\right)\iota_{w_2,w_1}\left(\frac{1}{w_2+w_1}\right) +\iota_{z_2,w_1}\left(\frac{1}{z_2+w_1}\right)\iota_{w_2,z_1}\left(\frac{1}{w_2+z_1}\right).
\end{align*}
Note that we used (the rather peculiar) fact that for $\lambda =-1$ we have $\delta(z_1+z_2)=\delta(z_1, -z_2)=-\delta(z_2+z_1)=-\delta(z_2, -z_1)$.
\end{proof}

\subsection{The infinite dimensional Lie algebra $d_{\infty}$} \ \\
The infinite dimensional Lie algebra $\bar{d}_{\infty}$ is the subalgebra of $\bar{a}_{\infty}$ consisting of the infinite matrices preserving the bilinear form $D(v_i, v_j)=\delta_{i, 1-j}$, i.e.,
\begin{equation}
\bar{d}_{\infty}=\{(a_{ij})\in \bar{a}_{\infty} | \ a_{ij}=-a_{1-j, 1-i} \}.
\end{equation}
Denote by $d_{\infty}$ the  central extension of $\bar{d}_{\infty}$ by a central element $c$, $d_{\infty}=\bar{d}_{\infty}\oplus \mathbb{C} c$, with the same cocycle as for $a_{\infty}$, \eqref{equation:cocycle-a}.  The commutation relations for the elementary matrices in $d_\infty$ is
\begin{align*}
[E_{ij},E_{kl}]=\delta_{jk}E_{il}-\delta_{li}E_{kj}+\frac{1}{2}C(E_{ij},E_{kl})c.
\end{align*}
The generators for the algebra $d_\infty$ can be written in terms of these elementary matrices as:
\[
\{ E_{i, j} - E_{1-j, 1-i}, \  i, j \in \mathbb{Z}; \text{and} \ \ c\}.
\]
We can arrange the non-central generators in a generating series
\begin{equation}
E^D(z, w) =\sum _{i, j\in \mathbb{Z}} (E_{ij}-E_{1-j, 1-i})z^{i-1}w^{-j}.
\end{equation}
\begin{lem}
\label{lem:generseriesD}
The generating series  $E^D(z,w)$ obeys the following relations:
\[
E^D(z,w) = -E^D(w,z)
\]
and
\begin{align*}
[E^D(z_1, w_1), &E^D(z_2, w_2)]
 =  E^D(z_1, w_2)\delta(z_2-w_1) -E^D(z_2, w_1)\delta(z_1-w_2) \\
&\hspace{1.7cm} +E^D(w_2,w_1)\delta(z_1-z_2) -E^D(z_1,z_2)\delta(w_1-w_2)  \\
&\quad +\iota_{z_1,w_2}\frac{1}{z_1-w_2}\iota_{w_1,z_2}\frac{1}{w_1-z_2}c   -\iota_{z_2,w_1}\frac{1}{z_2-w_1}\iota_{w_2,z_1}\frac{1}{w_2-z_1}c  \\
&\quad -\iota_{z_1,z_2}\frac{1}{z_1-z_2}\iota_{w_1,w_2}\frac{1}{w_1-w_2}c +\iota_{z_2,z_1}\frac{1}{z_1-z_2}\iota_{w_2,w_1}\frac{1}{w_2-w_1}c.
\end{align*}
\end{lem}
\begin{proof}
We present the proof here for completeness. Again, the first property is very easy to check. For the second property, we have:
\begin{align*}
&[E^D(z_1, w_1),  E^D(z_2, w_2)]=\sum _{i, j,k,l\in \mathbb{Z}} [(E_{ij}-E_{1-j, 1-i}),(E_{kl}-E_{1-l, 1-k})]z_1^{i-1}z_2^{k-1}w^{-j}_1w^{-l}_2 \\
&=\sum _{i, j,k,l\in \mathbb{Z}} [E_{ij} ,E_{kl}]z_1^{i-1}z_2^{k-1}w^{-j}_1w^{-l}_2  - \sum _{i, j,k,l\in \mathbb{Z}} [E_{ij} , E_{1-l, 1-k}]z_1^{i-1}z_2^{k-1}w^{-j}_1w^{-l}_2   \\
& \quad - \sum _{i, j,k,l\in \mathbb{Z}} [E_{1-j, 1-i},E_{kl}]z_1^{i-1}z_2^{k-1}w^{-j}_1w^{-l}_2   + \sum _{i, j,k,l\in \mathbb{Z}} [ E_{1-j, 1-i}, E_{1-l, 1-k}]z_1^{i-1}z_2^{k-1}w^{-j}_1w^{-l}_2 \\
&=\sum _{i, j,l\in \mathbb{Z}}  E_{il} z_1^{i-1}z_2^{j-1}w^{-j}_1w^{-l}_2-\sum _{i, j,k\in \mathbb{Z}}  E_{kj}z_1^{i-1}z_2^{k-1}w^{-j}_1w^{-i}_2  \\
&\quad - \sum _{i, j,k \in \mathbb{Z}} E_{i,1-k} z_1^{i-1}z_2^{k-1}w^{-j}_1w^{j-1}_2 +
 \sum _{i, j,l \in \mathbb{Z}}   E_{1-l, j} z_1^{i-1}z_2^{-i}w^{-j}_1w^{-l}_2   \\
&\quad - \sum _{i, j,l\in \mathbb{Z}} E_{1-j,l} z_1^{i-1}z_2^{-i}w^{-j}_1w^{-l}_2 +
 \sum _{i, j,k \in \mathbb{Z}}  E_{k,1-i}  z_1^{i-1}z_2^{k-1}w^{-j}_1w^{j-1}_2   \\
&\quad - \sum _{i, j,l\in \mathbb{Z}} E_{1-j,l} z_1^{i-1}z_2^{-i}w^{-j}_1w^{-l}_2 +
 \sum _{i, j,k \in \mathbb{Z}}  E_{k,1-i}  z_1^{i-1}z_2^{k-1}w^{-j}_1w^{j-1}_2   \\
&\quad + \sum _{i, j,k \in \mathbb{Z}} E_{1-j,1-k}z_1^{i-1}z_2^{k-1}w^{-j}_1w^{-i}_2
-   \sum _{i, j,l\in \mathbb{Z}}  E_{1-l, 1-i} z_1^{i-1}z_2^{j-1}w^{-j}_1w^{-l}_2 \\
&\quad + \sum _{i\leq 0, j\geq 1} z_1^{i-1}z_2^{j-1}w^{-j}_1w^{-i}_2 c- \sum _{j\leq 0, i\geq 1} z_1^{i-1}z_2^{j-1}w^{-j}_1w^{-i}_2 c \\
 &\quad -  \sum _{i\leq 0, j\geq 1}z_1^{i-1}z_2^{-i}w^{-j}_1w^{j-1}_2 c+ \sum _{i\geq 1,j\leq 0}z_1^{i-1}z_2^{-i}w^{-j}_1w^{j-1}_2 c \\
&=  E^D(z_1, w_2)\delta(z_2-w_1) -E^D(z_2, w_1)\delta(z_1-w_2) \\
&\quad +E^D(w_2,w_1)\delta(z_1-z_2) -E^D(z_1,z_2)\delta(w_1-w_2) \\
&\quad +\iota_{z_1,w_2}\left(\frac{1}{z_1-w_2}\right)\iota_{w_1,z_2}\left(\frac{1}{w_1-z_2}\right) c -\iota_{z_2,w_1}\left(\frac{1}{z_2-w_1}\right)\iota_{w_2,z_1}\left(\frac{1}{w_2-z_1}\right)c \\
&\quad -\iota_{z_1,z_2}\left(\frac{1}{z_1-z_2}\right)\iota_{w_1,w_2}\left(\frac{1}{w_1-w_2}\right) c +\iota_{z_2,z_1}\left(\frac{1}{z_2-z_1}\right)\iota_{w_2,w_1}\left(\frac{1}{w_2-w_1}\right) c.
\end{align*}
\end{proof}
\begin{prop}
Recall the field $\phi ^D (z)$ on the highest weight module $\mathit{F_D}$ from Example \ref{example:D}
The assignment $E(z, w)\to :\phi^D(z)\phi^D (w):$, \ $c\to Id_{\mathit{F_D}}$ gives a representation of the Lie algebra $d_{\infty}$.
\end{prop}
\begin{proof}
We will use Wick's theorem to prove this lemma. The conditions of Wick's theorem are  satisfied (similarly to the case of $c_{\infty}$ case),  and
 we have
\begin{align*}
:\phi ^D(z_1)& \phi ^D(w_1): :\phi ^D(z_2)\phi ^D(w_2):  \\
&=-\lfloor \phi ^D(z_1)\phi ^D(z_2)\rfloor:\phi ^D(w_1)\phi ^D(w_2): +\lfloor \phi ^D(z_1)\phi ^D(w_2)\rfloor:\phi ^D(w_1)\phi ^D(z_2): \\
&\quad+\lfloor \phi ^D(w_1)\phi ^D(z_2)\rfloor:\phi ^D(z_1)\phi ^D(w_2):-\lfloor \phi ^D(w_1)\phi ^D(w_2)\rfloor:\phi ^D(z_1)\phi ^D(z_2): \\
&\quad -\lfloor\phi ^D(z_1)\phi ^D(z_2)\rfloor\lfloor \phi ^D(w_1)\phi ^D(w_2)\rfloor +\lfloor\phi ^D(z_1)\phi ^D(w_2)\rfloor\lfloor \phi ^D(w_1)\phi ^D(z_2)\rfloor
\end{align*}
and hence
\begin{align*}
[:& \phi ^D(z_1)\phi ^D(w_1):,:\phi ^D(z_2)\phi ^D(w_2)] \\
&=-\delta(z_1-z_2):\phi ^D(w_1)\phi ^D(w_2): +\delta(z_1-w_2):\phi ^D(w_1)\phi ^D(z_2): \\
&\quad+\delta(w_1-z_2):\phi ^D(z_1)\phi ^D(w_2):-\delta(w_1-w_2):\phi ^D(z_1)\phi ^D(z_2):  \\
&\quad -\iota_{z_1,z_2}\left(\frac{1}{z_1-z_2}\right)\iota_{w_1,w_2}\left(\frac{1}{w_1-w_2}\right) +\iota_{z_1,w_2}\left(\frac{1}{z_1-w_2}\right)\iota_{w_1,z_2}\left(\frac{1}{w_1-z_2}\right)  \\
&\quad +\iota_{z_2,z_1}\left(\frac{1}{z_2-z_1}\right)\iota_{w_2,w_1}\left(\frac{1}{w_2-w_1}\right) -\iota_{z_2,w_1}\left(\frac{1}{z_2-w_1}\right)\iota_{w_2,z_1}\left(\frac{1}{w_2-z_1}\right).
\end{align*}
\end{proof}

\section{Twisted Vertex Algebras}
\label{sectio:tva}

In this section we finally answer the following question: if we start with a given number of generating fields $a^i(z)$ (see a precise definition at \ref{defn:GenerFields}) , which are local at several but finitely many points $\lambda_1, \lambda_2, \dots, \lambda_N$ (also satisfying  some auxiliary conditions), and we allow the operations of differentiation ($\partial_z a^i(z)$), substitutions at the locality points ($a(\lambda _j z)$), normal ordered products ($:a^i (\lambda_k z)a_j (\lambda _l z):$) and OPE coefficients on all those fields and their descendants; and we further assume that all those descendants are still only local at $\lambda_1, \lambda_2, \dots, \lambda_N$; then what is the resulting structure that we get? The first observation is, as we proved in \secref{section:NormalOrdProd}, that in order for the descendants to stay local at {\bf finitely} many points, then the points of locality have to form a multiplicative group. The  only finite multiplicative subgroups of the complex plane $\mathbb{C}$ are the roots of unity cyclic groups. Hence if we want to have the {\bf total number of points of locality stay finite}, the only possible choice for the $N$ points of locality are the $N$-th roots of unity.  Then the answer to our ``generating descendants" problem is: the fields $a^i(z)$ will generate  a twisted vertex algebra. Twisted vertex algebras of general order $N$ were defined in \cite{AngTVA} (of order 2 in \cite{Ang-Varna2}) in order to answer the question: what are the  boson-fermion correspondences of type B, C and D-A? These correspondences  must be isomorphisms between some type of structures, and we wanted to understand those structures. Although the boson-fermion correspondence of type A is an isomorphism of super vertex algebras (which are 1-point local), the boson-fermion correspondences of type B, C and D-A  are  isomorphisms of twisted vertex algebras (which are at least 2-point local). Before giving the definition of a twisted vertex algebra, we want to mention one other very important difference between super vertex algebras and twisted vertex algebras, and it is that the space of fields in a twisted vertex algebra is larger than the space of states (unlike in super vertex algebras where there is a state-field bijection); and by definition there is a strict surjective projection from the space of fields to the space of states in a twisted vertex algebra. In this, twisted vertex algebras resemble the concept of a deformed chiral algebras from \cite{FR}, and differ from the concept of a $\Gamma$-vertex algebra of \cite{Li2}. The $\Gamma$-vertex algebras were introduced in  \cite{Li2}, and  are indeed multilocal, but  the number of fields included in a $\Gamma$-vertex algebra is not enough to describe the examples of the boson-fermion correspondences for which the concept of a twisted vertex algebra was introduced. Namely, there are  fields which are indispensable to the examples, like the Heisenberg-producing  fields $:\phi ^B(z)\phi ^B(-z):$ and $:\phi ^D(z)\phi ^D(-z):$ from \secref{section: examples}, which cannot be incorporated in a $\Gamma$-vertex algebra of \cite{Li2} (see \secref{section:appendix} for a proof). This is  mainly due to the fact that such nontrivial fields like $:\phi ^D(z)\phi ^D(-z):$  are associated to elements with projection 0, see \cite{AngTVA}. With that said, in this section we will  recall the definition of twisted vertex algebra from \cite{AngTVA}.

We will work with the category of super vector spaces, i.e., $\mathbb  {Z}_{2}$ graded vector spaces. The flip map $\tau $  is defined by
\begin{equation}
\label{eq:flip}
 \tau (a\ten b) =(-1)^{\tilde{a} \cdot \tilde{b}} (b\ten a)
\end{equation}
for any homogeneous elements $a, b$ in the super vector space, where $\tilde{a}$ or $p(a)$, $\tilde{b}$ or $p(b)$  denote correspondingly  the parity of $a$, $b$.
\begin{defn}\begin{bf}(The Hopf algebra $H_D=\mathbb{C} [D]$)\end{bf}\\
The Hopf algebra $H_D=\mathbb{C} [D]$ is the universal enveloping algebra of the commutative one-dimensional Lie algebra with \emph{even} generator $D$, i.e., the polynomial algebra with a primitive generator $D$.
We have
\begin{equation}
\label{eq:D^n}
\del D^{(n)}=\sum _{k+l=n} D^{(k)}\ten D^{(l)}.
\end{equation}
\end{defn}
\begin{defn}\begin{bf}(The Hopf algebra $H^N_{T_{\epsilon}}$)\end{bf}\\
Let $\epsilon$ be  a primitive root of unity of order $N$.
The Hopf algebra $H^N_{T_{\epsilon}}$ is the Hopf algebra with  a primitive generator $D$ and a grouplike generator $T_{\epsilon}$ subject to the  relations:
\begin{equation}
DT_{\epsilon }=\epsilon T_{\epsilon } D, \quad \text{and} \ (T_{\epsilon })^n=1
\end{equation}
$H^N_{T_{\epsilon}}$ has $H_D$ as a Hopf subalgebra. Both $H_D$ and $H^N_{T_{\epsilon}}$ are  entirely even.
\end{defn}
\color{black}
Recall from \secref{section:NormalOrdProd} the space  $\mathbf{F}^N_{\epsilon}(z, w)$-- the space of rational functions in the  formal variables $z, w$ with only poles at $z=0, w=0, \ z= \epsilon^i w$, $i=1, \dots , N$; and  $\mathbf{F}^N_{\epsilon}(z, w)^{+, w}$-- the space of rational functions in the formal variables $z, w$ with only poles at $z=0,  \ z= \epsilon^i w$, $i=1, \dots , N$ (no  pole at $w=0$ in $\mathbf{F}^N_{\epsilon}(z, w)^{+, w}$).
Similarly, $\mathbf{F}^N_{\epsilon}(z_1, z_2, \dots , z_l)$ was the space of rational functions in variables $z_1, z_2, \dots, z_l$ with only poles at $z_m=0$, $m=1,\dots, l$, or at $z_j= \epsilon^{i_j} z_k$, for fixed $k$ and $j=1,\dots, l$, $i_j=1, \dots , N$.
 If $N$ is clear from the context, or the property doesn't depend on the particular value of $N$, we will just write  $\mathbf{F}_{\epsilon}(z, w)$ or $\mathbf{F}_{\epsilon}(z, w)^{+, w}$.

\begin{prop}
$\mathbf{F}^N_{\epsilon}(z, w)$  (and $\mathbf{F}_{\epsilon}(z, w)^{+, w}$) are  $H^N_{T_{\epsilon}}\ten H^N_{T_{\epsilon}}$ (and consequently an $H_D\ten H_D$)  Hopf modules by
\begin{align}
&D_z f(z, w)=\partial_z f(z, w), \quad (T_{\epsilon}) _{z} f(z, w)=f(\epsilon z, w) \\
&D_w f(z, w)=\partial_w f(z, w), \quad (T_{\epsilon}) _{w} f(z, w)=f(z, \epsilon  w)
\end{align}
We will denote the action of elements $h\ten 1 \in H^N_{T_{\epsilon}}\ten H^N_{T_{\epsilon}}$  on $\mathbf{F}_{\epsilon}(z, w)$ (and $\mathbf{F}_{\epsilon}(z, w)^{+, w}$) by $h_z\cdot$, correspondingly $h_w\cdot$ will denote the action of the elements $1\ten h \in H^N_{T_{\epsilon}}\ten H^N_{T_{\epsilon}}$.
\end{prop}
\color{black}
\begin{defn}\label{defn:twistedVA} \begin{bf}(Twisted vertex algebra of order $N$)\end{bf}
A twisted vertex algebra of order $N$ is a collection of  the following data $(V, W, \pi_{f}, Y)$:
\begin{itemize}
\item  a super vector space $V$, which is an $H^N_{T_{\epsilon}}$ module,  graded as an  $H_D$-module, called the space of fields;
\item  a super vector space $W\subset V$, called the space of states;
\item a linear surjective projection map $\pi_{f}:V\to W$, such that $\pi_{f}\arrowvert _W=Id_W$
\item a linear map $Y$ from $V$ to the space of fields on $W$, called  the vertex map;
\item  a distinguished vector, called the vacuum vector, denoted $1=|0\rangle \in W\subset V$.
\end{itemize}
Satisfying the following set of axioms:
\begin{itemize}
\item vacuum axiom: \ \ $Y(1, z)=Id_W$;
\item modified creation axiom: \ \ $Y(a, z)|0 \rangle \arrowvert _{z=0}=\pi_f(a)$, for any $a\in V$;
\item transfer of action:\ \  $Y(ha,z)=h_z\cdot Y(a, z)$ for any $h\in H^N_{T_{\epsilon}}$;
\item analytic continuation: For any $a, b, c\in V$ exists \\ $X_{z, w, 0}(a\ten b\ten c)\in W[\!]z, w]\!]\ten \mathbf{F}_{\epsilon}(z, w)$ such that
\begin{equation*}
Y(a, z)Y(b, w)\pi_{f}(c)=i_{z, w} X_{z, w,  0}(a\ten b\ten  c)
\end{equation*}
\item symmetry: $X_{z, w, 0}(a\ten b\ten c)=X_{w, z, 0}(\tau (a\ten b)\ten c)$ where $\tau:V\otimes V\to V$ is defined by $\tau(a\otimes b)=(-1)^{p(a)p(b)}b\otimes a$ for $a$ and $b$ homogeneous;
\item Completeness with respect to Operator Product Expansions (OPE's): For each $i\in 0, 1, \dots, N-1$, \ $k\in \mathbb{Z}$, any $a, b, c\in V$,  $a, b$--homogeneous with respect to the grading by $D$, there exist $l_k\in \mathbb{Z}$ such that
\begin{equation}
Res_{z=\epsilon ^i w}X_{z, w, 0}(a\ten b\ten c)(z-\epsilon ^i w)^k =\sum_s^{\text{finite}} w^{l^s_{k,i}}  Y(v^s_{k, i}, w)\pi_{f}(c)
\end{equation}
for some \ homogeneous \ elements  $v^s_{k, i}\in V, \ \ l^s_{k,i}\in \mathbb{Z}$.
\end{itemize}
\end{defn}
\begin{remark} If $V$ is an (ordinary) super vertex algebra, then the data $(V, V, \pi_{f}=Id_V, Y)$ is a twisted vertex algebra of order 1.
\end{remark}
\begin{remark} Any vertex algebra, be it super, twisted or quantum,  is first a  collection of fields, as all these objects may be used to model various chiral quantum field theories. In  an (ordinary) super vertex algebra, the vertex map $Y$ is a bijection from  the vector space of states $W$ to  the specific collection of  fields on $W$ (which collection in fact constitutes the vertex algebra). Thus we can speak interchangeably about the (ordinary) super vertex algebra as both the collection of fields $Y(v, z), \   v\in W$, constituting the super vertex algebra,  and the vector space $W$.  For twisted vertex algebras this is not the case: the collection of fields $V$ that constitutes the twisted vertex algebra is usually strictly larger that the space of states $W$ on which those fields act. Thus, unlike for the ordinary super vertex algebras, the map $Y$ is a map from the super space $V$ (which really is the collection of fields constituting the twisted vertex algebra) to another, bijective, collection of fields. The reason for the necessity of such a map $Y$, is that the OPE coefficients are always fields in the collection $V$, but are often not quite vertex operators, they need a ``shift", see the next remark. In this paper the space of fields $V$ is really just an abstract vector space (given by the collection of fields), which is given the structure of a vector super space by the parity of each field, and is  an $H^N_{T_{\epsilon}}$ module in the obvious way. In \cite{AngTVA} the spaces of fields $V$ for all the examples are given additional structure, in fact  even a Hopf algebra structure.
\end{remark}
\begin{remark} (\textbf{Shift restriction})
\label{remark:shift}
The axiom/property  requiring completeness with respect to the Operator Product Expansions (OPE's) is a weaker one than in the classical vertex algebra case. We can express this weaker axiom as follows.  Any field $v(w)$  is characterized by first, the doubly infinite sequence of the operators representing its modes (see remark \ref{remark:fieldindexing}); and second, by the indexing of that doubly infinite sequence. One can shift the indexing of  this  sequence (i.e, place the 0 index at  different modes), and each shift in the indexing of the sequence of modes  corresponds to a multiplication by an integer  power of $w$ of the field $v(w)$.
For example,  the OPE coefficient in \eqref{eqn:OPE-B} of example \ref{example:B} is $-2w=-2wId_{\mathit{F_B}}$, (recall $\phi^B(z)\phi^B (w)\sim \frac{-2w}{z+w}$). As a doubly-infinite sequence of modes, this is the sequence
\[
(\dots, 0, \dots, 0, 0\arrowvert_{0\ \text{index}}, -2, 0, 0, 0, \dots ).
\]
The distribution  $-2w$ is a field, but is not a vertex operator in any vertex algebra (it cannot be a vertex operator as indexed for many well known reasons,  for example if it was a vertex operator it  would contradict creation and the $D$-invariance). But if we shift the place of the 0-index,  to
\[
(\dots, 0, \dots, 0, 0, -2\arrowvert_{0\ \text{index}}, 0, 0, 0, \dots ),
\]
this now represents  the field $-2=-2Id_{\mathit{F_B}}$, which is the vertex operator assigned to the element $-2 |0\rangle$. In other words, if we disregard the indexing place, but just look at the doubly-infinite sequence, the sequence of modes
\[
(\dots, 0, \dots, 0, 0, -2, 0, 0, 0, \dots )
\]
is represented in the vertex algebra. The changing of the indexing place corresponds to multiplying the original field $-2w=-2wId_{\mathit{F_B}}$ by a factor of $w^{-1}$ in this case. This is then the role of the vertex map $Y$: in this case the map $Y$ maps the field $-2wId_W\in V$ to the vertex operator $Y(-2 |0\rangle, w)=-2Id_W$ (here $W=\mathit{F_B}$). The shift restriction that the completeness with respect to OPE's axiom implies is that OPE's such as
$a(z)a(w)\sim \frac{-2w+w^2}{z+w}$, for example, are not allowed, as we wouldn't be able to re-index the sequence of modes. In any physical model of quantum field theory in two dimensions it is very important that we are able to recover the doubly-infinite sequence of operators corresponding to different ``excitations". That is what the modified completeness with respect to OPE's ensures: that if we consider a product of fields (``excitation") we can {\bf recover} the vertex operators to which the  doubly infinite sequence of this products of fields corresponds to.
In an (ordinary)  super vertex algebra we have a stronger property, there the products of fields automatically  are vertex operators  in the same vertex algebra. This stronger property cannot hold in the certain interesting examples, as one clearly sees in Example \ref{example:B}.

  The ``modified completeness with respect to  OPE's" axiom is in some sense the weakest requirement one can impose, as in a physical theory on the one hand one needs to ``close the algebra", but on the other hand it is the sequence of modes that is important, not so much its indexing.   The axiom then requires that the sequence of modes is ``in" the twisted vertex algebra, albeit after potential reindexing. The vertex map $Y$ is what accomplishes this reindexing.
\end{remark}
\begin{lem}\label{lem:AnalContToLocal} {\bf (Analytic continuation and symmetry imply $N$-point locality)}\\
Let $(V, W, \pi_{f}, Y)$ be a twisted vertex algebra of order $N$. Then for any $a, b, c\in V$ there exists $M$ such that
\begin{equation}
(z^N-w^N)^M \big( Y(a, z)Y(b, w) -(-1)^{p(a)p(b)}Y(b, w)Y(a, z)\big) \pi_{f} (c) =0,
\end{equation}
i.e., the fields $Y(a, z)$, \ $Y(b, w)$ are $N$-point local with points of \ locality \ $1, \epsilon, \dots, \epsilon ^{N-1}$, where $\epsilon$ is a primitive $N$-th root of unity.
\end{lem}
\begin{proof}
Since
\[
X_{z,w,0}(a\otimes b\otimes c)\in W[\!]z, w]\!]\ten \mathbf{F}_{\epsilon}(z, w) =W[\!]z, w ]\!][z^{-1}, w^{-1}, (z-w)^{-1},\dots, (z-\epsilon^{N-1}w)^{-1}],
\]
then exists $M$ such that
\[
(z^N-w^N)^M X_{z,w,0}(a\otimes b\otimes c)\in W[\!]z, w]\!][z^{-1}, w^{-1}],
\]
as can be obtained by multiplying by a common denominator. Then we can write
\[
\iota _{z, w}\big((z^N-w^N)^M X_{z,w,0}(a\otimes b\otimes c)\big) = \iota _{w, z}\big((z^N-w^N)^M X_{z,w,0}(a\otimes b\otimes c)\big),
\]
since in $W[\!]z, w]\!][z^{-1}, w^{-1}]$ the two expansions do not differ. But then
\begin{align*}
\iota _{z, w}(z^N-w^N)^M \iota _{z, w} X_{z,w,0}(a\otimes b\otimes c)& = \iota _{w, z}(z^N-w^N)^M \iota _{w, z} X_{z,w,0}(a\otimes b\otimes c)\\
& = \iota _{w, z}(z^N-w^N)^M \iota _{w, z} X_{w, z, 0}(\tau(a\ten b)\otimes c),
\end{align*}
where we used the symmetry in the last equality. Thus we have
\begin{equation*}
\iota _{z, w}(z^N-w^N)^M Y(a, z)Y(b, w) \pi_{f} (c) = \iota _{w, z}(z^N-w^N)^M (-1)^{p(a)p(b)}Y(b, w)Y(a, z) \pi_{f} (c)
\end{equation*}
\end{proof}
In fact, the converse lemma is also true:
\begin{lem} \label{lem:localityToanalcont}{\bf ($N$-point locality implies analytic continuation and symmetry)}\\
Let $a(z), \ b(w)$ be fields on a vector space $W$ which are $N$-point local with points of \ locality $1, \epsilon, \dots, \epsilon ^{N-1}$, where $\epsilon$ is a primitive $N$th root of unity. Then there  exists $X_{z, w, 0}: V\ten V\ten V \to W[\!]z, w]\!]\ten \mathbf{F}_{\epsilon}(z, w)$ such that
\begin{equation*}
a(z)b(w)\pi_f(c)=\iota _{z, w} X_{z, w,  0}(a\ten b\ten  c);
\end{equation*}
and  $X_{z, w, 0}(a\ten b\ten c)=X_{w, z, 0}(\tau (a\ten b)\ten c)$. Moreover, $X_{z, w, 0}(a\ten b\ten c)$ is the unique such element in $W[\!]z, w]\!]\ten \mathbf{F}_{\epsilon}(z, w)$ for any $a, b\in V$, \ $\pi_f(c)\in W$.
\end{lem}
\begin{proof}
Since the locality points are $1, \epsilon, \dots, \epsilon ^{N-1}$,  $\epsilon$ is a primitive $N$th root of unity, we have for the locality polynomial $\Pi _{z, w}=z^N-w^N$. $N$-point locality then implies
\begin{equation*}
\Pi _{z, w}^M \big( Y(a, z)Y(b, w) -(-1)^{p(a)p(b)}Y(b, w)Y(a, z)\big) \pi_f(c) =0.
\end{equation*}
Let
\begin{equation*}
A_{z,w}(a\otimes b\otimes c)=\iota_{z,w}\Pi _{z, w}^M Y(a,z)Y(b,w)\pi_f(c) = (-1)^{p(a)p(b)} \iota_{w,z}\Pi _{z, w}^MY(b,w)Y(a,z)\pi_f(c).
\end{equation*}
Thus
\[
A_{z,w}(a\otimes b\otimes c)\in W(\!(w)\!)(\!(z)\!)\cap W(\!(z)\!)(\!(w)\!),
\]
but we have
\[
W(\!(w)\!)(\!(z)\!)\cap W(\!(z)\!)(\!(w)\!) =W[\![z,w]\!][z^{-1},w^{-1}],
\]
so that
\[
A_{z,w}(a\otimes b\otimes c)\in W[\![z,w]\!][z^{-1},w^{-1}].
\]
Let then
\[
X_{z,w,0}(a\otimes b\otimes c)=\Pi_{z,w}^{-M}A_{z,w}(a\otimes b\otimes c)\in  W[\!]z, w ]\!][z^{-1}, w^{-1}, (z-w)^{-1},\dots, (z-\epsilon^{N-1}w)^{-1}],
\]
and
\[
W[\!]z, w ]\!][z^{-1}, w^{-1}, (z-w)^{-1},\dots, (z-\epsilon^{N-1}w)^{-1}]= W[\!]z, w]\!]\ten \mathbf{F}_{\epsilon}(z, w).
\]
Then
\[
\iota_{z,w}X_{z,w,0}(a\otimes b\otimes c)=\iota_{z,w}\Pi _{z, w}^{-M}A_{z,w}(a\otimes b\otimes c)
\]
as required.
Since we can also start from
\begin{equation*}
A_{w, z}(b\otimes a\otimes c)= \iota_{w ,z}\Pi _{z, w}^M Y(b, w)Y(a, z)\pi_f(c),
\end{equation*}
we see that
\[
A_{z,w}(a\otimes b\otimes c) = (-1)^{p(a)p(b)} A_{w, z}(b\otimes a\otimes c) =  A_{w, z}(\tau (a\otimes b) \otimes c),
\]
which symmetry transfers immediately to $X_{z, w, 0}(a\ten b\ten c)$.
\end{proof}
These two lemmas  imply that in the definition of twisted vertex algebra we can substitute the axioms of analytic continuation and symmetry with a single  axiom of $N$-point locality. The axiom requiring completeness with respect to OPE's is though much more easily expressible if we use the analytic continuation of fields (recall we have two different expressions for products of fields if we use the fields themselves and not the analytic continuation--see \lemref{lem:ResFormulasForProducts}):
\begin{lem} \label{lem:AnalContFormProd} Let $(V, W, \pi_{f}, Y)$ be a twisted vertex algebra of order $N$. Then for any $a, b, c\in V$. We have
\[
Res_{z=\epsilon ^j w}X_{z, w, 0}(a\ten b\ten c)(z-\epsilon ^i w)^k =Y(a, w)_{(j, k)}Y(b, w)\pi_f(c),
\]
where the products $Y(a, w)_{(j, k)}Y(b, w)$ were the $(i,k)$ products of the fields $Y(a, z)$ and $Y(b, w)$ we defined in \ref{defn:products-of-fields}.
\end{lem}
\begin{proof}
From the OPE of the fields $Y(a, z)$ and $Y(b, w)$ (which follows from locality), we get from \lemref{lem:OPE}:
\[
Y(a, z)Y(b, w)=i_{z, w} \sum_{j=1}^N\sum_{k=0}^{n_j-1}\frac{Y(a, w)_{(j, k)}Y(b, w)}{(z-\lambda_j w)^{k+1}} + :a(z)b(w):.
\]
From the analytic continuation property (specifically the uniqueness property of \ref{lem:localityToanalcont}) we have
\[
X_{z, w, 0}(a\ten b\ten c)= \sum_{j=1}^N\sum_{k=0}^{n_j-1}\frac{Y(a, w)_{(j, k)}Y(b, w)\pi_f(c)}{(z-\lambda_j w)^{k+1}} + :a(z)b(w):\pi_f(c)
\]
From \ref{lem:normalprodexpansion} we have
\[
i_{w, (z-\lambda_j w)}:a(\lambda w +(z-\lambda_j w))b( w): =\sum _{k\geq 0}\Big(:(\partial_{\lambda w} ^{(k)}a(\lambda w))b(w):\Big) (z-\lambda_j w)^k.
\]
Thus, the OPE part contains the residues $Res_{z=\epsilon ^j w}X_{z, w, 0}(a\ten b\ten c)(z-\epsilon ^i w)^k$ for \ $k\geq 0$, and  the other residues $Res_{z=\epsilon ^j w}X_{z, w, 0}(a\ten b\ten c)(z-\epsilon ^i w)^k$ for $k< 0$ are contained in the normal ordered product part.
Thus we clearly have
\[
Res_{z=\epsilon ^j w}X_{z, w, 0}(a\ten b\ten c)(z-\epsilon ^i w)^k =Y(a, w)_{(j, k)}Y(b, w)\pi_f(c),
\]
\end{proof}
\begin{defn}\label{defn:GenerFields} {\bf (Generating fields for a twisted vertex algebra of order $N$)}
We say that the  fields $\{a^0(z), a^1(z), \dots a^g(z), \dots\}$ are a set of generating fields for a twisted vertex algebra on a vector space of states $W$ (where $W$ contains the special vacuum vector $|0 \rangle$) if the following conditions are satisfied:
\begin{itemize}
\item The field $a^0(z)=Id_W$ (i.e., the identity operator on $W$ is always included,  by convention as $a^0(z)$);
\item Each field $a^i(z)$, indexed by $a^i (z) =\sum _{n\in \mathbb{Z}} a^i _{(k)}z^{-k-1}$, obeys the modified creation axiom: \ $a^i(z)|0 \rangle \arrowvert _{z=0}=a^i_{-1}|0 \rangle$; we denote $a^i: =a^i_{-1}|0 \rangle \in W$;
\item The fields $\{a^0(z), a^1(z), \dots a^g(z), \dots\}$ are self and mutually $N$-point local with points of locality \ $1, \epsilon, \dots, \epsilon ^{N-1}$, where $\epsilon$ is a primitive $N$th root of unity;
\item  The set of fields $\{a^0(z), a^1(z), \dots a^g(z), \dots\}$ is  {\bf closed under OPE's}, which requires  two conditions:
    \begin{enumerate}
    \item For any $i, j=1, \dots , g$
 \begin{equation}
 a^i (z) a^j (w) \sim  \sum_{m=1}^{N}\sum_{p=0}^{M}\frac{c^{ij}_{mp}(w)}{(z-\epsilon ^m w)^{p+1}},
 \end{equation}
 where the OPE coefficients are  {\bf uniformly shifted}. We say that the  OPE coefficients are  {\bf uniformly shifted} if for any $m, p$, exists $s_{p}\in \mathbb{Z}$,  such that $c^{ij}_{mp}(w) =w^{s_{p}} \sum_{l=0}^{g} C^{m, p}_{i, j, l} a^l(w)$, where $C^{m, p}_{i, j, l}\in \mathbb{C}$ are constants; but most importantly $s_{p}\in \mathbb{Z}$ is  independent of the choice of $i, j, m$ (the shift $s_{p}$ is the same for any choice of the fields $a^i (z), \  a^j (w)$ and any point of locality $\epsilon ^m$).
 \item  For any $i, j=1, \dots , g, \dots $
 the resulting OPE triple coefficients\footnote{Recall we defined triple OPE coefficients in \lemref{lem:tripleOPEcoef}.} $c^{i,jk}_{mp}(w)$ and $c^{ij,k}_{mp}(w)$ are also uniformly shifted and the shift $s_{p}$ is the same for  each  OPE coefficient $c^{ij}_{mp}(w)$,  $c^{i,jk}_{mp}(w)$ and $c^{ij,k}_{mp}(w)$  for any choice of the fields $a^i (z)$,  $a^j (z), a^k (z)$ and any point of locality $\epsilon ^m$).
 \end{enumerate}
 \item $W=span \{a^{i_1}_{(-n_1)}a^{i_2}_{(-n_2)}\dots a^{i_k}_{(-n_k)}|0 \rangle \ | i_1, i_2, \dots, i_k \in \{1, 2, \dots p\}; \ \ n_i\in \mathbb{N} \}$
\end{itemize}
\end{defn}
\begin{remark} For any two $N$-point local fields we have a Residue formula for the OPE coefficients,  \eqref{defn:products-of-fields}. In other words we can always express the OPE coefficients of two fields through the fields themselves, their derivatives, and {\bf various} shifts; but that formula doesn't allow us to control the shifts. For generating fields for a twisted vertex algebra we require something more: that the OPE coefficients of two generating fields are expressible only through the set of generating fields, that the shift is uniform, and that the uniformity of the shift persists (\lemref{lem:tripleOPEcoef} will ensure that). We will point out how that works specifically in the examples.
\end{remark}
\begin{remark} The indexing set for the set of the generating fields can be infinite, though the sum in the definition of uniformly shifted fields has to be finite.
\end{remark}
Now we set up what are necessary preliminary results needed to prove the ``Strong Generation Theorem" for a twisted vertex algebra. If we start with a set of generating fields   $\{a^0(z), a^1(z), \dots , a^g(z), \dots \}$ on a vector space $W$, we will take as the  super vector space $V$ for the twisted vertex algebra the collection  $\mathfrak{FD} \{a_0 (z), a_1(z), \dots a_g(z), \dots \}$ of the field descendants of the generating fields (see Definition \ref{defn:fielddesc}). We need to define the vertex map $Y$ and the projection map $\pi_f$ from the definition of a twisted vertex algebra. We start with the vertex map $Y$. First, to the generating fields, which by the definition above are indexed as required, we associate the same field as vertex operator (no shift is necessary):
\[
Y: a^i(z)\mapsto Y(a^i, z), \quad i=0, 1, \dots , g, \dots
\]
For their descendants, we proceed as follows:
\begin{align*}
Y&: \partial_z  a^i( z)\mapsto Y(Da^i, z):= \partial _z Y(a^i,  z), \quad i=0, 1, \dots , g, \dots;\\
Y&: a^i(\epsilon ^m z)\mapsto Y(T^{\epsilon^m}(a^i), z):= Y(a^i, \epsilon ^m z), \quad i=0, 1, \dots , g,  \dots\ \  m=0, 1, \dots, N-1.
\end{align*}
For the products of the generating fields: for the normal ordered products $a^i( w)_{(m, k)} a^j(w)$ ($i, j=0, 1, \dots , g,  \dots$,\ \  $m=1, \dots, N$, \ $k<0$), i.e., for the products with negative $k$, each normal ordered product  maps to itself via the vertex map $Y$:
 \[
 Y(a^i_{(m, k)} a^j,  w): =a^i( w)_{(m, k)} a^j(w), \ \ i, j=0, 1, \dots , g, \ \ m=1, \dots, N, \ \ k<0.
 \]
 (Here we really write $a^i_{(m, k)} a^j$ just as a shorthand notation for $a^i( w)_{(m, k)} a^j(w)$ with no further meaning). \\
 For the OPE coefficients, i.e., for the products   $a^i( w)_{(m, k)} a^j(w)$ with nonnegative $k$, we have
 \[
 Y(a^i_{(m, k)} a^j,  w): =\sum_{l=0}^{g} C^{m, k}_{i, j, l} Y(a^l, w) =\sum_{l=0}^{g} C^{m, k}_{i, j, l} a^l(w), \ \  m=1, \dots, N, \ \ k\geq 0.
 \]
 Here we used the fact that for the generating fields $\{a^0(z), a^1(z), \dots ,a^g(z), \dots \}$ the OPE coefficients are {\bf uniformly shifted}, see the definition \ref{defn:GenerFields} above. \\
 We can proceed recursively as above for all the newly generated descendants in \\ $\mathfrak{FD} \{a_0 (z), a_1(z), \dots a_g(z)\}$, as the uniformly-shifted property holds for all the descendants by \lemref{lem:tripleOPEcoef}.

Thus we can formulate an answer to the question that we started with in this paper:
\begin{thm} {\bf (Strong generation of a twisted vertex algebra)}\\
Let $\{a^0(z), a^1(z), \dots a^p(z)\}$ be a set of generating fields for a twisted vertex algebra on a vector space of states $W$. Then the data $(V, W, \pi_f, Y)$ forms a twisted vertex algebra, where $V=\mathfrak{FD} \{a_0 (z), a_1(z), \dots a_g(z)\}$, the map $Y$ is  defined as above, and the projection map $\pi_f$ is given by
\[
\pi_f (v):=Y(v, z)|0 \rangle \arrowvert _{z=0}, \quad \text{for \ any}\quad v\in V.
\]
\end{thm}
\begin{proof}
Since the space of fields $V$ is defined recursively by adding new descendants to $V=\mathfrak{FD} \{a_0 (z), a_1(z), \dots a_g(z)\}$ (see Definition \ref{defn:fielddesc}), we have to prove that in each recursive step the new descendant fields obey the axioms. By Dong's lemma all the new descendants are $N$-point local, from Lemma \ref{lem:localityToanalcont} it follows they satisfy the analytic continuation and the symmetry axioms. From Lemma \ref{lem:AnalContFormProd} we know that the required residues are either OPE coefficients, or the normal ordered products defined in \ref{lem:ResFormulasForProducts}. The only property we need to prove is that the uniformity of the shift persists for the descendants. The uniformity follows from \lemref{lem:tripleOPEcoef} and the definition of generating fields, as the OPE coefficients are finite linear combinations of the generating fields themselves.
\end{proof}
In some sense twisted vertex algebra is the unique object which is generated by a set of generating $N$-point local fields on a space of states $W$, but also differs in essential ways from an ordinary super vertex algebra:
\begin{lem} {\bf (Goddard uniqueness theorem for twisted vertex algebras)}\\
Let $(V, W, \pi_{f}, Y)$ be a twisted vertex algebra. Suppose $b(z)$ is a field on $W$ which is $N$-point self-local, and also mutually $N$-point local with all the fields $Y(a, z)$, $a\in V$,  from the twisted vertex algebra. Suppose also
$b(z)$ is indexed as $b(z)=\sum _{n\in \mathbb{Z}} b _{(k)}z^{-k-1}$ and in addition
\[
b(z) |0\rangle =e^{\epsilon ^m zD} b,
\]
where the vector $b\in W$ is the vector
\[
b: =b(z)|0 \rangle \arrowvert _{z=0} =b_{(-1)} |0\rangle.
\]
Then
\[
b(z)=Y(b, \epsilon ^m z).
\]
(Here the exponential $e^A$ is the usual exponential of an operator).
\end{lem}
There are a variety of different field theories, each of those theories is designed to describe different sets of examples of collections of fields. The best known is the theory of super vertex algebras (see e.g. \cite{FLM}, \cite{Kac}, \cite{LiLep}, \cite{BorcVA}, \cite{FZvi}). The axioms of super vertex algebras are often given in terms of locality (see \cite{Kac}, \cite{FZvi}), but in all cases locality is a property that plays crucial importance for super vertex algebras (\cite{LiLocality}). On the other hand, there are vertex algebra like objects which do not satisfy the usual locality property, but rather a generalization. Twisted vertex algebras are among them, but there are also generalized vertex algebras, $\Gamma$-vertex algebras, deformed chiral algebras, quantum vertex algebras. Unlike twisted vertex algebras, quantum vertex algebras do not satisfy the symmetry condition, and thus a locality-type property is very hard to write in that case. Thus the axioms for deformed chiral algebras (\cite{FR}) are given in terms of requiring existence of analytic continuations, and then the braided symmetry axiom is given in terms of the analytic continuations (\cite{FR}).
Twisted vertex algebras occupy intermediate step between super vertex algebras and deformed chiral algebras. We choose here to define twisted vertex algebras with axioms closer to the deformed chiral algebra axioms,  including the analytic continuation axiom. But twisted vertex algebras are in fact closer to super vertex algebras, in many ways. For one, they satisfy $N$-point locality (finitely many points of locality), unlike the deformed chiral algebras which have lattices of points of locality. Also, like in super vertex algebras, the axioms requiring existence of analytic continuation of product of two vertex operators plus the  symmetry axiom do in fact enforce the property that analytic continuation of arbitrary product of fields exist; something that is not true for deformed chiral algebras (see \cite{FR}). We finish with that property:
{\begin{prop} {\bf (Analytic continuation for arbitrary products of fields)}\\
Let $(V, W, \pi_{f}, Y)$ be a twisted vertex algebra.   For any $a_i \in V, i=1,\dots ,k$,  there exist a  rational vector valued   function
\begin{equation*}
\label{eq:defX}
X_{z_1, z_2, \dots , z_k} : V^{\ten k}
\to W[\!]z_1, z_2, \dots , z_k]\!]\ten \mathbf{F}^N_{\epsilon}(z_1, z_2, \dots , z_k)^{+, z_k},
\end{equation*}
    such that
\begin{equation*}
Y(a_1, z_1)Y(a_2, z_2)\dots Y (a_k, z_k)1=i_{z_1,z_2,\dots,z_k} X_{z_1, z_2, \dots , z_k}(a_1\ten a_2\ten \dots \ten a_k)
\end{equation*}
\end{prop}}
\begin{proof} The proof is by induction on the number of products $k$. For $k=2$ we have the axiom of analytic continuation: for any $a_1, a_2, c\in V$ exists \\ $X_{z_1, z_2, 0}(a_1\ten a_2\ten c)\in W[\!]z_1, z_2]\!]\ten \mathbf{F}_{\epsilon}(z_1, z_2)$ such that
\begin{equation*}
Y(a_1, z_1)Y(a_2, z_2)\pi_{f}(c)=i_{z_1, z_2} X_{z_1, z_2,  0}(a_1\ten a_2\ten  c)
\end{equation*}
If we take $c=1=|0 \rangle$ we have the desired
\[
X_{z_1, z_2}(a_1\ten a_2):=X_{z, w, 0}(a_1\ten a_2\ten 1)\in W[\!]z_1, z_2]\!]\ten \mathbf{F}_{\epsilon}(z_1, z_2).
\]
Suppose the analytic continuation property holds for a product of $k-1$ fields, and consider
 $Y(a_1, z_1)Y(a_2, z_2)\dots Y (a_k, z_k)1$. From Lemma \ref{lem:AnalContToLocal} we know that the fields $Y(a_1, z_1),$ \ $Y(a_2, z_2), \dots , \ Y (a_k, z_k)$ are $N$-point mutually local. Thus there exists a high enough power $M$ so that if we multiply by $\Pi_{z_1, z_2, \dots , z_k}^M$, where
 \[
 \Pi_{z_1, z_2, \dots , z_k}=\prod_{i<j}^k(z_i^N-z_j^N)
 \]
 we can interchange any two of the fields $Y(a_1, z_1)$, \ $Y(a_2, z_2), \dots , \ Y (a_k, z_k)$ in the product.
 Also, by the induction hypothesis, we know there exists $X_{z_2, \dots , z_k}(a_2\ten \dots \ten a_k)\in W[\!]z_2, \dots , z_k]\!]\ten \mathbf{F}^N_{\epsilon}(z_2, \dots , z_k)^{+, z_k}$ such that
\begin{equation*}
Y(a_1, z_1)Y(a_2, z_2)\dots Y (a_k, z_k)1=i_{z_2,\dots,z_k} Y(a_1, z_1) X_{z_2, \dots , z_k}(a_2\ten \dots \ten a_k).
\end{equation*}
Let $F_{z_2, \dots, z_k}=W[\!]z_2, \dots , z_k]\!]\ten \mathbf{F}^N_{\epsilon}(z_2, \dots , z_k)^{+, z_k}$.
We see that
\[
Y(a_1, z_1) X_{z_2, \dots , z_k}(a_2\ten \dots \ten a_k)\in F_{z_2, \dots, z_k}(\!(z_1)\!).
\]
But also
\begin{align*}
\Pi_{z_1, z_2, \dots , z_k}^M &Y(a_1, z_1)Y(a_2, z_2)\dots Y (a_k, z_k)1 \\
&=\Pi_{z_1, z_2, \dots , z_k}^M  Y(a_2, z_2)\dots  Y (a_{k-1}, z_{k-1}) Y (a_k, z_k)Y(a_1, z_1)1.
\end{align*}
But from the axioms of vertex algebra we know that $Y(a_1, z_1)1$ is regular in $z_1$ (has no negative powers of $z_1$). Thus we see that the last line has no pole in $z_1$. But that is also true then for
\[
\Pi_{z_1, z_2, \dots , z_k}^M Y(a_1, z_1) X_{z_2, \dots , z_k}(a_2\ten \dots \ten a_k)\in F_{z_2, \dots, z_k}(\!(z_1)\!).
\]
Which forces
\[
\Pi_{z_1, z_2, \dots , z_k}^M Y(a_1, z_1) X_{z_2, \dots , z_k}(a_2\ten \dots \ten a_k)\in F_{z_2, \dots, z_k}[\![z_1]\!];
\]
and thus the required expansion
\[
Y(a_1, z_1)Y(a_2, z_2)\dots Y (a_k, z_k)1=i_{z_1,z_2,\dots,z_k} X_{z_1, z_2, \dots , z_k}(a_1\ten a_2\ten \dots \ten a_k)
\]
exists for an element
\[
X_{z_1, z_2, \dots , z_k}(a_1\ten a_2\ten \dots \ten a_k)\in W[\!]z_1, z_2, \dots , z_k]\!]\ten \mathbf{F}^N_{\epsilon}(z_1, z_2, \dots , z_k)^{+, z_k}.
\]
\end{proof}

\section{Appendix: Comparison of Twisted vertex algebras and $\Gamma$-vertex algebras.}
\label{section:appendix}

In \secref{section:NormalOrdProd} we defined products of fields $a(w)_{(j, k)}b(w)$ for any $j=1, 2, \dots, N$ and $k\in \mathbb{Z}$. For $k\geq 0$ these products were the coefficients of the OPEs, and for $k<0$ they were the normal ordered products.

As we mentioned in \secref{section:NormalOrdProd}, in the paper  \cite{Li2} Li also defines products of fields for  ``compatible" pairs of vertex operators (see Definition 3.4 in \cite{Li2}). We provide examples to demonstrate that the products of fields differ in our approach.

In  Example \ref{example:B}, we have
$\phi^B(w)_{(2, 0)}\phi^B (w)=-2w$, $\phi^B(w)_{(1, -1)}\phi^B (w)=1$.
Let us calculate the corresponding products according to Definition 3.4 in \cite{Li2}. We will denote the products in \cite{Li2} by $\phi^B(w)_{\overline{(\alpha, k)}}\phi^B (w)$. According to \cite{Li2} the pair of fields $(\phi^B(x_1), \phi^B (x_2))$ is compatible  by the polynomial $f(x_1, x_2)=x_1 +x_2$, since we have from the OPE  of the field $\phi^B(z)$ with itself
\[
(x_1 +x_2)\phi^B (x_1)\phi^B(x_2) =-2x_2 +(x_1 +x_2):\phi^B (x_1)\phi^B(x_2):,
\]
and as always $:\phi^B (x_1)\phi^B(x_2):$ is an element of  $Hom(\mathit{F_B},  \mathit{F_B}((x_1, x_2)))$.
Thus according to Definition 3.4 in \cite{Li2}, if we want to calculate the products at $\lambda =1$ ($\alpha =1$ in \cite{Li2}), we calculate
\begin{align*}
((x_1 +x)\phi^B (x_1)\phi^B(x))\arrowvert _{x_{1}=x+x_{0}} &=-2x +(2x +x_0)(:\phi^B (x+x_0)\phi^B(x):)  \\
& =-2x +(2x +x_0)\sum _{k\geq 0}x_0^k :(\partial_x^{(k)}\phi^B (x))\phi^B (x):.
\end{align*}
Next,
\begin{align*}
&i_{x, x_0}(x_0+x+x)^{-1}((x_1 +x)\phi^B (x_1)\phi^B(x))\arrowvert _{x_{1} =x+x_{0}}  \\
& \hspace{1cm} =-2x\Big(\sum_{l\geq 0}\frac{(-1)^lx_0^l}{(2x)^{l+1}}\Big) +\sum _{k\geq 0}x_0^k :(\partial_x^{(k)}\phi^B (x))\phi^B (x):\Big)  \\
& \hspace{1cm} =\big(-1 +:\phi^B (x)\phi^B (x):\big) +x_0 \Big(\frac{1}{2x} +:(\partial_x\phi^B (x))\phi^B (x): \Big) +O(x_0^2).
\end{align*}
Thus, as we commented in \secref{section:NormalOrdProd}, since we have a single pole in the OPE (at $z=-w$) here our positive (the OPE) products coincide with the \cite{Li2} in that we have
\[
\phi^B(x)_{\overline{(1, k)}}\phi^B (x) =\phi^B(x)_{(1, k)}\phi^B (x) =0 \quad \text{for} \quad k\geq 0.
\]
On the other hand, our  negative (normal ordered products) differ from \cite{Li2}, for instance
\begin{align*}
& \phi^B(x)_{(1, -1)}\phi^B (x)= :\phi^B (x)\phi^B (x): =1 \quad \text{vs}  \quad \phi^B(x)_{\overline{(1, -1)}}\phi^B (x) =-1 +:\phi^B (x)\phi^B (x):=0;\\
& \phi^B(x)_{(1, -2)}\phi^B (x)= :(\partial _x \phi^B (x))\phi^B (x):\quad \text{vs} \quad \phi^B(x)_{\overline{(1, -2)}}\phi^B (x) =\frac{1}{2x} +:(\partial_x\phi^B (x))\phi^B (x):
\end{align*}
Similarly, for $\lambda =-1$ ($\alpha =-1$ in \cite{Li2}), we calculate
\begin{align*}
((x_1 +x)\phi^B (x_1)\phi^B(x))\arrowvert _{x_{1}=-x+x_{0}} &=-2x +(x_1 +x)(:\phi^B (-x+x_0)\phi^B(x):)  \\
& = -2x +x_0\sum _{k\geq 0}x_0^k :(\partial_{-x}^{(k)}\phi^B (-x))\phi^B (x):.
\end{align*}
Next we have
\begin{align*}
&i_{x, x_0}(x_0-x+x)^{-1}((x_1 +x)\phi^B (x_1)\phi^B(x))\arrowvert _{x_{1} =-x+x_{0}}  \\
& \hspace{2cm} =\frac{1}{x_0}\Big(-2x +x_0\sum _{k\geq 0}x_0^k :(\partial_{-x}^{(k)}\phi^B (-x))\phi^B (x):\Big)  \\
& \hspace{2cm} =(-2x)x_0^{-1} +(:\phi^B (-x)\phi^B (x):) +\big(:(\partial _{-x} \phi^B (-x))\phi^B (x):)x_0 +O(x_0^2).
\end{align*}
Since we have a single pole in the OPE (and precisely at $\alpha =-1$) in this example our  products coincide with \cite{Li2}, in that we have
\[
\phi^B(x)_{\overline{(-1, k)}}\phi^B (x) =\phi^B(x)_{(2, k)}\phi^B (x) =0 \quad \text{for} \quad k\in \mathbb{Z}.
\]
We want to finish with a calculation of an example of products of \cite{Li2} in the case of multiple poles in the OPE, and we will take as an example the OPE from  \eqnref{eqn:HeisOPEsD}.
The pair of fields $(h^D(x_1), h^D(x_2))$ is compatible  by the polynomial $f(x_1, x_2)=(x_1^2 -x_2^2)^2$, since we have from the OPE of the field $h^D(z)$ with itself
\[
(x_1^2 -x_2^2)^2 h^D(x_1) h^D(x_2) =x_1 x_2 +(x_1^2 -x_2^2)^2 :h^D(x_1) h^D(x_2):,
\]
and as always $:h^D (x_1) h^D(x_2):$ is an element of  $Hom(\mathit{F_D},  \mathit{F_D}((x_1, x_2)))$.
Thus according to Definition 3.4 in \cite{Li2}, if we want to calculate the products at $\lambda =1$ ($\alpha =1$ in \cite{Li2}), we calculate
\begin{align*}
((x_1^2 -x^2)^2 h^D(x_1) h^D(x))&\arrowvert _{x_{1}=x+x_{0}} =(x+x_0)x +(x_1^2 -x^2)^2(:h^D(x+x_0)h^D(x):) \\
& =(x+x_0)x +(2xx_0 +x_0^2)^2\sum _{k\geq 0}x_0^k :(\partial_x^{(k)}h^D(x))h^D(x):.
\end{align*}
Next,
\begin{align*}
&i_{x, x_0}(2xx_0 +x_0^2)^{-2} ((x_1^2 -x^2)^2h^D(x_1) h^D(x))\arrowvert _{x_{1}=x+x_{0}} \\
& \hspace{1cm} =(x+x_0)x \cdot i_{x, x_0}(2xx_0 +x_0^2)^{-2} + \sum _{k\geq 0}x_0^k :(\partial_x^{(k)}h^D(x))h^D(x): \\
& \hspace{1cm} =\frac{(x+x_0)x}{x_0^2}\Big(\sum_{l\geq 0}(l+1)\frac{(-1)^lx_0^l}{(2x)^{l+2}}\Big) + \sum _{k\geq 0}x_0^k :(\partial_x^{(k)}h^D(x))h^D(x): \\
& \hspace{1cm}  = \frac{1}{4}x_0^{-2}+ (\frac{1}{8x}x_0^{-1} +\big(-\frac{1}{16x^2} +:h^D(x)h^D(x):\big) +O(x_0).
\end{align*}
Thus for the first three nontrivial products ($h^D(x)_{\overline{(1, k)}}h^D(x)=0$ for $k\geq 2$) we have
\begin{align*}
&h^D(x)_{\overline{(1, 1)}}h^D(x) =  \frac{1}{4} = h^D(x)_{(1, 1)}h^D(x), \\
& h^D(x)_{\overline{(1, 0)}}h^D(x) = \frac{1}{8x} \quad  \text{vs} \quad h^D(x)_{(1, 0)}h^D(x) =0\\
& h^D(x)_{\overline{(1, -1)}}h^D(x) = -\frac{1}{16x^2} +:h^D(x)h^D(x): =-\frac{1}{16x^2} +h^D(x)_{(1, -1)}h^D(x).
\end{align*}
 Thus products in \cite{Li2} differ from ours, except at the highest order of the pole, where $h^D(x)_{\overline{(1, 1)}}h^D(x) = h^D(x)_{(1, 1)}h^D(x)= \frac{1}{4}$.

As above, we can relate the products in \cite{Li2} to ours, and vice versa, however the algebraic structures incorporating these fields are not equivalent. In \cite{Li2} the author defines the notion of a $\Gamma$-vertex algebra, which is an algebraic structure incorporating certain collections of  multi-local fields.  This structure however does not  incorporate the examples we want to consider, namely the boson-fermion correspondences of type B, C and D-A.
To see this consider the fields $\phi ^B(z)$, $\phi ^B(-z)$, $:\phi ^B(z)\phi ^B(-z):$ (and their descendants) of \secref{section: examples}.
The field  $h^B(z)= \frac{1}{4}(:\phi ^B(z)\phi ^B(-z): -1)$, which has only  odd-indexed modes,   $h^B(z)=\sum _{n\in \mathbb{Z}} h_{2n+1} z^{-2n-1}$, is not a vertex operator in a $\Gamma$-vertex algebra for the following simple reason:
$\Gamma$-vertex algebras are comprised of vertex operators $Y_{\alpha } (v, z)$, for $v$ an element of a given vector space $V$, and $\alpha$ an element of a subgroup of $\mathbb{C}$, satisfying certain properties such as:
\[
[D, Y_{\alpha } (v, z)]=\alpha \frac{d}{dz} Y_{\alpha } (v, z),
\]
where $D$ is a nontrivial linear operator $D: V\to V$ (see \cite{Li2}, Lemma 6.5).
Now if $h^B(z)$ is a vertex operator in a $\Gamma$-vertex algebra, there should exists $\alpha$ and $v\in V$ such that $h^B(z)=Y_\alpha(v,z)$.
The field $h(z)$ has only odd powers of $z$, thus
$[D, Y_{\alpha } (v, z)]$ has only {\bf odd} powers of $z$ (as $D$ only acts on elements of $V$), but the $\alpha \frac{d}{dz} Y_{\alpha } (v, z)$ have {\bf even} powers of $z$ for all $\alpha$.
Thus there is no $v\in V$ and $\alpha\in \Gamma$ such that $h^B(z)=Y_\alpha(v,z)$.

The boson-fermion correspondence of type B, C, and D-A require the presence of even and odd fields. This motivates the definition of twisted vertex algebras. In general, for a finite cyclic group $\Gamma$,  $\Gamma$-vertex algebras are described by  smaller collections of descendant fields.  The examples of $\Gamma$-vertex algebras in \cite{Li2} can be incorporated as   subsets of  the corresponding spaces of  fields  of  the related twisted vertex algebras.

\def\cprime{$'$}

\end{document}